\documentclass[a4paper,onecolumn,11pt,unpublished]{quantumarticle}
\pdfoutput=1

\usepackage[utf8]{inputenc}
\usepackage{amsmath,amsthm,mathrsfs}
\usepackage{amssymb}
\usepackage{bbold}
\usepackage{xparse}
\usepackage{physics}
\usepackage{hyperref}
\usepackage{mathtools}
\usepackage{tikz}
\usepackage{tikz-cd}
\usepackage{tikzit}
\input{styles.tikzdefs}

\tikzstyle{white}=[fill=white, draw=black, shape=circle]
\tikzstyle{thickwhite}=[fill=white,very thick, draw=black, shape=circle]
\tikzstyle{black}=[fill=black, draw=black, shape=circle]
\tikzstyle{sqr}=[fill=white, draw=black, shape=rectangle]
\tikzstyle{wgrey}=[fill=white, draw={rgb,255: red,191; green,191; blue,191}, shape=circle]
\tikzstyle{bgrey}=[fill={rgb,255: red,191; green,191; blue,191}, draw={rgb,255: red,191; green,191; blue,191}, shape=circle]

\tikzstyle{discarding}=[fill=white, draw=black, shape=circle, style=upground, scale = 1.25]
\tikzstyle{smalldiscarding}=[fill=white, draw=black, style=upground, scale=0.75]
\tikzstyle{backdiscard}=[fill=white, draw=black, shape=circle, style=downground, scale=0.5]
\tikzstyle{smallbackdiscard}=[fill=white, draw=black, shape=circle, style=downground, scale=0.5]
\tikzstyle{state}=[fill=white, draw=black, style=triang, tikzit shape=rectangle]
\tikzstyle{kstate}=[fill=white, draw=black, style=kpoint, tikzit shape=rectangle]
\tikzstyle{kstateconj}=[fill=white, draw=black, style=kpoint conjugate, tikzit shape=rectangle]
\tikzstyle{kstateBIG}=[fill=white, draw=black, style=big kpoint, tikzit shape=rectangle]
\tikzstyle{effect}=[fill=white, draw=black, style=triangdag]
\tikzstyle{keffect}=[fill=white, draw=black, style=kpoint adjoint]
\tikzstyle{keffectconj}=[fill=white, draw=black, style=kpoint transpose]
\tikzstyle{morphdag}=[style=mapdag]
\tikzstyle{morph}=[style=hadamard]
\tikzstyle{WIDEmorph}=[style=hadamard, minimum width=14mm]
\tikzstyle{morphtrans}=[style=maptrans]
\tikzstyle{morphconj}=[style=mapconj]
\tikzstyle{CPMmorph}=[style=dmap]
\tikzstyle{CPMmorphconj}=[style=dmapconj]
\tikzstyle{CPMmorphdag}=[style=dmapdag]
\tikzstyle{CPMmorphtrans}=[style=dmaptrans]
\tikzstyle{CPMstate}=[fill=white, draw=black, style=triang, doubled]
\tikzstyle{CPMstateBIG}=[fill=white, draw=black, style={triang_lesssep}, doubled]
\tikzstyle{CPMkstate}=[fill=white, draw=black, style=kpoint, tikzit shape=rectangle, doubled]
\tikzstyle{CPMkstateconj}=[fill=white, draw=black, style=kpoint conjugate, tikzit shape=rectangle, doubled]
\tikzstyle{CPMkstateBIG}=[fill=white, draw=black, style=big kpoint, tikzit shape=rectangle, doubled]
\tikzstyle{CPMkeffect}=[fill=white, draw=black, style=kpoint adjoint, doubled]
\tikzstyle{CPMkeffectconj}=[fill=white, draw=black, style=kpoint transpose, doubled]
\tikzstyle{UHfB}=[fill=white, draw=black, style=triangdag, doubled, inner sep=-2pt]
\tikzstyle{leak}=[style=tinypoint, regular polygon rotate=-90]
\tikzstyle{leakfill}=[style=tinypoint, regular polygon rotate=-90, fill=black]
\tikzstyle{Z}=[style=dot, fill=green]
\tikzstyle{X}=[style=dot, fill=red]
\tikzstyle{black_dot}=[style=dot, fill=black]
\tikzstyle{white_dot}=[style=dot, fill=white]
\tikzstyle{qblack_dot}=[style=ddot, fill=black]
\tikzstyle{qwhite_dot}=[style=ddot, fill=white]
\tikzstyle{whitephase}=[style=wphase dot, fill=white]
\tikzstyle{qredphase}=[style=phase dot, fill=red]
\tikzstyle{qgreenphase}=[style=phase dot, fill=green]
\tikzstyle{had}=[style=hadamard, doubled]
\tikzstyle{box}=[style=hadamard]
\tikzstyle{bigbox}=[style=hadamard, minimum height=4mm, minimum width=8mm]
\tikzstyle{classhad}=[style=hadamard]
\tikzstyle{antipode}=[style=anti]

\tikzstyle{dotted zigzag}=[-, dotted, decorate, decoration={zigzag}]
\tikzstyle{zigzag edge}=[-, decorate, decoration={zigzag}]
\tikzstyle{zigzag custom}=[-, decorate, decoration={zigzag, segment length=4pt, amplitude=2pt}]
\tikzstyle{dottededge}=[-, dash pattern=on 1pt off 0.7pt]
\tikzstyle{double edge}=[-, style=doubled, draw=black, tikzit draw={rgb,255: red,18; green,168; blue,191}]
\tikzstyle{arrow}=[->]
\tikzstyle{new edge style 1}=[-, draw={rgb,255: red,242; green,233; blue,206}, fill={rgb,255: red,242; green,233; blue,206}]
\tikzstyle{morphism_shade}=[-, draw=black, fill={rgb,255: red,242; green,233; blue,206}, line join=bevel]
\tikzstyle{supermap_shade}=[-, fill={rgb,255: red,216; green,215; blue,242}, draw=black, line join=bevel]
\tikzstyle{hole_shade}=[-, fill=white, draw=black, line join=bevel]
\tikzstyle{new edge style 2}=[-, draw={rgb,255: red,14; green,188; blue,83}]
\tikzstyle{new edge style 3}=[<-, draw={rgb,255: red,234; green,209; blue,255}]
\tikzstyle{new edge style 4}=[<-, draw={rgb,255: red,0; green,106; blue,106}]
\tikzstyle{new edge style 5}=[-, draw={rgb,255: red,214; green,110; blue,62}]
\tikzstyle{new edge style 6}=[-, draw={rgb,255: red,174; green,20; blue,174}]
\tikzstyle{new edge style 0}=[-, fill=none, draw={rgb,255: red,0; green,106; blue,106}]

\usepackage{cite}

\usepackage{graphicx}   
\usepackage{subcaption} 

\usepackage{macrosOctave}

\begin{document}

\title{Non-factor quantum dynamics and relativity}

\author{Octave Mestoudjian}
\affiliation{Université Paris-Saclay, Inria, CNRS, LMF, 91190 Gif-sur-Yvette, France}

\author{Matt Wilson}
\affiliation{Université Paris-Saclay, CentraleSupélec, Inria, CNRS, LMF, 91190 Gif-sur-Yvette, France}

\maketitle

\begin{abstract}


We consider the effect of relativistic principles on the structure of the local evolutions of non-factor quantum systems such as those which naturally arise in quantum spacetimes and causal structures.
Generalising a former result that had been given for factor systems, we show that any evolution which can be applied locally on a direct sum of factor systems and without producing superluminal signals must be linear, unitary, and respect the block decomposition of the system.
\end{abstract}


\section{Introduction}

In an attempt to understand the particular features of quantum gravity, one can wonder how the compatibility with relativity should affect the possible evolutions of a quantum system. This question of rederiving or modifying the properties of quantum evolutions has a long history that spans in particular spontaneous collapse theories attempting to avoid a measurement problem \cite{GRW, PearleReductionAndLocalization, GPR, DIOSI1987377, DiosiUniversalReduction, PenroseGravity’sRole, PenroseGravitizationQM, GGPRelativisticDynamicalReduction, RelativisticGRW} as well as theories of semi-classical gravity \cite{WEINBERG1989336, DIOSI1984199, EDvsID, WeinbergPrecision, Carlip_2008, NonlinearEvolutionandSignaling, KentTestinQG, RayNoSignaling, Stamp_2015, Barvinsky_Structure, Barvinsky_Correlated, Wilson-Gerow_Propagators, OppenheimPostquantum} (where the matter and energy fields are quantum but the gravitational field is classical). Part of that history is also Gisin's influential article \cite{Gisin_1989} in which he argues that Schrödinger dynamics are the only evolution sending pure states to pure states compatible with the relativistic ``no superluminal signalling" principle. 

There are a variety of responses that can be made to the argument of Gisin. In \cite{Kent_2005} an alternative way to model the progression of the values of states of entangled subsystems is put forward, in which evolutions can be non-linear at the cost of abandoning a consistent global description for the evolution of the quantum state. 
A recent work \cite{wilson2023originlinearityunitarityquantum} extends arguments against non-linear evolution to the single-shot (discrete time) setting, proving that any local transformation which induces a consistent global state evolution without superluminal signalling must be linear and unitary. 




All these approaches limit the subsystems of consideration to those which factorise into tensor products, and yet, subsystems modelled by non-factor algebras arise very naturally in foundations of physics, in particular at the very intersection between relativity theory and quantum theory. Indeed, within quantum gravity, considering that the geometry of space-time could itself be in a superposition makes some very simple systems non-factor \cite{Arrighi2024quantumnetworks, Bianchi2023}. Similarly, the causal structure of some quantum evolutions can only be derived from their compositional and local properties by mixing together the tensor product and direct sum \cite{lorenz_unitary, vanrietvelde2023consistentcircuitsindefinitecausal, Barrett2021} which is equivalent to considering partitions into non-factor subsystems \cite{vanrietvelde2025partitionsquantumtheory, vanrietvelde2025causaldecompositions1dquantum}. This raises the question of what are the constraints imposed by relativity on the evolutions of these non-factor subsystems.

Recent tools \cite{mestoudjian2026picturinggeneralquantumsubsystems} were developed for treating finite-dimensional non-factor systems on an equal footing with the factor case and allowing graphical reasoning techniques from categorical quantum mechanics \cite{coecke_kissinger_2017, heunen_categories, abramsky_coecke} to be applied to non-factor systems. In this article, we use these tools to derive the form of dynamics compatible with relativity on non-factor subsystems. We find that these evolutions should act linearly and untarily on states but furthermore in the non-factor case they must block decompose, thereby deriving routed quantum evolutions \cite{vanrietvelde2023consistentcircuitsindefinitecausal} as the most general possible dynamics of non-factor subsystems compatible with relativity. From a purely algebraic point of view we obtain that the evolutions of non-factor systems compatible with relativity are inner *-automorphisms of the corresponding finite-dimensional von Neumann algebra. Finally we tease apart the independent consequences of relativistic compatibility, demonstrating that imposing locality without superluminal signalling leads to a restricted class of non-linear quantum dynamics in which deterministic post-selection can be arbitrarily-well approximated.


\section{The result in the factor case}\label{resultfactorcase}

In this section, we remind the reader about the context of the result of \cite{wilson2023originlinearityunitarityquantum}. We place ourselves in the usual formalism of factor systems where they are represented by Hilbert spaces and are composed through the tensor product, meaning that given two systems $A$ and $B$ associated to Hilbert spaces $\ch_A$ and $\ch_B$ respectively, $AB$ will correspond to the Hilbert space $\ch_A \otimes \ch_B$. Given a system $A$, we will write $S_A = \{ \ket{\psi} \in \ch_A, \Vert \ket{\psi} \Vert = 1\}$ the set of normalised states of the system $A$.

Let us now fix a system $A$. 
A locally applicable evolution of the system $A$ is the data of a family of functions $(\ml_X : S_{AX} \rightarrow S_{AX})_{X}$, indexed by the possible environments $X$ to the system $A$, sending the (normalised) states of $AX$ to (normalised) states of $AX$, and that satisfy the three following axioms:

\begin{itemize}
\item Local applicability. Given any two systems $X,X'$ and states $\ket{\psi} \in S_{AX}$ and $\ket{\phi} \in S_{X'}$,
\be
\ml_{XX'}(\ket{\psi}\otimes \ket{\phi}) = \ml_X(\ket{\psi})\otimes \ket{\phi} \, ;
\ee
we consider that it is possible to have an empty environment and thus that there exists a function $\ml_{[-]} : S_A \rightarrow S_A$ such that (choosing $X$ to be this empty environment in the previous equation),
\be
\ml_{X'}(\ket{\psi}\otimes \ket{\phi}) = \ml_{[-]}(\ket{\psi})\otimes \ket{\phi} \, ;
\ee
\item No superluminal signalling. Given any system $X$, any state $\ket{\psi} \in S_{AX}$ and any measurement result $m$ on $X$(by which we mean a projector of $\cl(\ch_X)$),
\be
\ml_X(\ket{\psi})^{\dagger}(\id \otimes m)\ml_X(\ket{\psi}) = \bra{\psi} (\id \otimes m) \ket{\psi} \, ;
\ee
\item Update commutativity. Given any system $X$, any state $\ket{\psi} \in S_{AX}$ and any measurement result $m$ on $X$,
\be
\frac{(\id \otimes \pi) \ml_X(\ket{\psi})}{\Vert (\id \otimes \pi) \ml_X(\ket{\psi}) \Vert} = \ml_X\left( \frac{(\id \otimes \pi)\ket{\psi}}{\Vert \id \otimes \pi)\ket{\psi} \Vert} \right) \, .
\ee
\end{itemize}

The main result of \cite{wilson2023originlinearityunitarityquantum} is to show that, for every locally-applicable transformation on pure quantum states there exists a unitary $U$ of $\cl(\ch_A)$ such that, for all $X$, $\ml_X = U \otimes \id_{\ch_X}$. In the rest of this article, we aim at generalising this result to the non-factor case. To do so, we need to be able to express a similar definition of locally applicable evolution of a non-factor system. We start by describing a new point of view on our theory of splitting maps, that generalises the compositional view and the tensor product of Hilbert spaces.

\section{Splitting maps and generalised tensors}

The mathematical proofs of this article are based on the formalism of splitting maps that was developed in \cite{mestoudjian2026picturinggeneralquantumsubsystems}. This formalism, that can be expressed in a decompositional way, when referring to splitting maps, but also in a more compositional way when referring to generalised tensors, allows one to see general quantum (sub)systems, i.e. von Neumann (sub)algebras, as local with respect to a decomposition of the underlying Hilbert space in the same way that it's done for usual quantum systems described by factor von Neumann algebras.

In the following we formalise the following intuitive picture of identifying subsystems by the legitimate ways of splitting global systems into parts:
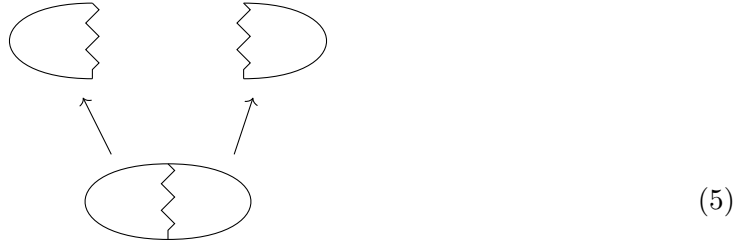
\be
\begin{tikzpicture}
	\begin{pgfonlayer}{nodelayer}
		\node [style=none] (0) at (0, 1) {};
		\node [style=none] (1) at (0, -1) {};
		\node [style=none] (2) at (-1.5, 1.25) {};
		\node [style=none] (3) at (-2.25, 2.75) {};
		\node [style=none] (4) at (1.75, 1.25) {};
		\node [style=none] (5) at (2.25, 2.75) {};
		\node [style=none] (6) at (-2, 5.25) {};
		\node [style=none] (7) at (-2, 3.25) {};
		\node [style=none] (8) at (2, 5.25) {};
		\node [style=none] (9) at (2, 3.25) {};
	\end{pgfonlayer}
	\begin{pgfonlayer}{edgelayer}
		\draw [bend right=90, looseness=3.75] (0.center) to (1.center);
		\draw [bend left=90, looseness=3.75] (0.center) to (1.center);
		\draw [style=zigzag edge] (0.center) to (1.center);
		\draw [style=arrow] (2.center) to (3.center);
		\draw [style=arrow] (4.center) to (5.center);
		\draw [bend right=90, looseness=3.75] (6.center) to (7.center);
		\draw [style=zigzag edge] (6.center) to (7.center);
		\draw [bend left=90, looseness=3.75] (8.center) to (9.center);
		\draw [style=zigzag edge] (8.center) to (9.center);
	\end{pgfonlayer}
\end{tikzpicture}
\ee
When formalised in linear algebra, we refer to such splits as splitting maps.  

\begin{definition}[Splitting map and generalised tensor]
A {\em splitting map} $\chi$ on $\ch$ is an isometry $\chi : \ch \rightarrow \chlc \otimes \chrc$ where $\chlc$ and $\chrc$ are Hilbert spaces. A {\em generalised tensor} is the adjoint of a splitting map, i.e. a coisometry $\chi^{\dagger} : \chlc \otimes \chrc \rightarrow \ch$. We will write a splitting map $\chi$ diagrammatically as 
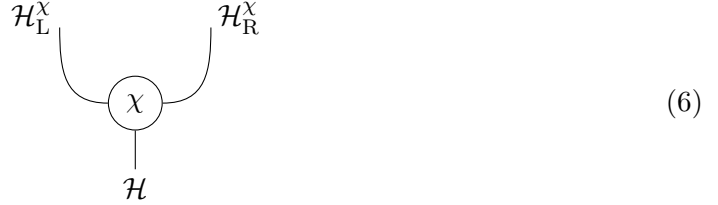
\begin{equation}
  \begin{tikzpicture}
	\begin{pgfonlayer}{nodelayer}
		\node [style=white] (0) at (0, 0) {$\chi$};
		\node [style=none] (1) at (-2, 2) {};
		\node [style=none] (2) at (2, 2) {};
		\node [style=none] (3) at (0, -1.75) {};
		\node [style=none] (4) at (-2.75, 2.25) {$\chlc$};
		\node [style=none] (5) at (2.75, 2.25) {$\chrc$};
		\node [style=none] (6) at (0, -2.25) {$\ch$};
	\end{pgfonlayer}
	\begin{pgfonlayer}{edgelayer}
		\draw [in=180, out=-90, looseness=1.25] (1.center) to (0);
		\draw [in=0, out=-90, looseness=1.25] (2.center) to (0);
		\draw (0) to (3.center);
	\end{pgfonlayer}
\end{tikzpicture}
\end{equation}
and we will write the associated generalised tensor as 
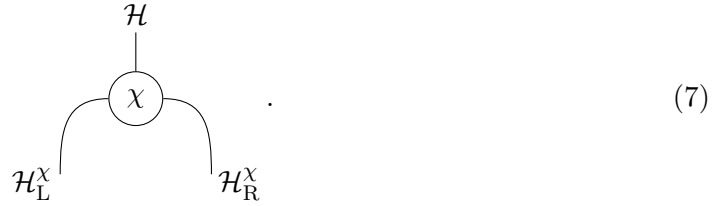
\begin{equation}
 \begin{tikzpicture}
	\begin{pgfonlayer}{nodelayer}
		\node [style=white] (0) at (3.25, 0) {$\chi$};
		\node [style=none] (1) at (1.25, -2) {};
		\node [style=none] (2) at (5.25, -2) {};
		\node [style=none] (3) at (3.25, 1.75) {};
		\node [style=none] (4) at (0.5, -2.25) {$\chlc$};
		\node [style=none] (5) at (6, -2.25) {$\chrc$};
		\node [style=none] (6) at (3.25, 2.25) {$\ch$};
	\end{pgfonlayer}
	\begin{pgfonlayer}{edgelayer}
		\draw [in=-180, out=90, looseness=1.25] (1.center) to (0);
		\draw [in=0, out=90, looseness=1.25] (2.center) to (0);
		\draw (0) to (3.center);
	\end{pgfonlayer}
\end{tikzpicture} \, .
\end{equation}
\end{definition}
Traditional tensors between a left an right system have the seemingly obvious property that any state from the left system can be tensored with any state from the right system. This property does not hold for generalised tensors, different left and right states may or may not fit together. 
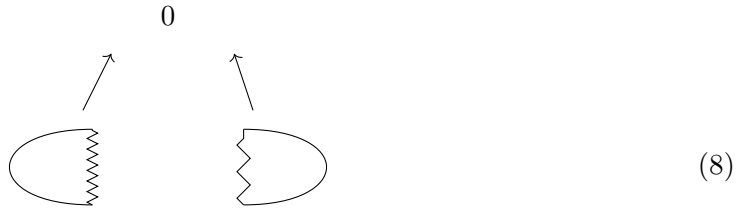
\be
\begin{tikzpicture}
	\begin{pgfonlayer}{nodelayer}
		\node [style=none] (2) at (-2.25, 1.5) {};
		\node [style=none] (3) at (-1.5, 3) {};
		\node [style=none] (4) at (2.25, 1.5) {};
		\node [style=none] (5) at (1.75, 3) {};
		\node [style=none] (6) at (-2, -1) {};
		\node [style=none] (7) at (-2, 1) {};
		\node [style=none] (8) at (2, -1) {};
		\node [style=none] (9) at (2, 1) {};
		\node [style=none] (11) at (0, 4) {$0$};
	\end{pgfonlayer}
	\begin{pgfonlayer}{edgelayer}
		\draw [style=arrow] (2.center) to (3.center);
		\draw [style=arrow] (4.center) to (5.center);
		\draw [bend left=90, looseness=3.75] (6.center) to (7.center);
		\draw [style=zigzag custom] (6.center) to (7.center);
		\draw [bend right=90, looseness=3.75] (8.center) to (9.center);
		\draw [style=zigzag edge] (8.center) to (9.center);
	\end{pgfonlayer}
\end{tikzpicture}
\ee

Generalised tensors can be extended from states to operators. 
Let $\ket{\psi} \in \chlc$ and $\ket{\phi} \in \chrc$, we will write $\ket{\psi} \tc \ket{\phi} = \chi^{\dagger}(\ket{\psi} \otimes \ket{\phi}) \in \ch \, .$ Following naturally, if $A \in \cl(\chlc)$ and $B \in \cl(\chrc)$, we will write 
\be
A \tc B = \chi^{\dagger} (A \otimes B) \chi = \begin{tikzpicture}
	\begin{pgfonlayer}{nodelayer}
		\node [style=white] (0) at (0, -2) {$\chi$};
		\node [style=none] (1) at (-2, 0) {};
		\node [style=none] (2) at (2, 0) {};
		\node [style=none] (3) at (0, -3.75) {};
		\node [style=white] (4) at (0, 2) {$\chi$};
		\node [style=none] (5) at (-2, 0) {};
		\node [style=none] (6) at (2, 0) {};
		\node [style=none] (7) at (0, 3.75) {};
		\node [style=sqr] (8) at (-2, 0) {$A$};
		\node [style=sqr] (9) at (2, 0) {$B$};
	\end{pgfonlayer}
	\begin{pgfonlayer}{edgelayer}
		\draw [in=180, out=-90, looseness=1.25] (1.center) to (0);
		\draw [in=0, out=-90, looseness=1.25] (2.center) to (0);
		\draw (0) to (3.center);
		\draw [in=-180, out=90, looseness=1.25] (5.center) to (4);
		\draw [in=0, out=90, looseness=1.25] (6.center) to (4);
		\draw (4) to (7.center);
	\end{pgfonlayer}
\end{tikzpicture} \in \cl(\ch) \, .
\ee
General quantum systems arise from considering operators that are local with respect to one side of a splitting map/ generalised tensor.

\begin{definition}[Subsystems from splitting maps]
Let $\chi : \ch \rightarrow \chlc \otimes \chrc$ be a splitting map. Then one can define:
\begin{itemize}
\item The operator system $\loc_{\LL}(\chi)$ of the left {\em $\chi$-local} elements of $\cl(\ch)$, which are the operators of the form 
\be
\chi^{\dagger} \left(A \otimes \id\right) \chi = A \tc \id = \begin{tikzpicture}
	\begin{pgfonlayer}{nodelayer}
		\node [style=white] (0) at (3, -2.5) {$\chi$};
		\node [style=none] (1) at (1, 0) {};
		\node [style=none] (2) at (5, 0) {};
		\node [style=none] (3) at (3, -4) {};
		\node [style=none] (4) at (0.75, -3.25) {$\chlc$};
		\node [style=none] (5) at (5.75, 0) {$\chrc$};
		\node [style=none] (6) at (3, -5) {$\ch$};
		\node [style=sqr] (7) at (1, 0) {$A$};
		\node [style=white] (9) at (3, 2.5) {$\chi$};
		\node [style=none] (10) at (3, 4) {};
		\node [style=none] (15) at (3, 5) {$\ch$};
	\end{pgfonlayer}
	\begin{pgfonlayer}{edgelayer}
		\draw [in=180, out=-90, looseness=1.25] (1.center) to (0);
		\draw [in=0, out=-90, looseness=1.25] (2.center) to (0);
		\draw (0) to (3.center);
		\draw [in=90, out=180, looseness=1.25] (9) to (7);
		\draw (9) to (10.center);
		\draw [in=90, out=0, looseness=1.25] (9) to (2.center);
	\end{pgfonlayer}
\end{tikzpicture} \, .
\ee
\item The von Neumann algebra $\cons_{\LL}(\chi)$ of the left {\em $\chi$-consistent} elements in $\cl\left(\chlc\right)$, which are the operators $A$ such that $(A \otimes \id)\chi \chi^{\dagger} = \chi \chi^{\dagger}  (A \otimes \id)$ or, written diagrammatically,
\begin{equation}
 \begin{tikzpicture}
	\begin{pgfonlayer}{nodelayer}
		\node [style=white] (0) at (4, -3.25) {$\chi$};
		\node [style=none] (1) at (2, -1.25) {};
		\node [style=none] (2) at (6, -1.25) {};
		\node [style=none] (3) at (4, -4.25) {};
		\node [style=sqr] (7) at (2, -1.25) {$A$};
		\node [style=white] (8) at (4, -5.25) {$\chi$};
		\node [style=none] (9) at (2, -7.25) {};
		\node [style=none] (10) at (6, -7.25) {};
		\node [style=none] (11) at (2, 0) {};
		\node [style=none] (12) at (6, 0) {};
		\node [style=none] (13) at (0, -3.5) {$=$};
		\node [style=white] (14) at (-4, -2) {$\chi$};
		\node [style=none] (15) at (-6, 0) {};
		\node [style=none] (16) at (-2, 0) {};
		\node [style=none] (17) at (-4, -3) {};
		\node [style=white] (21) at (-4, -4) {$\chi$};
		\node [style=none] (23) at (-2, -6) {};
		\node [style=sqr] (26) at (-6, -6) {$A$};
		\node [style=none] (27) at (-6, -7.25) {};
		\node [style=none] (28) at (-2, -7.25) {};
		\node [style=none] (29) at (8, -3.5) {.};
	\end{pgfonlayer}
	\begin{pgfonlayer}{edgelayer}
		\draw [in=180, out=-90, looseness=1.25] (1.center) to (0);
		\draw [in=0, out=-90, looseness=1.25] (2.center) to (0);
		\draw (0) to (3.center);
		\draw (3.center) to (8);
		\draw [in=90, out=-180, looseness=1.25] (8) to (9.center);
		\draw [in=90, out=0, looseness=1.25] (8) to (10.center);
		\draw [in=180, out=-90, looseness=1.25] (15.center) to (14);
		\draw [in=0, out=-90, looseness=1.25] (16.center) to (14);
		\draw (14) to (17.center);
		\draw (17.center) to (21);
		\draw [in=90, out=0, looseness=1.25] (21) to (23.center);
		\draw [in=90, out=-180, looseness=1.25] (21) to (26);
		\draw (11.center) to (7);
		\draw (12.center) to (2.center);
		\draw (23.center) to (28.center);
		\draw (26) to (27.center);
	\end{pgfonlayer}
\end{tikzpicture} 
\end{equation}
\item The von Neumann algebra $\stloc_{\LL}$ of the left {\em strictly $\chi$-local} elements of $\cl(\ch)$, which are the operators of the form $\chi^{\dagger} \left(A \otimes \id\right) \chi = A \tc \id$ with $A \in \cons_{\LL}(\chi)$.
\end{itemize}
And right local/consistent/strictly local operators symmetrically.
\end{definition}

In this article, we will only work with a special class of splitting maps/generalised tensors called canonical and that are constructed using a representation theorem for finite dimensional von Neumann algebras.
\begin{definition}[Canonical splitting maps]
A splitting map $\chi$ is said canonical if there exists families of Hilbert spaces $\{\ch_{\LL}^{i} \}$ and $\{\ch_{\R}^{i} \}$ as well as a unitary $U : \ch \rightarrow \bigoplus_i (\ch_{\LL}^{i} \otimes \ch_{\R}^{i})$ such that
\begin{equation}
\begin{split}
\chi : \ch & \rightarrow \bigoplus_i (\ch_{\LL}^{i} \otimes \ch_{\R}^{i}) \hookrightarrow (\bigoplus_i \ch_{\LL}^{i}) \otimes (\bigoplus_i  \ch_{\R}^{i}) = \chlc \otimes \chrc \\ 
\ket{\psi} & \rightarrow U \ket{\psi}
\end{split}
\end{equation}
When we don't want to be precise about the unitary $U$, we will simply write $\ch \cong\bigoplus_i (\ch_{\LL}^{i} \otimes \ch_{\R}^{i})$ to say that the two spaces are (unitarily) isomorphic.
\end{definition}

It is possible to say that a splitting is bigger than another, for example when its left part contains the left part of the second one. For this comparison to be possible, it is necessary that the composition of the two splittings is well defined, i.e.\ that it doesn't matter which one is applied first and which one second.
\be
\begin{tikzpicture}
	\begin{pgfonlayer}{nodelayer}
		\node [style=none] (0) at (-1, -3.75) {};
		\node [style=none] (1) at (-1, -5.75) {};
		\node [style=none] (2) at (-2.5, -3.5) {};
		\node [style=none] (3) at (-3.25, -2) {};
		\node [style=none] (4) at (1.5, -3.5) {};
		\node [style=none] (5) at (2, -2) {};
		\node [style=none] (6) at (-3, 0.5) {};
		\node [style=none] (7) at (-3, -1.5) {};
		\node [style=none] (10) at (1, -3.75) {};
		\node [style=none] (11) at (1, -5.75) {};
		\node [style=none] (12) at (0.25, 0.5) {};
		\node [style=none] (13) at (0.25, -1.5) {};
		\node [style=none] (15) at (2.25, 0.5) {};
		\node [style=none] (16) at (2.25, -1.5) {};
		\node [style=none] (17) at (0.75, 1.25) {};
		\node [style=none] (18) at (0, 2.75) {};
		\node [style=none] (19) at (3.25, 1.25) {};
		\node [style=none] (20) at (3.75, 2.75) {};
		\node [style=none] (21) at (-1.25, 5.25) {};
		\node [style=none] (22) at (-1.25, 3.25) {};
		\node [style=none] (23) at (0.75, 5.25) {};
		\node [style=none] (24) at (0.75, 3.25) {};
		\node [style=none] (27) at (3.25, 5.25) {};
		\node [style=none] (28) at (3.25, 3.25) {};
	\end{pgfonlayer}
	\begin{pgfonlayer}{edgelayer}
		\draw [bend right=90, looseness=3.75] (0.center) to (1.center);
		\draw [style=zigzag edge] (0.center) to (1.center);
		\draw [style=arrow] (2.center) to (3.center);
		\draw [style=arrow] (4.center) to (5.center);
		\draw [bend right=90, looseness=3.75] (6.center) to (7.center);
		\draw [style=zigzag edge] (6.center) to (7.center);
		\draw (0.center) to (10.center);
		\draw (1.center) to (11.center);
		\draw [style=zigzag custom] (10.center) to (11.center);
		\draw [bend left=90, looseness=3.75] (10.center) to (11.center);
		\draw [style=zigzag edge] (12.center) to (13.center);
		\draw (12.center) to (15.center);
		\draw (13.center) to (16.center);
		\draw [style=zigzag custom] (15.center) to (16.center);
		\draw [bend left=90, looseness=3.75] (15.center) to (16.center);
		\draw [style=arrow] (17.center) to (18.center);
		\draw [style=arrow] (19.center) to (20.center);
		\draw [style=zigzag edge] (21.center) to (22.center);
		\draw (21.center) to (23.center);
		\draw (22.center) to (24.center);
		\draw [style=zigzag custom] (23.center) to (24.center);
		\draw [style=zigzag custom] (27.center) to (28.center);
		\draw [bend left=90, looseness=3.75] (27.center) to (28.center);
	\end{pgfonlayer}
\end{tikzpicture} \quad = \quad \begin{tikzpicture}
	\begin{pgfonlayer}{nodelayer}
		\node [style=none] (0) at (0.75, -3.75) {};
		\node [style=none] (1) at (0.75, -5.75) {};
		\node [style=none] (2) at (2.25, -3.5) {};
		\node [style=none] (3) at (3, -2) {};
		\node [style=none] (4) at (-1.75, -3.5) {};
		\node [style=none] (5) at (-2.25, -2) {};
		\node [style=none] (6) at (2.75, 0.5) {};
		\node [style=none] (7) at (2.75, -1.5) {};
		\node [style=none] (10) at (-1.25, -3.75) {};
		\node [style=none] (11) at (-1.25, -5.75) {};
		\node [style=none] (12) at (-0.5, 0.5) {};
		\node [style=none] (13) at (-0.5, -1.5) {};
		\node [style=none] (15) at (-2.5, 0.5) {};
		\node [style=none] (16) at (-2.5, -1.5) {};
		\node [style=none] (17) at (-1, 1.25) {};
		\node [style=none] (18) at (-0.25, 2.75) {};
		\node [style=none] (19) at (-3.5, 1.25) {};
		\node [style=none] (20) at (-4, 2.75) {};
		\node [style=none] (21) at (1, 5.25) {};
		\node [style=none] (22) at (1, 3.25) {};
		\node [style=none] (23) at (-1, 5.25) {};
		\node [style=none] (24) at (-1, 3.25) {};
		\node [style=none] (27) at (-3.5, 5.25) {};
		\node [style=none] (28) at (-3.5, 3.25) {};
	\end{pgfonlayer}
	\begin{pgfonlayer}{edgelayer}
		\draw [bend left=90, looseness=3.75] (0.center) to (1.center);
		\draw [style=zigzag custom] (0.center) to (1.center);
		\draw [style=arrow] (2.center) to (3.center);
		\draw [style=arrow] (4.center) to (5.center);
		\draw [bend left=90, looseness=3.75] (6.center) to (7.center);
		\draw [style=zigzag custom] (6.center) to (7.center);
		\draw (0.center) to (10.center);
		\draw (1.center) to (11.center);
		\draw [style=zigzag edge] (10.center) to (11.center);
		\draw [bend right=90, looseness=3.75] (10.center) to (11.center);
		\draw [style=zigzag custom] (12.center) to (13.center);
		\draw (12.center) to (15.center);
		\draw (13.center) to (16.center);
		\draw [style=zigzag edge] (15.center) to (16.center);
		\draw [bend right=90, looseness=3.75] (15.center) to (16.center);
		\draw [style=arrow] (17.center) to (18.center);
		\draw [style=arrow] (19.center) to (20.center);
		\draw [style=zigzag custom] (21.center) to (22.center);
		\draw (21.center) to (23.center);
		\draw (22.center) to (24.center);
		\draw [style=zigzag edge] (23.center) to (24.center);
		\draw [style=zigzag edge] (27.center) to (28.center);
		\draw [bend right=90, looseness=3.75] (27.center) to (28.center);
	\end{pgfonlayer}
\end{tikzpicture} 
\ee
When applied to splitting maps, these ideas lead to the definition of a preorder called comprehension.

\begin{definition}[Comprehension]
Let $\chi$ and $\zeta$ be two canonical splitting maps on $\ch$. We say that $\chi$ is {\em comprehended} in $\zeta$ and write $\chi \sqsubseteq \zeta$ if there exists a Hilbert space $\ch_{\MM}$ and canonical splitting maps $\mu: \chlz \rightarrow \chlc \otimes \ch_{\MM}$  and $\xi : \chrc \rightarrow \ch_{\MM} \otimes \chrz$ such that $\left(\mu \otimes \id_{\chrz}\right) \zeta = \left(\id_{\chlc} \otimes \xi \right) \chi$ or equivalently that $\left(\left(\cdot \tm \cdot\right) \tz \cdot\right) = \left( \cdot \tc \left(\cdot \tx \cdot \right)\right) $ which can be stated diagrammatically as 
\be
\begin{tikzpicture}
	\begin{pgfonlayer}{nodelayer}
		\node [style=none] (0) at (-4.25, -2.75) {};
		\node [style=white] (1) at (-4.25, -1) {$\zeta$};
		\node [style=none] (5) at (-7.5, 3) {};
		\node [style=none] (6) at (-4.5, 3) {};
		\node [style=none] (7) at (-1.5, 3) {};
		\node [style=none] (10) at (4, -2.75) {};
		\node [style=white] (11) at (4, -1) {$\chi$};
		\node [style=none] (14) at (0, 0) {$=$};
		\node [style=none] (16) at (1.5, 3) {};
		\node [style=none] (17) at (4.5, 3) {};
		\node [style=none] (18) at (7.5, 3) {};
		\node [style=white] (19) at (-6, 1) {$\mu$};
		\node [style=white] (20) at (6, 1) {$\xi$};
		\node [style=none] (21) at (9, 0) {.};
	\end{pgfonlayer}
	\begin{pgfonlayer}{edgelayer}
		\draw (0.center) to (1);
		\draw (10.center) to (11);
		\draw [in=-90, out=0, looseness=1.25] (1) to (7.center);
		\draw [in=-90, out=180, looseness=1.25] (11) to (16.center);
		\draw [in=180, out=-90, looseness=1.25] (5.center) to (19);
		\draw [in=-90, out=0, looseness=1.25] (19) to (6.center);
		\draw [in=180, out=-90, looseness=1.25] (19) to (1);
		\draw [in=-90, out=0, looseness=1.25] (11) to (20);
		\draw [in=-90, out=-180, looseness=1.25] (20) to (17.center);
		\draw [in=-90, out=0, looseness=1.25] (20) to (18.center);
	\end{pgfonlayer}
\end{tikzpicture}
\ee
Given a Hilbert space $\ch$, comprehension is a preorder on the set of splitting maps on $\ch$. 
\end{definition}

\begin{remark}\label{good dots}
Note that the splitting maps $\mu$ and $\xi$ can always be chosen such that $\cons_{\LL}(\mu) = \cons_{\LL}(\chi)$ and $\cons_{\R}(\xi) = \cons_{\R}(\zeta)$ (see Appendix B. of \cite{mestoudjian2026picturinggeneralquantumsubsystems}). We will thus assume, without loss of generality, that the $\mu$ and $\xi$ witnessing a comprehension relation always satisfy this property.
\end{remark}

Moreover the comprehension preorder perfectly captures, in the case of canonical splitting maps, the inclusion of the corresponding algebras of strictly left-local operators.

\begin{proposition}
Let $\chi$ and $\zeta$ be canonical splitting maps, then $\chi \sqsubseteq \zeta$ if and only if $\stloc_{\LL}(\chi) \subseteq \stloc_{\LL}(\zeta)$.
\end{proposition}

%
%
%
%
%
%

\section{Minimal splitting maps}

Let us now define the notion of a minimal splitting map which will represent general quantum systems with trivial environment.

\begin{definition}[Minimal splitting map]
Let $\nu : \ch \rightarrow \ch \otimes \ch_{\R}$ be a canonical splitting map. We will say that $\nu$ is minimal if there exists an isometry $\tikzcircle{4pt} : \ch_{\R} \rightarrow \ch_{\R} \otimes \ch_{\R}$ such that
\begin{equation}
\label{eqminimal1}
 \begin{tikzpicture}
	\begin{pgfonlayer}{nodelayer}
		\node [style=white] (0) at (-5, 0) {$\nu$};
		\node [style=white] (1) at (-7, 2) {$\nu$};
		\node [style=none] (2) at (-2, 4) {};
		\node [style=none] (3) at (-5, -2) {};
		\node [style=none] (4) at (-9, 4) {};
		\node [style=none] (5) at (-5, 4) {};
		\node [style=none] (6) at (0, 0) {$=$};
		\node [style=white] (7) at (5, 0) {$\nu$};
		\node [style=black] (8) at (7, 2) {};
		\node [style=none] (9) at (2, 4) {};
		\node [style=none] (10) at (5, 4) {};
		\node [style=none] (11) at (9, 4) {};
		\node [style=none] (12) at (5, -2) {};
		\node [style=none] (13) at (-8, 1) {$\ch$};
		\node [style=none] (14) at (-5, -2.75) {$\ch$};
		\node [style=none] (15) at (5, -2.75) {$\ch$};
		\node [style=none] (16) at (2, 4.75) {$\ch$};
		\node [style=none] (17) at (-9, 4.75) {$\ch$};
		\node [style=none] (18) at (-5, 4.75) {$\ch_{\R}$};
		\node [style=none] (19) at (-2, 4.75) {$\ch_{\R}$};
		\node [style=none] (20) at (5, 4.75) {$\ch_{\R}$};
		\node [style=none] (21) at (9, 4.75) {$\ch_{\R}$};
	\end{pgfonlayer}
	\begin{pgfonlayer}{edgelayer}
		\draw [in=-180, out=-90, looseness=1.25] (4.center) to (1);
		\draw [in=-90, out=0, looseness=1.25] (1) to (5.center);
		\draw [in=180, out=-90, looseness=1.25] (1) to (0);
		\draw [in=-90, out=0, looseness=1.25] (0) to (2.center);
		\draw (0) to (3.center);
		\draw (7) to (12.center);
		\draw [in=-90, out=180, looseness=1.25] (7) to (9.center);
		\draw [in=-90, out=0, looseness=1.25] (7) to (8);
		\draw [in=-90, out=-180, looseness=1.25] (8) to (10.center);
		\draw [in=-90, out=0, looseness=1.25] (8) to (11.center);
	\end{pgfonlayer}
\end{tikzpicture} 
\end{equation}
and
\begin{equation}
\label{eqminimal2}
 \begin{tikzpicture}
	\begin{pgfonlayer}{nodelayer}
		\node [style=white] (0) at (1, -1) {$\nu$};
		\node [style=white] (1) at (1, 1) {$\nu$};
		\node [style=none] (2) at (-1, 3.25) {};
		\node [style=none] (3) at (3, 3.25) {};
		\node [style=none] (4) at (-1, -3.25) {};
		\node [style=none] (5) at (3, -3.25) {};
		\node [style=none] (6) at (5, 0) {$=$};
		\node [style=none] (7) at (-3, 0) {$=$};
		\node [style=white] (8) at (9, -2) {$\nu$};
		\node [style=black] (9) at (12.5, 2) {};
		\node [style=none] (11) at (7, 3) {};
		\node [style=none] (12) at (9, -3.25) {};
		\node [style=none] (13) at (14.5, -3.25) {};
		\node [style=none] (14) at (12.5, 3.25) {};
		\node [style=black] (15) at (-7, -2) {};
		\node [style=white] (16) at (-10.5, 2) {$\nu$};
		\node [style=none] (17) at (-10.5, 3.25) {};
		\node [style=none] (18) at (-5, 3.25) {};
		\node [style=none] (19) at (-7, -3.25) {};
		\node [style=none] (20) at (-12.5, -3.25) {};
	\end{pgfonlayer}
	\begin{pgfonlayer}{edgelayer}
		\draw [in=180, out=-90, looseness=1.25] (2.center) to (1);
		\draw [in=-90, out=0, looseness=1.25] (1) to (3.center);
		\draw (1) to (0);
		\draw [in=90, out=-180] (0) to (4.center);
		\draw [in=90, out=0] (0) to (5.center);
		\draw [in=-180, out=-90, looseness=0.75] (11.center) to (8);
		\draw (9) to (14.center);
		\draw [in=90, out=0, looseness=0.75] (9) to (13.center);
		\draw (8) to (12.center);
		\draw [in=0, out=-180, looseness=1.25] (9) to (8);
		\draw [in=-180, out=0, looseness=1.25] (16) to (15);
		\draw [in=90, out=-180, looseness=0.75] (16) to (20.center);
		\draw (16) to (17.center);
		\draw (15) to (19.center);
		\draw [in=0, out=-90, looseness=0.75] (18.center) to (15);
	\end{pgfonlayer}
\end{tikzpicture} 
\end{equation}
\end{definition}

\begin{proposition}
\label{shapeofminimal}
Let $\nu : \ch \rightarrow \ch \otimes \ch_{\R}$ be a canonical splitting map, then $\nu$ is minimal if and only if $\stloc_{\LL}(\nu) = \cons_{\LL}(\nu) = \bigoplus_i \cl(\ch_i)$, and it can therefore be written as
\begin{equation}
\nu : \ch = \bigoplus_{i=1}^N (\ch^{i} \otimes \mathbb{C}) \hookrightarrow (\bigoplus_{i=1}^N  \ch^{i}) \otimes (\bigoplus_{i=1}^N   \mathbb{C}) = \ch \otimes \mathbb{C}^{N} \, .
\end{equation}
\end{proposition}

Extensions of a minimal splitting map will be representing ways of embedding the minimal system into a bigger one with a (possibly) non-trivial  environment. 

\begin{definition}[Extension]
We will say that $\chi : \ch \rightarrow \ch_{\LL} \otimes \ch_{\R}$ is an extension of a minimal splitting map $\nu : \ch_{\LL} \rightarrow \ch_{\LL} \otimes \ch_{\MM}$ is there exists an isometry $\bigcirc : \ch_{\R} \rightarrow  \ch_{\MM} \otimes  \ch_{\R}$ such that 
\begin{equation}
\label{extension1}
 \begin{tikzpicture}
	\begin{pgfonlayer}{nodelayer}
		\node [style=white] (0) at (-5, 0) {$\chi$};
		\node [style=white] (1) at (-7, 2) {$\nu$};
		\node [style=none] (2) at (-2, 4) {};
		\node [style=none] (3) at (-5, -2) {};
		\node [style=none] (4) at (-9, 4) {};
		\node [style=none] (5) at (-5, 4) {};
		\node [style=none] (6) at (0, 0) {$=$};
		\node [style=white] (7) at (5, 0) {$\chi$};
		\node [style=none] (9) at (2, 4) {};
		\node [style=none] (10) at (5, 4) {};
		\node [style=none] (11) at (9, 4) {};
		\node [style=none] (12) at (5, -2) {};
		\node [style=none] (13) at (-8, 1) {$\ch_{\LL}$};
		\node [style=none] (14) at (-5, -2.75) {$\ch$};
		\node [style=none] (15) at (5, -2.75) {$\ch$};
		\node [style=none] (16) at (2, 4.75) {$\ch_{\LL}$};
		\node [style=none] (17) at (-9, 4.75) {$\ch_{\LL}$};
		\node [style=none] (18) at (-5, 4.75) {$\ch_{\MM}$};
		\node [style=none] (19) at (-2, 4.75) {$\ch_{\R}$};
		\node [style=none] (20) at (5, 4.75) {$\ch_{\MM}$};
		\node [style=none] (21) at (9, 4.75) {$\ch_{\R}$};
		\node [style=white] (22) at (7, 2) {};
	\end{pgfonlayer}
	\begin{pgfonlayer}{edgelayer}
		\draw [in=-180, out=-90, looseness=1.25] (4.center) to (1);
		\draw [in=-90, out=0, looseness=1.25] (1) to (5.center);
		\draw [in=180, out=-90, looseness=1.25] (1) to (0);
		\draw [in=-90, out=0, looseness=1.25] (0) to (2.center);
		\draw (0) to (3.center);
		\draw (7) to (12.center);
		\draw [in=-90, out=180, looseness=1.25] (7) to (9.center);
		\draw [in=180, out=-90, looseness=1.25] (10.center) to (22);
		\draw [in=-90, out=0, looseness=1.25] (22) to (11.center);
		\draw [in=0, out=-90, looseness=1.25] (22) to (7);
	\end{pgfonlayer}
\end{tikzpicture} 
\end{equation}
and
\begin{equation}
\label{extension2}
 \begin{tikzpicture}
	\begin{pgfonlayer}{nodelayer}
		\node [style=white] (0) at (1, -1) {$\chi$};
		\node [style=white] (1) at (1, 1) {$\chi$};
		\node [style=none] (2) at (-1, 3.25) {};
		\node [style=none] (3) at (3, 3.25) {};
		\node [style=none] (4) at (-1, -3.25) {};
		\node [style=none] (5) at (3, -3.25) {};
		\node [style=none] (6) at (5, 0) {$=$};
		\node [style=none] (7) at (-3, 0) {$=$};
		\node [style=white] (8) at (9, -2) {$\nu$};
		\node [style=none] (11) at (7, 3) {};
		\node [style=none] (12) at (9, -3.25) {};
		\node [style=none] (13) at (14.5, -3.25) {};
		\node [style=none] (14) at (12.5, 3.25) {};
		\node [style=white] (16) at (-10.5, 2) {$\nu$};
		\node [style=none] (17) at (-10.5, 3.25) {};
		\node [style=none] (18) at (-5, 3.25) {};
		\node [style=none] (19) at (-7, -3.25) {};
		\node [style=none] (20) at (-12.5, -3.25) {};
		\node [style=white] (21) at (12.5, 2) {};
		\node [style=white] (22) at (-7, -2) {};
	\end{pgfonlayer}
	\begin{pgfonlayer}{edgelayer}
		\draw [in=180, out=-90, looseness=1.25] (2.center) to (1);
		\draw [in=-90, out=0, looseness=1.25] (1) to (3.center);
		\draw (1) to (0);
		\draw [in=90, out=-180] (0) to (4.center);
		\draw [in=90, out=0] (0) to (5.center);
		\draw [in=-180, out=-90, looseness=0.75] (11.center) to (8);
		\draw (8) to (12.center);
		\draw [in=90, out=-180, looseness=0.75] (16) to (20.center);
		\draw (16) to (17.center);
		\draw [in=180, out=0] (16) to (22);
		\draw [in=-90, out=0, looseness=0.75] (22) to (18.center);
		\draw (22) to (19.center);
		\draw [in=-180, out=0] (8) to (21);
		\draw (21) to (14.center);
		\draw [in=90, out=0, looseness=0.75] (21) to (13.center);
	\end{pgfonlayer}
\end{tikzpicture} 
\end{equation}
\end{definition}

Note that a canonical splitting map is minimal if and only if it is an extension of itself.

\begin{proposition}
\label{shapeofextensions}
Let $\nu : \ch_{\LL} \rightarrow \ch_{\LL} \otimes \ch_{\MM}$ be a minimal splitting map, a canonical splitting map $\chi : \ch \rightarrow \ch_{\LL} \otimes \ch_{\R}$ is an extension of $\nu$ if and only if $\cons_{\LL}(\chi) = \cons_{\LL}(\nu)$ and therefore that they can jointly be written
\begin{equation}
\nu : \ch_{\LL} = \bigoplus_{i=1}^N (\ch_{\LL}^{i} \otimes \mathbb{C}) \hookrightarrow (\bigoplus_{i=1}^N  \ch_{\LL}^{i}) \otimes (\bigoplus_{i=1}^N   \mathbb{C}) = \ch_{\LL} \otimes \mathbb{C}^{N} \, , 
\end{equation}
\begin{equation}
\chi : \ch \cong \bigoplus_{i=1}^N (\ch_{\LL}^{i} \otimes \ch_{\R}^{i}) \hookrightarrow (\bigoplus_{i=1}^N  \ch_{\LL}^{i}) \otimes (\bigoplus_{i=1}^N   \ch_{\R}^{i}) = \ch_{\LL} \otimes \ch_{\R} \, .
\end{equation}
\end{proposition}

\begin{remark}
Note that the reasoning used in Appendix \ref{proofextensions} to prove the previous proposition could symmetrically be applied to $\bigcirc$, giving that $\cons_{\R}(\bigcirc) = \cons_{\R}(\chi)$ and thus making Equation \ref{extension1} a witness of the comprehension $\chi \sqsubseteq \chi$.
\end{remark}

\section{Linearity, locality and unitarity}


Let us now turn to the properties that the relativistic evolutions of quantum systems should satisfy.
\begin{itemize}
\item An evolution of a system should not depend on the potential environment of the system and it should therefore be possible to apply it locally on the system, before composing the system with an environment. We will call this property ``local applicability".
\be
\begin{tikzpicture}
	\begin{pgfonlayer}{nodelayer}
		\node [style=none] (2) at (-2.25, -2.25) {};
		\node [style=none] (3) at (-1.5, -0.75) {};
		\node [style=none] (4) at (2.25, -2.25) {};
		\node [style=none] (5) at (1.75, -0.75) {};
		\node [style=none] (6) at (-2, -4.75) {};
		\node [style=none] (7) at (-2, -2.75) {};
		\node [style=none] (8) at (2, -4.75) {};
		\node [style=none] (9) at (2, -2.75) {};
		\node [style=none] (11) at (0, 1.5) {};
		\node [style=none] (12) at (0, -0.5) {};
		\node [style=none] (15) at (-1, 0.5) {};
		\node [style=none] (16) at (-1, 3.75) {};
		\node [style=none] (17) at (-1.5, 2.75) {};
		\node [style=none] (18) at (-0.5, 2.75) {};
		\node [style=none] (19) at (-0.5, 1.75) {};
		\node [style=none] (20) at (-1.5, 1.75) {};
		\node [style=none] (21) at (-1, 2.75) {};
		\node [style=none] (22) at (-1, 1.75) {};
		\node [style=none] (23) at (-1, 2.25) {$\mathscr{L}$};
		\node [style=none] (24) at (0, 5) {};
		\node [style=none] (25) at (0, 3) {};
		\node [style=none] (26) at (1, 3.75) {};
		\node [style=none] (27) at (1, 0.5) {};
	\end{pgfonlayer}
	\begin{pgfonlayer}{edgelayer}
		\draw [style=arrow] (2.center) to (3.center);
		\draw [style=arrow] (4.center) to (5.center);
		\draw [bend left=90, looseness=3.75] (6.center) to (7.center);
		\draw [style=zigzag edge] (6.center) to (7.center);
		\draw [bend right=90, looseness=3.75] (8.center) to (9.center);
		\draw [style=zigzag edge] (8.center) to (9.center);
		\draw [bend right=90, looseness=3.75] (11.center) to (12.center);
		\draw [bend left=90, looseness=3.75] (11.center) to (12.center);
		\draw [style=zigzag edge] (11.center) to (12.center);
		\draw (17.center) to (18.center);
		\draw (18.center) to (19.center);
		\draw (19.center) to (20.center);
		\draw (20.center) to (17.center);
		\draw (16.center) to (21.center);
		\draw (22.center) to (15.center);
		\draw [bend right=90, looseness=3.75] (24.center) to (25.center);
		\draw [bend left=90, looseness=3.75] (24.center) to (25.center);
		\draw [style=zigzag edge] (24.center) to (25.center);
		\draw (26.center) to (27.center);
	\end{pgfonlayer}
\end{tikzpicture} \quad = \quad \begin{tikzpicture}
	\begin{pgfonlayer}{nodelayer}
		\node [style=none] (2) at (-2.25, 1.5) {};
		\node [style=none] (3) at (-1.5, 3) {};
		\node [style=none] (4) at (2.25, 1.5) {};
		\node [style=none] (5) at (1.75, 3) {};
		\node [style=none] (6) at (-2, -1) {};
		\node [style=none] (7) at (-2, 1) {};
		\node [style=none] (8) at (2, -1) {};
		\node [style=none] (9) at (2, 1) {};
		\node [style=none] (11) at (0, 5.25) {};
		\node [style=none] (12) at (0, 3.25) {};
		\node [style=none] (13) at (-2, -4.5) {};
		\node [style=none] (14) at (-2, -2.5) {};
		\node [style=none] (15) at (-3, -3.5) {};
		\node [style=none] (16) at (-3, -0.25) {};
		\node [style=none] (21) at (-3, -1.25) {};
		\node [style=none] (22) at (-3, -2.25) {};
		\node [style=none] (23) at (-3.5, -1.25) {};
		\node [style=none] (24) at (-2.5, -1.25) {};
		\node [style=none] (25) at (-2.5, -2.25) {};
		\node [style=none] (26) at (-3.5, -2.25) {};
		\node [style=none] (27) at (-3, -1.25) {};
		\node [style=none] (28) at (-3, -2.25) {};
		\node [style=none] (29) at (-3, -1.75) {$\mathscr{L}$};
	\end{pgfonlayer}
	\begin{pgfonlayer}{edgelayer}
		\draw [style=arrow] (2.center) to (3.center);
		\draw [style=arrow] (4.center) to (5.center);
		\draw [bend left=90, looseness=3.75] (6.center) to (7.center);
		\draw [style=zigzag edge] (6.center) to (7.center);
		\draw [bend right=90, looseness=3.75] (8.center) to (9.center);
		\draw [style=zigzag edge] (8.center) to (9.center);
		\draw [bend right=90, looseness=3.75] (11.center) to (12.center);
		\draw [bend left=90, looseness=3.75] (11.center) to (12.center);
		\draw [style=zigzag edge] (11.center) to (12.center);
		\draw [bend left=90, looseness=3.75] (13.center) to (14.center);
		\draw [style=zigzag edge] (13.center) to (14.center);
		\draw (16.center) to (21.center);
		\draw (22.center) to (15.center);
		\draw (23.center) to (24.center);
		\draw (24.center) to (25.center);
		\draw (25.center) to (26.center);
		\draw (26.center) to (23.center);
	\end{pgfonlayer}
\end{tikzpicture} 
\ee
\item An evolution of the system should not signal to an observer on the environment, by which we mean that it should not influence the result of a measurement on the environment. This property called ``no superluminal signalling" corresponds to  the following picture where the doted lines represent the future light cone of the evolution.
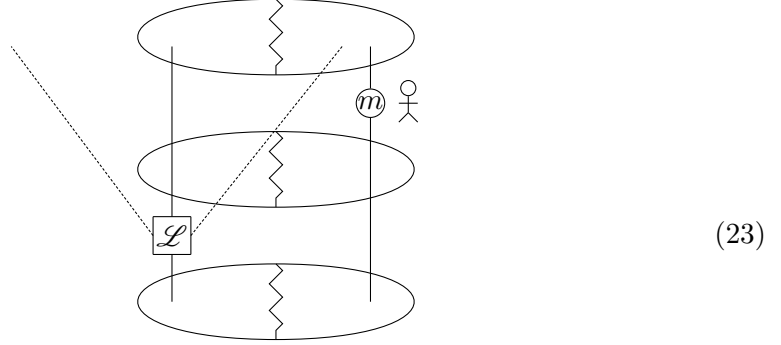
\be
\begin{tikzpicture}
	\begin{pgfonlayer}{nodelayer}
		\node [style=none] (11) at (0, -0.75) {};
		\node [style=none] (12) at (0, -2.75) {};
		\node [style=none] (15) at (-2.75, -1.75) {};
		\node [style=none] (16) at (-2.75, 5) {};
		\node [style=none] (17) at (-3.25, 0.5) {};
		\node [style=none] (18) at (-2.25, 0.5) {};
		\node [style=none] (19) at (-2.25, -0.5) {};
		\node [style=none] (20) at (-3.25, -0.5) {};
		\node [style=none] (21) at (-2.75, 0.5) {};
		\node [style=none] (22) at (-2.75, -0.5) {};
		\node [style=none] (23) at (-2.75, 0) {$\mathscr{L}$};
		\node [style=none] (24) at (0, 2.75) {};
		\node [style=none] (25) at (0, 0.75) {};
		\node [style=none] (26) at (-2.25, 0) {};
		\node [style=none] (27) at (-3.25, 0) {};
		\node [style=none] (28) at (-7, 5) {};
		\node [style=none] (29) at (1.75, 5) {};
		\node [style=none] (30) at (2.5, -1.75) {};
		\node [style={white_dot}] (37) at (2.5, 3.5) {$m$};
		\node [style=none] (38) at (2.5, 5) {};
		\node [style={white_dot}] (39) at (3.5, 3.9) {};
		\node [style=none] (40) at (3.5, 3.25) {};
		\node [style=none] (41) at (3.25, 3) {};
		\node [style=none] (42) at (3.75, 3) {};
		\node [style=none] (43) at (3.75, 3.5) {};
		\node [style=none] (44) at (3.25, 3.5) {};
		\node [style=none] (45) at (0, 6.25) {};
		\node [style=none] (46) at (0, 4.25) {};
	\end{pgfonlayer}
	\begin{pgfonlayer}{edgelayer}
		\draw [bend left=270, looseness=6.25] (11.center) to (12.center);
		\draw [bend left=90, looseness=6.25] (11.center) to (12.center);
		\draw [style=zigzag edge] (11.center) to (12.center);
		\draw (17.center) to (18.center);
		\draw (18.center) to (19.center);
		\draw (19.center) to (20.center);
		\draw (20.center) to (17.center);
		\draw (16.center) to (21.center);
		\draw (22.center) to (15.center);
		\draw [bend right=90, looseness=6.25] (24.center) to (25.center);
		\draw [bend left=90, looseness=6.25] (24.center) to (25.center);
		\draw [style=zigzag edge] (24.center) to (25.center);
		\draw [style=dottededge] (29.center) to (26.center);
		\draw [style=dottededge] (27.center) to (28.center);
		\draw (37) to (30.center);
		\draw (38.center) to (37);
		\draw (39) to (40.center);
		\draw (43.center) to (44.center);
		\draw (41.center) to (40.center);
		\draw (40.center) to (42.center);
		\draw [bend right=90, looseness=6.25] (45.center) to (46.center);
		\draw [bend left=90, looseness=6.25] (45.center) to (46.center);
		\draw [style=zigzag edge] (45.center) to (46.center);
	\end{pgfonlayer}
\end{tikzpicture}
\ee
\item Finally, it should not matter whether the evolution of the system is applied before or after a measurement on the environment. We will call this property ``update commutativity".
\be
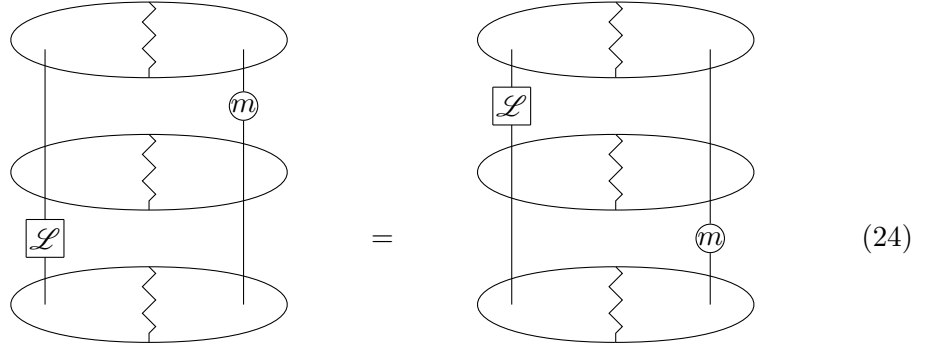
\begin{tikzpicture}
	\begin{pgfonlayer}{nodelayer}
		\node [style=none] (11) at (0, -0.75) {};
		\node [style=none] (12) at (0, -2.75) {};
		\node [style=none] (15) at (-2.75, -1.75) {};
		\node [style=none] (16) at (-2.75, 5) {};
		\node [style=none] (17) at (-3.25, 0.5) {};
		\node [style=none] (18) at (-2.25, 0.5) {};
		\node [style=none] (19) at (-2.25, -0.5) {};
		\node [style=none] (20) at (-3.25, -0.5) {};
		\node [style=none] (21) at (-2.75, 0.5) {};
		\node [style=none] (22) at (-2.75, -0.5) {};
		\node [style=none] (23) at (-2.75, 0) {$\mathscr{L}$};
		\node [style=none] (24) at (0, 2.75) {};
		\node [style=none] (25) at (0, 0.75) {};
		\node [style=none] (26) at (-2.25, 0) {};
		\node [style=none] (27) at (-3.25, 0) {};
		\node [style=none] (30) at (2.5, -1.75) {};
		\node [style={white_dot}] (37) at (2.5, 3.5) {$m$};
		\node [style=none] (38) at (2.5, 5) {};
		\node [style=none] (45) at (0, 6.25) {};
		\node [style=none] (46) at (0, 4.25) {};
	\end{pgfonlayer}
	\begin{pgfonlayer}{edgelayer}
		\draw [bend left=270, looseness=6.25] (11.center) to (12.center);
		\draw [bend left=90, looseness=6.25] (11.center) to (12.center);
		\draw [style=zigzag edge] (11.center) to (12.center);
		\draw (17.center) to (18.center);
		\draw (18.center) to (19.center);
		\draw (19.center) to (20.center);
		\draw (20.center) to (17.center);
		\draw (16.center) to (21.center);
		\draw (22.center) to (15.center);
		\draw [bend right=90, looseness=6.25] (24.center) to (25.center);
		\draw [bend left=90, looseness=6.25] (24.center) to (25.center);
		\draw [style=zigzag edge] (24.center) to (25.center);
		\draw (37) to (30.center);
		\draw (38.center) to (37);
		\draw [bend right=90, looseness=6.25] (45.center) to (46.center);
		\draw [bend left=90, looseness=6.25] (45.center) to (46.center);
		\draw [style=zigzag edge] (45.center) to (46.center);
	\end{pgfonlayer}
\end{tikzpicture} \quad = \quad \begin{tikzpicture}
	\begin{pgfonlayer}{nodelayer}
		\node [style=none] (11) at (0, -0.75) {};
		\node [style=none] (12) at (0, -2.75) {};
		\node [style=none] (15) at (-2.75, -1.75) {};
		\node [style=none] (16) at (-2.75, 5) {};
		\node [style=none] (17) at (-3.25, 4) {};
		\node [style=none] (18) at (-2.25, 4) {};
		\node [style=none] (19) at (-2.25, 3) {};
		\node [style=none] (20) at (-3.25, 3) {};
		\node [style=none] (21) at (-2.75, 4) {};
		\node [style=none] (22) at (-2.75, 3) {};
		\node [style=none] (23) at (-2.75, 3.5) {$\mathscr{L}$};
		\node [style=none] (24) at (0, 2.75) {};
		\node [style=none] (25) at (0, 0.75) {};
		\node [style=none] (26) at (-2.25, 3.5) {};
		\node [style=none] (27) at (-3.25, 3.5) {};
		\node [style=none] (30) at (2.5, -1.75) {};
		\node [style={white_dot}] (37) at (2.5, 0) {$m$};
		\node [style=none] (38) at (2.5, 5) {};
		\node [style=none] (45) at (0, 6.25) {};
		\node [style=none] (46) at (0, 4.25) {};
	\end{pgfonlayer}
	\begin{pgfonlayer}{edgelayer}
		\draw [bend left=270, looseness=6.25] (11.center) to (12.center);
		\draw [bend left=90, looseness=6.25] (11.center) to (12.center);
		\draw [style=zigzag edge] (11.center) to (12.center);
		\draw (17.center) to (18.center);
		\draw (18.center) to (19.center);
		\draw (19.center) to (20.center);
		\draw (20.center) to (17.center);
		\draw (16.center) to (21.center);
		\draw (22.center) to (15.center);
		\draw [bend right=90, looseness=6.25] (24.center) to (25.center);
		\draw [bend left=90, looseness=6.25] (24.center) to (25.center);
		\draw [style=zigzag edge] (24.center) to (25.center);
		\draw (37) to (30.center);
		\draw (38.center) to (37);
		\draw [bend right=90, looseness=6.25] (45.center) to (46.center);
		\draw [bend left=90, looseness=6.25] (45.center) to (46.center);
		\draw [style=zigzag edge] (45.center) to (46.center);
	\end{pgfonlayer}
\end{tikzpicture} 
\ee
\end{itemize}

These properties lead to the following mathematical definition. Given a Hilbert space $\ch$, we define $S_{\ch} = \{ \ket{\psi} \in \ch, \braket{\psi}{\psi} = 1 \}$ the set of (normalised) states of $\ch$. We call measurement results any orthogonal projector.

\begin{definition}[Local evolutions of general quantum systems]
\label{deflocalevolutions}
A local evolution of the quantum system associated to a minimal splitting map $\nu : \ch^{\nu} \rightarrow \ch^{\nu} \otimes \ch_{\R}^{\nu}$ is the data of a function $\ml_{\chi} : S_{\ch^{\chi}} \rightarrow S_{\ch^{\chi}}$ for every extension $\chi : \ch^{\chi} \rightarrow \ch^{\nu} \otimes \chrc$ of $\nu$, such that:
\begin{itemize}
\item Local applicability. For every $\zeta$ such that $\chi \sqsubseteq \zeta$, i.e. such that $\left(\left(\cdot \tm \cdot\right) \tz \cdot\right) = \left( \cdot \tc \left(\cdot \tx \cdot \right)\right)$ and every $\ket{\phi}, \ket{\psi}$ such that $\ket{\phi} \tz \ket{\psi} \neq 0$,
\be
\label{extension}
\ml_{\chi} \left( \frac{\ket{\phi} \tz \ket{\psi}}{\left\Vert \ket{\phi} \tz \ket{\psi} \right\Vert} \right) = \frac{\ml_{\mu}(\ket{\phi}) \tz \ket{\psi}}{\left\Vert \ml_{\mu}(\ket{\phi}) \tz \ket{\psi} \right\Vert } \, ,
\ee
or written diagrammatically that 
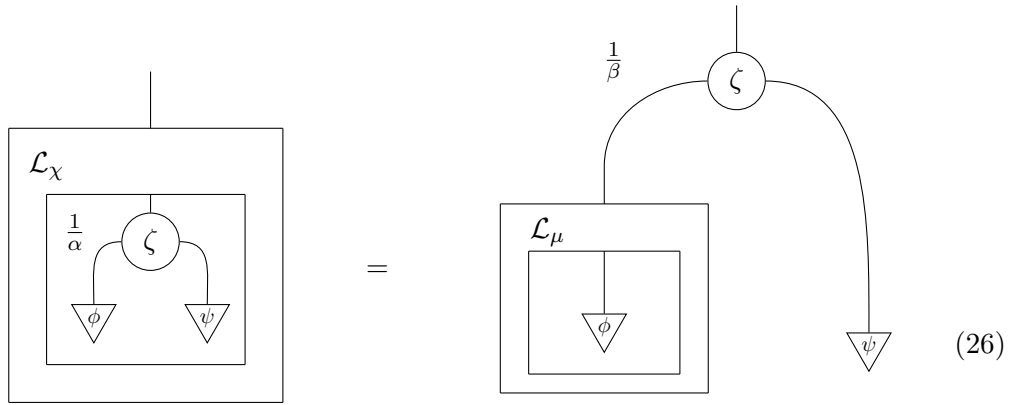
\begin{equation}
 \begin{tikzpicture}
	\begin{pgfonlayer}{nodelayer}
		\node [style=none] (0) at (-12.25, 5.75) {};
		\node [style=none] (1) at (-5, 5.75) {};
		\node [style=none] (2) at (-12.25, -1.5) {};
		\node [style=none] (3) at (-5, -1.5) {};
		\node [style=none] (4) at (-11.25, 4) {};
		\node [style=none] (5) at (-6, 4) {};
		\node [style=none] (6) at (-11.25, -0.5) {};
		\node [style=none] (7) at (-6, -0.5) {};
		\node [style=none] (8) at (-8.5, 5.75) {};
		\node [style=none] (9) at (-8.5, 4) {};
		\node [style=none] (10) at (-11.25, 4.75) {$\cl_{\chi}$};
		\node [style=none] (12) at (-8.5, 7.25) {};
		\node [style=none] (14) at (0.75, 3.75) {};
		\node [style=none] (15) at (6.25, 3.75) {};
		\node [style=none] (16) at (0.75, -1.25) {};
		\node [style=none] (17) at (6.25, -1.25) {};
		\node [style=none] (18) at (1.5, 2.5) {};
		\node [style=none] (19) at (5.5, 2.5) {};
		\node [style=none] (20) at (1.5, -0.75) {};
		\node [style=none] (21) at (5.5, -0.75) {};
		\node [style=none] (22) at (3.5, 3.75) {};
		\node [style=none] (23) at (3.5, 2.5) {};
		\node [style=none] (24) at (2, 3) {$\cl_{\mu}$};
		\node [style=none] (25) at (3.5, 4.75) {};
		\node [style=none] (26) at (3.5, 1.5) {};
		\node [style=white] (27) at (-8.5, 2.75) {$\zeta$};
		\node [style=state] (28) at (-10, 0.75) {$\phi$};
		\node [style=state] (29) at (-7, 0.75) {$\psi$};
		\node [style=state] (30) at (3.5, 0.5) {$\phi$};
		\node [style=white] (31) at (7, 7) {$\zeta$};
		\node [style=state] (32) at (10.5, 0) {$\psi$};
		\node [style=none] (33) at (-2.5, 2) {$=$};
		\node [style=none] (34) at (-10.5, 3) {$\frac{1}{\alpha}$};
		\node [style=none] (35) at (3.75, 7.5) {$\frac{1}{\beta}$};
		\node [style=none] (36) at (7, 9) {};
	\end{pgfonlayer}
	\begin{pgfonlayer}{edgelayer}
		\draw (0.center) to (1.center);
		\draw (4.center) to (5.center);
		\draw (6.center) to (7.center);
		\draw (7.center) to (5.center);
		\draw (4.center) to (6.center);
		\draw (0.center) to (2.center);
		\draw (2.center) to (3.center);
		\draw (3.center) to (1.center);
		\draw (8.center) to (12.center);
		\draw (14.center) to (15.center);
		\draw (18.center) to (19.center);
		\draw (20.center) to (21.center);
		\draw (21.center) to (19.center);
		\draw (18.center) to (20.center);
		\draw (14.center) to (16.center);
		\draw (16.center) to (17.center);
		\draw (17.center) to (15.center);
		\draw (26.center) to (23.center);
		\draw (22.center) to (25.center);
		\draw (9.center) to (27);
		\draw [in=90, out=-180, looseness=1.25] (27) to (28);
		\draw [in=90, out=0, looseness=1.25] (27) to (29);
		\draw (26.center) to (30);
		\draw [in=180, out=90] (25.center) to (31);
		\draw [in=90, out=0] (31) to (32);
		\draw (36.center) to (31);
	\end{pgfonlayer}
\end{tikzpicture} 
\end{equation}
where $\alpha = \left\Vert \ket{\phi} \tz \ket{\psi} \right\Vert$ and $\beta = \left\Vert \ml_{\mu}(\ket{\phi}) \tz \ket{\psi} \right\Vert$.
\item No superluminal signalling. For every right $\chi$-consistent measurement result $m$,
\be
\label{probas}
\left\Vert (\id \tc m ) \ml_{\chi}(\ket{\psi}) \right\Vert = \left\Vert (\id \tc m )\ket{\psi} \right\Vert \, , 
\ee
which can be written diagrammatically as
\begin{equation}
 \begin{tikzpicture}
	\begin{pgfonlayer}{nodelayer}
		\node [style=state] (15) at (5, -2) {$\psi$};
		\node [style=none] (28) at (-5, -0.25) {};
		\node [style=white] (31) at (5, 0) {$\chi$};
		\node [style=white] (32) at (5, 4) {$\chi$};
		\node [style=none] (33) at (3, 2) {};
		\node [style=sqr] (34) at (7, 2) {$m$};
		\node [style=effect] (35) at (5, 6) {$\psi$};
		\node [style=none] (36) at (0, 2) {$=$};
		\node [style=white] (37) at (-5, -0.25) {$\chi$};
		\node [style=white] (38) at (-5, 3.75) {$\chi$};
		\node [style=none] (39) at (-7, 1.75) {};
		\node [style=sqr] (40) at (-3, 1.75) {$m$};
		\node [style=none] (41) at (-6.75, 5) {};
		\node [style=none] (42) at (-3.25, 5) {};
		\node [style=none] (43) at (-6.75, 8.5) {};
		\node [style=none] (44) at (-3.25, 8.5) {};
		\node [style=none] (45) at (-6.25, 6) {};
		\node [style=none] (46) at (-3.75, 6) {};
		\node [style=none] (47) at (-6.25, 8) {};
		\node [style=none] (48) at (-3.75, 8) {};
		\node [style=none] (49) at (-5, 5) {};
		\node [style=none] (50) at (-5, 6) {};
		\node [style=none] (51) at (-4, 5.5) {$\cl_{\chi}$};
		\node [style=effect] (54) at (-5, 7) {$\psi$};
		\node [style=none] (55) at (-6.75, -1.5) {};
		\node [style=none] (56) at (-3.25, -1.5) {};
		\node [style=none] (57) at (-6.75, -5) {};
		\node [style=none] (58) at (-3.25, -5) {};
		\node [style=none] (59) at (-6.25, -2.5) {};
		\node [style=none] (60) at (-3.75, -2.5) {};
		\node [style=none] (61) at (-6.25, -4.5) {};
		\node [style=none] (62) at (-3.75, -4.5) {};
		\node [style=none] (63) at (-5, -1.5) {};
		\node [style=none] (64) at (-5, -2.5) {};
		\node [style=none] (65) at (-6, -2) {$\cl_{\chi}$};
		\node [style=state] (67) at (-5, -3.5) {$\psi$};
		\node [style=none] (68) at (9, 2) {.};
	\end{pgfonlayer}
	\begin{pgfonlayer}{edgelayer}
		\draw [in=90, out=0, looseness=1.25] (32) to (34);
		\draw [in=0, out=-90, looseness=1.25] (34) to (31);
		\draw [in=270, out=180, looseness=1.25] (31) to (33.center);
		\draw [in=180, out=90, looseness=1.25] (33.center) to (32);
		\draw (31) to (15);
		\draw (32) to (35);
		\draw [in=90, out=0, looseness=1.25] (38) to (40);
		\draw [in=0, out=-90, looseness=1.25] (40) to (37);
		\draw [in=270, out=180, looseness=1.25] (37) to (39.center);
		\draw [in=180, out=90, looseness=1.25] (39.center) to (38);
		\draw (41.center) to (42.center);
		\draw (45.center) to (46.center);
		\draw (47.center) to (48.center);
		\draw (48.center) to (46.center);
		\draw (45.center) to (47.center);
		\draw (41.center) to (43.center);
		\draw (43.center) to (44.center);
		\draw (44.center) to (42.center);
		\draw (54) to (50.center);
		\draw (55.center) to (56.center);
		\draw (59.center) to (60.center);
		\draw (61.center) to (62.center);
		\draw (62.center) to (60.center);
		\draw (59.center) to (61.center);
		\draw (55.center) to (57.center);
		\draw (57.center) to (58.center);
		\draw (58.center) to (56.center);
		\draw (37) to (63.center);
		\draw (64.center) to (67);
		\draw (49.center) to (38);
	\end{pgfonlayer}
\end{tikzpicture} 
\end{equation}
\item Update commutativity. For every right $\chi$-consistent measurement result $m$,
\be
\label{update}
\frac{(\id \tc m)\ml_{\chi}(\ket{\psi})}{\lv (\id \tc m)\ml_{\chi}(\ket{\psi}) \rv} = \ml_{\chi} \left( \frac{(\id \tc m)\ket{\psi}}{\lv (\id \tc m)\ket{\psi} \rv} \right) \, 
\ee
or equivalently 
\begin{equation}
 \begin{tikzpicture}
	\begin{pgfonlayer}{nodelayer}
		\node [style=state] (15) at (5, -2.75) {$\psi$};
		\node [style=none] (28) at (-8, 0.5) {};
		\node [style=white] (31) at (5, -1.25) {$\chi$};
		\node [style=white] (32) at (5, 2.75) {$\chi$};
		\node [style=none] (33) at (3, 0.75) {};
		\node [style=sqr] (34) at (7, 0.75) {$m$};
		\node [style=none] (36) at (-3, 0) {$=$};
		\node [style=white] (37) at (-8, 0.5) {$\chi$};
		\node [style=white] (38) at (-8, 4.5) {$\chi$};
		\node [style=none] (39) at (-10, 2.5) {};
		\node [style=sqr] (40) at (-6, 2.5) {$m$};
		\node [style=none] (55) at (-9.75, -0.75) {};
		\node [style=none] (56) at (-6.25, -0.75) {};
		\node [style=none] (57) at (-9.75, -4.25) {};
		\node [style=none] (58) at (-6.25, -4.25) {};
		\node [style=none] (59) at (-9.25, -1.75) {};
		\node [style=none] (60) at (-6.75, -1.75) {};
		\node [style=none] (61) at (-9.25, -3.75) {};
		\node [style=none] (62) at (-6.75, -3.75) {};
		\node [style=none] (63) at (-8, -0.75) {};
		\node [style=none] (64) at (-8, -1.75) {};
		\node [style=none] (65) at (-9, -1.25) {$\cl_{\chi}$};
		\node [style=state] (67) at (-8, -2.75) {$\psi$};
		\node [style=none] (69) at (-8, 6.5) {};
		\node [style=none] (70) at (1, 5) {};
		\node [style=none] (71) at (9, 5) {};
		\node [style=none] (72) at (1, -4.25) {};
		\node [style=none] (73) at (9, -4.25) {};
		\node [style=none] (74) at (1.5, 4) {};
		\node [style=none] (75) at (8.5, 4) {};
		\node [style=none] (76) at (1.5, -3.75) {};
		\node [style=none] (77) at (8.5, -3.75) {};
		\node [style=none] (78) at (5, 5) {};
		\node [style=none] (79) at (5, 4) {};
		\node [style=none] (80) at (2.25, 4.5) {$\cl_{\chi}$};
		\node [style=none] (81) at (5, 6.5) {};
		\node [style=none] (82) at (-10.25, 5.5) {$\frac{1}{\gamma}$};
		\node [style=none] (83) at (2.5, 3) {$\frac{1}{\delta}$};
	\end{pgfonlayer}
	\begin{pgfonlayer}{edgelayer}
		\draw [in=90, out=0, looseness=1.25] (32) to (34);
		\draw [in=0, out=-90, looseness=1.25] (34) to (31);
		\draw [in=270, out=180, looseness=1.25] (31) to (33.center);
		\draw [in=180, out=90, looseness=1.25] (33.center) to (32);
		\draw (31) to (15);
		\draw [in=90, out=0, looseness=1.25] (38) to (40);
		\draw [in=0, out=-90, looseness=1.25] (40) to (37);
		\draw [in=270, out=180, looseness=1.25] (37) to (39.center);
		\draw [in=180, out=90, looseness=1.25] (39.center) to (38);
		\draw (55.center) to (56.center);
		\draw (59.center) to (60.center);
		\draw (61.center) to (62.center);
		\draw (62.center) to (60.center);
		\draw (59.center) to (61.center);
		\draw (55.center) to (57.center);
		\draw (57.center) to (58.center);
		\draw (58.center) to (56.center);
		\draw (37) to (63.center);
		\draw (64.center) to (67);
		\draw (70.center) to (71.center);
		\draw (74.center) to (75.center);
		\draw (76.center) to (77.center);
		\draw (77.center) to (75.center);
		\draw (74.center) to (76.center);
		\draw (70.center) to (72.center);
		\draw (72.center) to (73.center);
		\draw (73.center) to (71.center);
		\draw (81.center) to (78.center);
		\draw (79.center) to (32);
		\draw (69.center) to (38);
	\end{pgfonlayer}
\end{tikzpicture} 
\end{equation}
with $\gamma = \lv (\id \tc m)\ml_{\chi}(\ket{\psi}) \rv$ and $\delta = \lv (\id \tc m)\ket{\psi} \rv$.
\end{itemize}
Where we have represented for any extensions $\chi$ of $\nu$, the map $\cl_{\chi}$ as
\begin{equation}
 \begin{tikzpicture}
	\begin{pgfonlayer}{nodelayer}
		\node [style=none] (0) at (-11.5, 4) {};
		\node [style=none] (1) at (-5.5, 4) {};
		\node [style=none] (2) at (-11.5, -3) {};
		\node [style=none] (3) at (-5.5, -3) {};
		\node [style=none] (4) at (-10.5, 2) {};
		\node [style=none] (5) at (-6.5, 2) {};
		\node [style=none] (6) at (-10.5, -2) {};
		\node [style=none] (7) at (-6.5, -2) {};
		\node [style=none] (8) at (-8.5, 4) {};
		\node [style=none] (9) at (-8.5, 2) {};
		\node [style=none] (10) at (-10.5, 3) {$\cl_{\chi}$};
		\node [style=none] (11) at (-3.5, 0) {$:$};
		\node [style=none] (12) at (-8.5, 5) {};
		\node [style=none] (13) at (-8.5, 1) {};
		\node [style=none] (14) at (-1, 2) {};
		\node [style=state] (15) at (-1, -2) {$\psi$};
		\node [style=none] (16) at (4, 0) {};
		\node [style=none] (17) at (6.5, 4) {};
		\node [style=none] (18) at (12.5, 4) {};
		\node [style=none] (19) at (6.5, -3) {};
		\node [style=none] (20) at (12.5, -3) {};
		\node [style=none] (21) at (7.5, 2) {};
		\node [style=none] (22) at (11.5, 2) {};
		\node [style=none] (23) at (7.5, -2) {};
		\node [style=none] (24) at (11.5, -2) {};
		\node [style=none] (25) at (9.5, 4) {};
		\node [style=none] (26) at (9.5, 2) {};
		\node [style=none] (27) at (7.5, 3) {$\cl_{\chi}$};
		\node [style=none] (28) at (9.5, 6) {};
		\node [style=none] (29) at (9.5, 1) {};
		\node [style=state] (30) at (9.5, -0.5) {$\psi$};
		\node [style=none] (31) at (1.75, 0) {};
		\node [style=none] (32) at (15, 0) {.};
	\end{pgfonlayer}
	\begin{pgfonlayer}{edgelayer}
		\draw (0.center) to (1.center);
		\draw (4.center) to (5.center);
		\draw (6.center) to (7.center);
		\draw (7.center) to (5.center);
		\draw (4.center) to (6.center);
		\draw (0.center) to (2.center);
		\draw (2.center) to (3.center);
		\draw (3.center) to (1.center);
		\draw (17.center) to (18.center);
		\draw (21.center) to (22.center);
		\draw (23.center) to (24.center);
		\draw (24.center) to (22.center);
		\draw (21.center) to (23.center);
		\draw (17.center) to (19.center);
		\draw (19.center) to (20.center);
		\draw (20.center) to (18.center);
		\draw (13.center) to (9.center);
		\draw (8.center) to (12.center);
		\draw (14.center) to (15);
		\draw (26.center) to (30);
		\draw (25.center) to (28.center);
		\draw [style=arrow] (31.center) to (16.center);
	\end{pgfonlayer}
\end{tikzpicture} 
\end{equation}

\end{definition}

First note that this definition does not assume any sort of linearity. On the contrary, the linearity will be derived from these properties  in Theorem \ref{linearity}. Let us now explain this definition. Let $\nu$ be a minimal splitting map, it represents a system without any environment. An evolution of this system should give rise to an evolution of all the systems that are extensions of $\nu$; these evolutions are described by the family $\ml_{\chi}$. Then, any $\zeta$ such that $\chi \sqsubseteq \zeta$ is describing a subsystem $\stloc_{\LL}(\zeta)$ of $\cl(\ch^{\chi})$ that is bigger than $\stloc_{\LL}(\chi)$. Equation \ref{extension} can thus be understood as asking that when acting on a state that is separable with respect to $\zeta$, the evolution $\ml_{\chi}$ should only act on the left part of the state which is already a bigger system than the one described by $\nu$. Moreover this local evolution should be the one associated to $\mu$ which separates $\zeta$ into $\chi$ on the left and $\bar{\chi} \cap \zeta$ on the right and is itself an extension of $\nu$ because $\cons_{\LL}(\mu) = \cons_{\LL}(\chi) = \cons_{\LL}(\nu)$ as stated in Remark \ref{good dots}. The two other postulates are more explicit. Indeed they are asking that for any $\chi$ embedding our initial system into a bigger one, applying $\ml_{\chi}$ should not affect the probabilities of any measurement done on the right of $\chi$, i.e. on the environment of the system we're evolving. And that our evolution should commute with the updates coming from measurement on the environment. 

Before proving our main characterisation, we prove a preliminary result, which establishes a form of consistency for locally-applicable transformations in the non-factor case that does not appear in the non-factor setting. Indeed, for non-factor systems, we must content with the possibility that two states might be inconsistent and so not tensor-able (represented by the isometric rather than unitary nature of splitting maps). In this case, we note that local-applicability entails that a transformation cannot, when applied locally, modify the tensor-ability, that is, the consistency, between two states. 
\begin{proposition}
\label{propositionnorms}
Consider the splitting maps $\chi, \zeta, \mu, \xi$ satisfying $\left(\left(\cdot \tm \cdot\right) \tz \cdot\right) = \left( \cdot \tc \left(\cdot \tx \cdot \right)\right)$ as in the premise of Equation \ref{extension} and suppose that $\cl_{\mu}$ satisfies Equation \ref{probas} for any measurement result $m$ on $\ch_{\R}^{\mu}$, then for all states $\ket{\phi}$ and $\ket{\psi}$
\be
\label{local norms}
\lv \ket{\phi} \tz \ket{\psi} \rv = \lv \ml_{\mu}(\ket{\phi}) \tz \ket{\psi} \rv
\ee
\end{proposition}

\begin{proof}
Because $\xi$ is a canonical splitting map, it is by definition of the form
\be
\xi :  \chrc = \bigoplus_i (\ch_{\LL}^{i} \otimes \ch_{\R}^{i}) \hookrightarrow (\bigoplus_i \ch_{\LL}^{i}) \otimes (\bigoplus_i  \ch_{\R}^{i}) = \chlx \otimes \chrx \, .
\ee
For all $i$, we define by $\pi^{i}$ the orthogonal projector of $\ch^{\xi} = \chrc = \bigoplus_i (\ch_{\LL}^{i} \otimes \ch_{\R}^{i})$ on the subspace $\ch_{\LL}^{i} \otimes \ch_{\R}^{i}$. The $\pi^{i}$ for a family of pairwise orthogonal projectors that sum to the identity (of $\chrc$). Moreover each $\pi_i$ can be written as $\pi^{i} = \pi_{\LL}^{i} \tx \pi_{\R}^{i}$ where $\pi_{\LL}^{i}$ is the orthogonal projector of $\bigoplus_i \ch_{\LL}^{i}$ onto $\ch_{\LL}^{i}$ and symmetrically for $\pi_{\R}^{i}$. We can then compute  
\be
\begin{split}
\| \ml_{\mu}(\ket{\phi}) \tz \ket{\psi} \| & = \| (\id \tc \sum_i \pi_i)(\ml_{\mu}(\ket{\phi}) \tz \ket{\psi}) \| \\
& = \| \sum_i  (\id \tc \pi_i)(\ml_{\mu}(\ket{\phi}) \tz \ket{\psi}) \| \\
& = \sqrt{\sum_i \|(\id \tc \pi_i)(\ml_{\mu}(\ket{\phi}) \tz \ket{\psi}) \|^{2}}\\
& = \sqrt{\sum_i \|(\id \tc (\pi_i^{l} \tx \pi_i^{r}))(\ml_{\mu}(\ket{\phi}) \tz \ket{\psi}) \|^{2}}\\
& = \sqrt{\sum_i \|((\id \tm \pi_i^{l}) \tz \pi_i^{r})(\ml_{\mu}(\ket{\phi}) \tz \ket{\psi}) \|^{2}}\\
& = \sqrt{\sum_i \|(\id \tm \pi_i^{l})\ml_{\mu}(\ket{\phi}) \tz \pi_i^{r} \ket{\psi}) \|^{2}}\\
& = \sqrt{\sum_i \|(\id \tm \pi_i^{l})\ml_{\mu}(\ket{\phi}) \|^{2} \| \pi_i^{r} \ket{\psi}) \|^{2}}\\
& = \sqrt{\sum_i \|(\id \tm \pi_i^{l})\ket{\phi} \|^{2} \| \pi_i^{r} \ket{\psi}) \|^{2}}\\
& = \sqrt{\sum_i \|(\id \tm \pi_i^{l})\ket{\phi} \tz \pi_i^{r} \ket{\psi}) \|^{2}}\\
& = \| \ket{\phi} \tz \ket{\psi} \| \, ,
\end{split}
\ee
where the third equality is the Pythagorean theorem, the sixth equality comes from the fact that $(\id \tm \pi_{\LL}^{i})$ and $\pi_{\R}^{i}$ are respectively left $\zeta$-consistent and right $\zeta$-consistent, the seventh from the fact that $\pi_{\LL}^{i} \tx \pi_{\R}^{j} = \delta_{ij} \pi^{i}$ and that $\zeta$ is an isometry, the eighth comes from Equation \ref{probas} and the last comes from applying all the previous steps in the inverse direction. 
\end{proof}

\begin{theorem}
\label{linearity}
A family of maps $(\cl_{\chi})$ indexed by the extensions $\chi$ of $\nu$, forms a local evolution associated to the splitting map $\nu : \ch_{\LL} \rightarrow \ch_{\LL} \otimes \ch_{\R}$ if and only if there exists a linear map $L \in \cl(\ch_{\LL})$, which is moreover a unitary element of $\stloc_{\LL}(\nu)$, such that for every extension $\chi$ of $\nu$, $\cl_{\chi} = L \tc \id$. 
\end{theorem}

\begin{proof}
Let $\nu$ be a canonical splitting map and let us prove that a local evolution associated to $\nu$ is a family of linear maps of the form $\cl_{\chi} = L \tc \id$. Let $\chi$ be a splitting map extending $\nu$ and let us prove that $\cl_{\chi}$ is linear. We remind that $\nu$ and $\chi$ are of the form 
\begin{equation}
\nu : \ch_{\LL} = \bigoplus_{i=1}^N (\ch_{\LL}^{i} \otimes \mathbb{C}) \hookrightarrow (\bigoplus_{i=1}^N  \ch_{\LL}^{i}) \otimes (\bigoplus_{i=1}^N   \mathbb{C}) = \ch_{\LL} \otimes \mathbb{C}^{N} \, , 
\end{equation}
\begin{equation}
\chi : \ch = \bigoplus_{i=1}^N (\ch_{\LL}^{i} \otimes \ch_{\R}^{i}) \hookrightarrow (\bigoplus_{i=1}^N  \ch_{\LL}^{i}) \otimes (\bigoplus_{i=1}^N   \ch_{\R}^{i}) = \ch_{\LL} \otimes \ch_{\R} \, .
\end{equation}
Let us now define the family of splitting maps $\chi_{k \cdots l / l+1 \cdots m}$, for $k \leq l \leq m \in \mathbb{Z}^{*}$, by 
\begin{equation}
\chi_{k \cdots l / l+1 \cdots m} : \bigoplus_{i=1}^N \left( \ch_{k}^{i} \otimes \cdots \otimes \ch_{m}^{i} \right) \hookrightarrow \left( \bigoplus_{i=1}^N \ch_{k}^{i} \otimes \cdots \otimes \ch_{l}^{i} \right) \otimes \left( \bigoplus_{i=1}^N \ch_{l+1}^{i} \otimes \cdots \otimes \ch_{m}^{i} \right)
\end{equation}
where $\ch_{j}^{i} = \ch_{\LL}^{i}$ when $j$ is odd and $\ch_{j}^{i} = \ch_{\R}^{i}$ when $j$ is even. We will write $\chi_{/l \cdots m}$ when $k=l$ and $\chi_{k \cdots l/ }$ when $l=m$. The $\chi_{k \cdots l / l+1 \cdots m}$ form a family of  canonical splitting maps such that :
\begin{itemize}
\item $\chi_{1/} = \nu$,
\item $\chi_{1/2} = \chi$,
\item $\chi_{1/2 \cdots l}$ is an extension of $\nu$ for all $l \geq 1$,
\item $\chi_{k \cdots l/l+1 \cdots m} \sqsubseteq \chi_{k \cdots l'/l' +1 \cdots m}$ if and only if $k \leq l \leq l' \leq m$ and this comprehension can be witnessed through the isometries $\mu = \chi_{k \cdots l /l+1 \cdots l'}$ and $\xi = \chi_{l \cdots l' / l'+1 \cdots m}$.
\end{itemize}
Let us call for every $1 \leq i \leq N$, $\ch^{i} = \ch_{\LL}^{i} \otimes \ch_{\R}^{i}$ and $d_i = \dim\left(\ch^{i}\right)$. Let $\{ \ket{j} \}_{1 \leq j \leq d_i}$ be an orthonormal basis of $\ch^{i}$ and define $\ket{\phi_{i}} = \sum_{j=1}^{d_i} \frac{1}{\sqrt{d_i}} \ket{jj}$, $\ket{\phi_{+}} = \sum_{i = 1}^{N} \frac{1}{\sqrt{N}} \ket{\phi_{i}}$ as well as $\ket{\phi} = \sum_{i = 1}^{N} d_i\sqrt{N}\ket{\phi_{i}}$ which is an unormalised state. For any state 
\be
\ket{\psi} = \sum_{k=1}^N \sum_{j=1}^{d_k} \alpha_j \ket{j} \in \ch = \bigoplus_i \ch_{\LL}^{i} \otimes \ch_{\R}^{i}\, , 
\ee
of norm one, we denote 
\be
\beta = \left\Vert \ket{\phi_+} \tc_{1234/56} \ket{\psi} \right\Vert
\ee
and 
\be
\begin{split}
\gamma & = \left\Vert \left (\id \tc_{12/3456 }\sum_{k = 1}^N \ketbra{\phi_k}{\phi_k} \right) \frac{\ket{\phi_+} \tc_{1234/56} \ket{\psi}}{\left\Vert \ket{\phi_+} \tc_{1234/56} \ket{\psi} \right\Vert}  \right\Vert   \\
& = \frac{1}{\beta} \left\Vert \left (\id \tc_{12/3456 }\sum_{k = 1}^N \ketbra{\phi_k}{\phi_k} \right)\ket{\phi_+} \tc_{1234/56} \ket{\psi} \right\Vert \, .
\end{split}
\ee
Note that, by proposition \ref{propositionnorms}, $\beta$ is also equal to 
\be
\beta = \left\Vert \ml_{\chi_{1/234}}(\ket{\phi_+}) \tc_{1234/56} \ket{\psi} \right\Vert
\ee
and that, by equation \ref{probas}, $\gamma$ is also equal to
\be
\begin{split}
\gamma & = \left\Vert \left (\id \tc_{12/3456 }\sum_{k = 1}^N \ketbra{\phi_k}{\phi_k} \right) \frac{\ket{\phi_+} \tc_{1234/56} \ket{\psi}}{\left\Vert \ket{\phi_+} \tc_{1234/56} \ket{\psi} \right\Vert}  \right\Vert \\
& = \left\Vert \left (\id \tc_{1/23456 } \left( \id \tc_{2/3456} \sum_{k = 1}^N \ketbra{\phi_k}{\phi_k} \right) \right) \frac{\ket{\phi_+} \tc_{1234/56} \ket{\psi}}{\left\Vert \ket{\phi_+} \tc_{1234/56} \ket{\psi} \right\Vert}  \right\Vert \\
& = \left\Vert \left (\id \tc_{1/23456 } \left( \id \tc_{2/3456} \sum_{k = 1}^N \ketbra{\phi_k}{\phi_k} \right) \right)\ml_{\chi_{1/23456}} \left(  \frac{\ket{\phi_+} \tc_{1234/56} \ket{\psi}}{\left\Vert \ket{\phi_+} \tc_{1234/56} \ket{\psi} \right\Vert} \right) \right\Vert \\
& = \left\Vert \left (\id \tc_{12/3456 }\sum_{k = 1}^N \ketbra{\phi_k}{\phi_k} \right) \ml_{\chi_{1/23456}} \left(  \frac{\ket{\phi_+} \tc_{1234/56} \ket{\psi}}{\left\Vert \ket{\phi_+} \tc_{1234/56} \ket{\psi} \right\Vert}  \right) \right\Vert \, .
\end{split}
\ee
Note that whichever $\ket{\psi}$ has been chosen, $\beta$ and $\gamma$ are non-zero and we can then compute that 
\be
\begin{split}
&\frac{1}{\beta \gamma} \left( \id \tc_{12/3456 }\sum_{k = 1}^N \ketbra{\phi_k}{\phi_k} \right) \ml_{\chi_{1/234}}\left(\ket{\phi_+}\right) \tc_{1234/56} \ket{\psi} \\
= & \frac{1}{\gamma} \left( \id \tc_{12/3456 }\sum_{k = 1}^N \ketbra{\phi_k}{\phi_k} \right) \ml_{\chi_{1/23456}} \left( \frac{\ket{\phi_+} \tc_{1234/56} \ket{\psi}}{\beta} \right) \\
= & \ml_{\chi_{1/23456}}  \left( \frac{1}{\beta \gamma} \left( \id \tc_{12/3456 }\sum_{k = 1}^N \ketbra{\phi_k}{\phi_k} \right) \ket{\phi_+} \tc_{1234/56} \ket{\psi}\right) \\
= & \ml_{\chi_{1/23456}}\left( \frac{1}{\beta \gamma} \left( \id \tc_{12/3456 }\sum_{k = 1}^N \ketbra{\phi_k}{\phi_k} \right) \sum_{i = 1}^N \sum_{j = 1}^{d_i} \frac{1}{\sqrt{Nd_i}} \left(\ket{j} \tc_{12/34} \ket{j}\right) \tc_{1234/56} \ket{\psi}\right) \\
= & \ml_{\chi_{1/23456}}\left(\frac{1}{\beta \gamma}\left(\id \tc_{12/3456} \sum_{k = 1}^N \ketbra{\phi_k}{\phi_k}\right) \sum_{i = 1}^N \sum_{j = 1}^{d_i} \frac{1}{\sqrt{Nd_i}} \ket{j} \tc_{12/3456} \left(\ket{j} \tc_{34/56} \ket{\psi}\right)\right) \\
= & \ml_{\chi_{1/23456}}\left(\frac{1}{\beta \gamma} \sum_{i = 1}^N \sum_{j = 1}^{d_i} \frac{1}{\sqrt{Nd_i}} \ket{j} \tc_{12/3456} \left(\sum_{k = 1}^N \ketbra{\phi_k}{\phi_k} \left(\ket{j} \tc_{34/56} \ket{\psi}\right)\right)\right)\\
= & \ml_{\chi_{1/23456}}\left(\frac{1}{\beta \gamma} \sum_{i = 1}^N \sum_{j = 1}^{d_i} \frac{1}{\sqrt{Nd_i}} \ket{j} \tc_{12/3456} \left(\sum_{k = 1}^N  \ket{\phi_k} \sum_{l = 1}^{d_k} \frac{1}{d_k}\left(\bra{l} \tc_{34/56} \bra{l}\right) \left(\ket{j} \tc_{34/56} \ket{\psi}\right)\right)\right)\\
= & \ml_{\chi_{1/23456}}\left(\frac{1}{\beta \gamma} \sum_{i = 1}^N \sum_{j = 1}^{d_i} \frac{1}{\sqrt{Nd_i}} \ket{j} \tc_{12/3456} \left(\sum_{k = 1}^N \frac{\alpha_j}{d_k} \ket{\phi_k}\right)\right)\\
= & \ml_{\chi_{1/23456}}\left( \frac{1}{\beta \gamma} \sum_{i = 1}^N \sum_{j = 1}^{d_i} \alpha_j \ket{j} \tc_{12/3456} \left(\sum_{k = 1}^N \frac{1}{d_k\sqrt{N}} \ket{\phi_k}\right)\right)\\
= & \ml_{\chi_{1/23456}}\left(\frac{1}{\beta \gamma} \ket{\psi} \tc_{12/3456} \left(\sum_{k = 1}^N \frac{1}{d_k\sqrt{N}} \ket{\phi_k}\right)\right)\\ 
= & \frac{1}{\beta \gamma} \ml_{\chi_{1/2}}\left(\ket{\psi}\right) \tc_{12/3456} \left(\sum_{k = 1}^N \frac{1}{d_k\sqrt{N}} \ket{\phi_k} \right)
\end{split}
\ee
where the first equality comes from applying Equation \ref{extension} and the second from applying Equation \ref{update} to the right $\chi_{1/23456}$-local measurement 
\be
\id \tc_{12/3456} m = \id \tc_{12/3456} \sum_{k=1}^N \ketbra{\phi_k}{\phi_k} = \id \tc_{1/23456} \sum_{k=1}^N (\id \tc_{2/3456} \ketbra{\phi_k}{\phi_k}) \, .
\ee
The next equalities are simple computations allowing us to rewrite the state on which $\ml_{\chi_{1/23456}}$ is applied as a different $\chi$-product. And the last equality comes from applying Equation \ref{extension} to the normalised state 
\be
\begin{split}
& \frac{1}{\beta \gamma} \ket{\psi} \tc_{12/3456} \left(\sum_{k = 1}^N \frac{1}{d_k\sqrt{N}} \ket{\phi_k}\right)\\
= & \frac{\lambda}{\beta \gamma} \ket{\psi} \tc_{12/3456} \left(\frac{1}{\lambda} \sum_{k = 1}^N \frac{1}{d_k\sqrt{N}} \ket{\phi_k}\right)
\end{split}
\ee
where
\be
\lambda = \left\Vert \sum_{k = 1}^N \frac{1}{d_k\sqrt{N}} \ket{\phi_k} \right\Vert = \sqrt{\sum_{k=1}^{N}\frac{1}{Nd_k^{2}}}
\ee
and then remarking that, by proposition \ref{propositionnorms},
\be
\begin{split}
& \left\Vert \ml_{\chi_{1/2}}\left(\ket{\psi}\right) \tc_{12/3456} \left( \frac{1}{\lambda} \sum_{k = 1}^N \frac{1}{d_k\sqrt{N}} \ket{\phi_k} \right) \right\Vert \\
= & \left\Vert \ket{\psi} \tc_{12/3456} \left( \frac{1}{\lambda} \sum_{k = 1}^N \frac{1}{d_k\sqrt{N}} \ket{\phi_k} \right) \right\Vert \\
= & \frac{\beta \gamma}{\lambda} \, .
\end{split}
\ee
It follows that for any (normalised) state $\ket{\psi} = \sum_{k=1}^{N} \ket{\psi_k} \in \ch$, 
\be
\begin{split}
& \left( \id \tc_{12/3456 }\sum_{k = 1}^N \ketbra{\phi_k}{\phi_k} \right) \ml_{\chi_{1/234}}\left(\ket{\phi_+}\right) \tc_{1234/56} \ket{\psi} \\
= & \ml_{\chi_{1/2}}\left(\ket{\psi}\right) \tc_{12/3456} \left(\sum_{k = 1}^N \frac{1}{d_k\sqrt{N}} \ket{\phi_k} \right)
\end{split}
\ee
which allows us to compute that for any state $\ket{\eta} = \sum_{k=1}^{N} \ket{\eta_k} \in \ch$,
\be
\begin{split}
&  \left(\bra{\eta} \tc_{12/3456} \bra{\phi}\right) \ml_{\chi_{1/234}}\left(\ket{\phi_+}\right) \tc_{1234/56} \ket{\psi} \\
= & \left(\bra{\eta} \tc_{12/3456} \left( \bra{\phi} \sum_{k=1}^N \ketbra{\phi_k}{\phi_k} \right) \right) \ml_{\chi_{1/234}}\left(\ket{\phi_+}\right) \tc_{1234/56} \ket{\psi} \\
= & \left(\bra{\eta} \tc_{12/3456} \bra{\phi}\right)\left( \id \tc_{12/3456 }\sum_{k = 1}^N \ketbra{\phi_k}{\phi_k} \right) \ml_{\chi_{1/234}}\left(\ket{\phi_+}\right) \tc_{1234/56} \ket{\psi} \\
= &  \left(\bra{\eta} \tc_{12/3456} \bra{\phi}\right) \ml_{\chi_{1/2}}\left(\ket{\psi}\right) \tc_{12/3456} \left(\sum_{k = 1}^N \frac{1}{d_k\sqrt{N}} \ket{\phi_k} \right) \\
= & \left( \sum_{k=1}^{N} \bra{\eta_k} \otimes d_k\sqrt{N}\bra{\phi_k} \right) \left( \sum_{k=1}^{N}\ml_{\chi_{1/2}}\left(\ket{\psi}\right)_k \otimes \frac{1}{d_k\sqrt{N}} \ket{\phi_k} \right) \\
= & \sum_{k=1}^{N} \bra{\eta_k}\ml_{\chi_{1/2}}\left(\ket{\psi}\right)_k \braket{\phi_k}{\phi_k} \\
= & \bra{\eta} \ml_{\chi_{1/2}}\left(\ket{\psi}\right) \, .
\end{split}
\ee
This whole computation could also be seen diagramatically as 
\be
\scalebox{0.9}{\begin{tikzpicture}
	\begin{pgfonlayer}{nodelayer}
		\node [style=none] (0) at (-12, -4.75) {};
		\node [style=none] (1) at (-8, -4.75) {};
		\node [style=none] (2) at (-12, -8.75) {};
		\node [style=none] (3) at (-8, -8.75) {};
		\node [style=none] (4) at (-11.25, -5.75) {};
		\node [style=none] (5) at (-8.75, -5.75) {};
		\node [style=none] (6) at (-11.25, -8.25) {};
		\node [style=none] (7) at (-8.75, -8.25) {};
		\node [style=none] (8) at (-10, -4.75) {};
		\node [style=none] (9) at (-10, -5.75) {};
		\node [style=none] (10) at (-10.5, -5.25) {$\cl_{\chi_{1/234}}$};
		\node [style=white] (14) at (-6, -1.75) {$\chi_{1234/56}$};
		\node [style=state] (15) at (-2, -4.75) {$\psi$};
		\node [style=state] (16) at (-10, -7) {$\phi_{+}$};
		\node [style=white] (17) at (-6, 3) {$\chi_{12/3456}$};
		\node [style=effect] (18) at (-10, 6) {$\eta$};
		\node [style=effect] (19) at (-2, 6) {$\phi$};
		\node [style=none] (20) at (0, 0.5) {$=$};
		\node [style=none] (42) at (-14, -19.75) {$=$};
		\node [style=none] (43) at (-10.75, -20.25) {};
		\node [style=none] (44) at (-1.25, -20.25) {};
		\node [style=none] (45) at (-10.75, -28.25) {};
		\node [style=none] (46) at (-1.25, -28.25) {};
		\node [style=none] (47) at (-10, -21.25) {};
		\node [style=none] (48) at (-2, -21.25) {};
		\node [style=none] (49) at (-10, -27.75) {};
		\node [style=none] (50) at (-2, -27.75) {};
		\node [style=none] (51) at (-10, -20.25) {};
		\node [style=none] (53) at (-9.25, -20.75) {$\cl_{\chi_{1/23456}}$};
		\node [style=white] (54) at (-6, -23.25) {$\chi_{1234/56}$};
		\node [style=state] (55) at (-3, -26.25) {$\psi$};
		\node [style=state] (56) at (-9, -26.25) {$\phi_{+}$};
		\node [style=white] (57) at (-6, -17.25) {$\chi_{12/3456}$};
		\node [style=effect] (58) at (-10, -14.25) {$\eta$};
		\node [style=effect] (59) at (-2, -14.25) {$\phi_k$};
		\node [style=state] (60) at (-2, -12.5) {$\phi_k$};
		\node [style=none] (61) at (-3.25, -13.25) {$\sum_k$};
		\node [style=effect] (62) at (-2, -10.75) {$\phi$};
		\node [style=none] (63) at (-6, -21.25) {};
		\node [style=none] (64) at (-6, -20.25) {};
		\node [style=none] (65) at (3.25, 0) {};
		\node [style=none] (66) at (12.75, 0) {};
		\node [style=none] (67) at (3.25, -8) {};
		\node [style=none] (68) at (12.75, -8) {};
		\node [style=none] (69) at (4, -1) {};
		\node [style=none] (70) at (12, -1) {};
		\node [style=none] (71) at (4, -7.5) {};
		\node [style=none] (72) at (12, -7.5) {};
		\node [style=none] (73) at (4, 0) {};
		\node [style=none] (74) at (4.75, -0.5) {$\cl_{\chi_{1/23456}}$};
		\node [style=white] (75) at (8, -3) {$\chi_{1234/56}$};
		\node [style=state] (76) at (11, -6) {$\psi$};
		\node [style=state] (77) at (5, -6) {$\phi_{+}$};
		\node [style=none] (78) at (8, -1) {};
		\node [style=none] (79) at (8, 0) {};
		\node [style=white] (80) at (8, 3) {$\chi_{12/3456}$};
		\node [style=effect] (81) at (4, 6) {$\eta$};
		\node [style=effect] (82) at (12, 6) {$\phi$};
		\node [style=none] (83) at (1, -23) {$=$};
		\node [style=none] (84) at (4.25, -27.5) {};
		\node [style=none] (85) at (13.75, -27.5) {};
		\node [style=none] (86) at (4.25, -35.5) {};
		\node [style=none] (87) at (13.75, -35.5) {};
		\node [style=none] (88) at (5, -28.5) {};
		\node [style=none] (89) at (13, -28.5) {};
		\node [style=none] (90) at (5, -35) {};
		\node [style=none] (91) at (13, -35) {};
		\node [style=none] (92) at (5, -27.5) {};
		\node [style=none] (93) at (5.75, -28) {$\cl_{\chi_{1/23456}}$};
		\node [style=white] (94) at (9, -30.5) {$\chi_{1234/56}$};
		\node [style=state] (95) at (12, -33.5) {$\psi$};
		\node [style=state] (96) at (6, -33.5) {$\phi_{+}$};
		\node [style=white] (97) at (9, -25.25) {$\chi_{12/3456}$};
		\node [style=effect] (98) at (5, -11) {$\eta$};
		\node [style=effect] (99) at (13, -22.25) {$\phi_k$};
		\node [style=state] (100) at (13, -20.5) {$\phi_k$};
		\node [style=none] (101) at (11.75, -21.25) {$\sum_k$};
		\node [style=effect] (102) at (13, -11) {$\phi$};
		\node [style=none] (103) at (9, -28.5) {};
		\node [style=none] (104) at (9, -27.5) {};
		\node [style=white] (105) at (9, -17.75) {$\chi_{12/3456}$};
		\node [style=white] (106) at (9, -14) {$\chi_{12/3456}$};
		\node [style=none] (107) at (5, -22.25) {};
		\node [style=none] (108) at (5, -20.5) {};
		\node [style=none] (109) at (-12, 4) {$\frac{1}{\beta \gamma}$};
		\node [style=none] (110) at (2, 4) {$\frac{1}{\gamma}$};
		\node [style=none] (111) at (5.25, -2) {$\frac{1}{\beta}$};
		\node [style=none] (112) at (-11, -12) {$\frac{1}{\gamma}$};
		\node [style=none] (113) at (-8.75, -22.25) {$\frac{1}{\beta}$};
		\node [style=none] (114) at (3.25, -12.5) {$\frac{1}{\gamma}$};
		\node [style=none] (115) at (6.25, -29.5) {$\frac{1}{\beta}$};
	\end{pgfonlayer}
	\begin{pgfonlayer}{edgelayer}
		\draw (0.center) to (1.center);
		\draw (4.center) to (5.center);
		\draw (6.center) to (7.center);
		\draw (7.center) to (5.center);
		\draw (4.center) to (6.center);
		\draw (0.center) to (2.center);
		\draw (2.center) to (3.center);
		\draw (3.center) to (1.center);
		\draw (17) to (14);
		\draw [in=90, out=0, looseness=1.25] (14) to (15);
		\draw (9.center) to (16);
		\draw [in=-90, out=180, looseness=1.50] (17) to (18);
		\draw [in=-90, out=0, looseness=1.50] (17) to (19);
		\draw [in=90, out=-180, looseness=1.25] (14) to (8.center);
		\draw (43.center) to (44.center);
		\draw (47.center) to (48.center);
		\draw (49.center) to (50.center);
		\draw (50.center) to (48.center);
		\draw (47.center) to (49.center);
		\draw (43.center) to (45.center);
		\draw (45.center) to (46.center);
		\draw (46.center) to (44.center);
		\draw [in=90, out=0, looseness=1.25] (54) to (55);
		\draw [in=-90, out=180, looseness=1.50] (57) to (58);
		\draw [in=-90, out=0, looseness=1.50] (57) to (59);
		\draw (60) to (62);
		\draw [in=180, out=90, looseness=1.25] (56) to (54);
		\draw (57) to (64.center);
		\draw (63.center) to (54);
		\draw (65.center) to (66.center);
		\draw (69.center) to (70.center);
		\draw (71.center) to (72.center);
		\draw (72.center) to (70.center);
		\draw (69.center) to (71.center);
		\draw (65.center) to (67.center);
		\draw (67.center) to (68.center);
		\draw (68.center) to (66.center);
		\draw [in=90, out=0, looseness=1.25] (75) to (76);
		\draw [in=180, out=90, looseness=1.25] (77) to (75);
		\draw (78.center) to (75);
		\draw [in=-90, out=180, looseness=1.50] (80) to (81);
		\draw [in=-90, out=0, looseness=1.50] (80) to (82);
		\draw (80) to (79.center);
		\draw (84.center) to (85.center);
		\draw (88.center) to (89.center);
		\draw (90.center) to (91.center);
		\draw (91.center) to (89.center);
		\draw (88.center) to (90.center);
		\draw (84.center) to (86.center);
		\draw (86.center) to (87.center);
		\draw (87.center) to (85.center);
		\draw [in=90, out=0, looseness=1.25] (94) to (95);
		\draw [in=-90, out=0, looseness=1.50] (97) to (99);
		\draw [in=180, out=90, looseness=1.25] (96) to (94);
		\draw (97) to (104.center);
		\draw (103.center) to (94);
		\draw (105) to (106);
		\draw [in=-90, out=180, looseness=1.25] (106) to (98);
		\draw [in=-90, out=0, looseness=1.25] (106) to (102);
		\draw [in=90, out=0, looseness=1.25] (105) to (100);
		\draw [in=-180, out=90] (108.center) to (105);
		\draw (108.center) to (107.center);
		\draw [in=180, out=-90, looseness=1.25] (107.center) to (97);
	\end{pgfonlayer}
\end{tikzpicture}}
\ee
\be
\scalebox{0.9}{\begin{tikzpicture}
	\begin{pgfonlayer}{nodelayer}
		\node [style=none] (84) at (-12.75, 12) {};
		\node [style=none] (85) at (-1.25, 12) {};
		\node [style=none] (86) at (-12.75, -8) {};
		\node [style=none] (87) at (-1.25, -8) {};
		\node [style=none] (88) at (-12, 11) {};
		\node [style=none] (89) at (-2, 11) {};
		\node [style=none] (90) at (-12, -7.5) {};
		\node [style=none] (91) at (-2, -7.5) {};
		\node [style=none] (92) at (-12, 12) {};
		\node [style=none] (93) at (-11.25, 11.5) {$\cl_{\chi_{1/23456}}$};
		\node [style=white] (94) at (-7, -2.75) {$\chi_{1234/56}$};
		\node [style=state] (95) at (-4, -5.75) {$\psi$};
		\node [style=state] (96) at (-10, -5.75) {$\phi_{+}$};
		\node [style=white] (97) at (-7, 1.5) {$\chi_{12/3456}$};
		\node [style=effect] (98) at (-11, 17) {$\eta$};
		\node [style=effect] (99) at (-3, 4.5) {$\phi_k$};
		\node [style=state] (100) at (-3, 6.25) {$\phi_k$};
		\node [style=none] (101) at (-4.25, 5.5) {$\sum_k$};
		\node [style=effect] (102) at (-3, 17) {$\phi$};
		\node [style=none] (104) at (-7, -2.75) {};
		\node [style=white] (105) at (-7, 9) {$\chi_{12/3456}$};
		\node [style=white] (106) at (-7, 14) {$\chi_{12/3456}$};
		\node [style=none] (107) at (-11, 4.5) {};
		\node [style=none] (108) at (-11, 6.25) {};
		\node [style=none] (109) at (-7, 12) {};
		\node [style=none] (110) at (-7, 11) {};
		\node [style=none] (111) at (4.25, 12) {};
		\node [style=none] (112) at (18.75, 12) {};
		\node [style=none] (113) at (4.25, -8) {};
		\node [style=none] (114) at (18.75, -8) {};
		\node [style=none] (115) at (5, 11) {};
		\node [style=none] (116) at (18, 11) {};
		\node [style=none] (117) at (5, -7.5) {};
		\node [style=none] (118) at (18, -7.5) {};
		\node [style=none] (119) at (5, 12) {};
		\node [style=none] (120) at (5.75, 11.5) {$\cl_{\chi_{1/23456}}$};
		\node [style=state] (122) at (16.5, -5.75) {$\psi$};
		\node [style=state] (123) at (9, -5.75) {$\phi_{+}$};
		\node [style=effect] (125) at (7, 17) {$\eta$};
		\node [style=effect] (126) at (14, 4.5) {$\phi_k$};
		\node [style=state] (127) at (14, 6.25) {$\phi_k$};
		\node [style=none] (128) at (12.75, 5.5) {$\sum_k$};
		\node [style=effect] (129) at (15, 17) {$\phi$};
		\node [style=white] (131) at (11, 9) {$\chi_{12/3456}$};
		\node [style=white] (132) at (11, 14) {$\chi_{12/3456}$};
		\node [style=none] (133) at (6, 1) {};
		\node [style=none] (134) at (6, 3.25) {};
		\node [style=none] (135) at (11, 12) {};
		\node [style=none] (136) at (11, 11) {};
		\node [style=none] (137) at (1.5, 5) {$=$};
		\node [style=white] (138) at (9, -3) {$\chi_{12/34}$};
		\node [style=white] (139) at (14, 2) {$\chi_{34/56}$};
		\node [style=none] (140) at (-13.5, -15) {};
		\node [style=none] (141) at (-4, -15) {};
		\node [style=none] (142) at (-13.5, -22.5) {};
		\node [style=none] (143) at (-4, -22.5) {};
		\node [style=none] (144) at (-12.75, -16) {};
		\node [style=none] (145) at (-4.75, -16) {};
		\node [style=none] (146) at (-12.75, -22) {};
		\node [style=none] (147) at (-4.75, -22) {};
		\node [style=none] (148) at (-12.75, -15) {};
		\node [style=none] (149) at (-12, -15.5) {$\cl_{\chi_{1/23456}}$};
		\node [style=state] (150) at (-11.75, -21) {$\psi$};
		\node [style=effect] (152) at (-12.75, -10) {$\eta$};
		\node [style=state] (154) at (-5.75, -20.75) {$\phi_k$};
		\node [style=none] (155) at (-8.25, -20.75) {$\sum_k\frac{1}{d_k \sqrt{N}}$};
		\node [style=effect] (156) at (-4.75, -10) {$\phi$};
		\node [style=white] (157) at (-8.75, -18) {$\chi_{12/3456}$};
		\node [style=white] (158) at (-8.75, -13) {$\chi_{12/3456}$};
		\node [style=none] (161) at (-8.75, -15) {};
		\node [style=none] (162) at (-8.75, -16) {};
		\node [style=none] (163) at (-15, -16) {$=$};
		\node [style=none] (164) at (-1, -20.75) {};
		\node [style=none] (165) at (2.5, -20.75) {};
		\node [style=none] (166) at (-1, -24.25) {};
		\node [style=none] (167) at (2.5, -24.25) {};
		\node [style=none] (168) at (-0.25, -21.75) {};
		\node [style=none] (169) at (1.75, -21.75) {};
		\node [style=none] (170) at (-0.25, -23.75) {};
		\node [style=none] (171) at (1.75, -23.75) {};
		\node [style=none] (172) at (-0.25, -20.75) {};
		\node [style=none] (173) at (0.5, -21.25) {$\cl_{\chi_{1/2}}$};
		\node [style=state] (174) at (0.75, -22.75) {$\psi$};
		\node [style=effect] (175) at (0.75, -10) {$\eta$};
		\node [style=state] (176) at (8.75, -20.75) {$\phi_k$};
		\node [style=none] (177) at (6.25, -20.75) {$\sum_k\frac{1}{d_k \sqrt{N}}$};
		\node [style=effect] (178) at (8.75, -10) {$\phi$};
		\node [style=white] (179) at (4.75, -17.25) {$\chi_{12/3456}$};
		\node [style=white] (180) at (4.75, -13) {$\chi_{12/3456}$};
		\node [style=none] (183) at (0.75, -21.75) {};
		\node [style=none] (184) at (0.75, -20.75) {};
		\node [style=none] (185) at (-1.75, -16) {$=$};
		\node [style=none] (186) at (13, -16) {};
		\node [style=none] (187) at (16.5, -16) {};
		\node [style=none] (188) at (13, -20) {};
		\node [style=none] (189) at (16.5, -20) {};
		\node [style=none] (190) at (13.75, -17) {};
		\node [style=none] (191) at (15.75, -17) {};
		\node [style=none] (192) at (13.75, -19.5) {};
		\node [style=none] (193) at (15.75, -19.5) {};
		\node [style=none] (194) at (13.75, -16) {};
		\node [style=none] (195) at (14.5, -16.5) {$\cl_{\chi_{1/2}}$};
		\node [style=state] (196) at (14.75, -18.5) {$\psi$};
		\node [style=effect] (197) at (14.75, -14) {$\eta$};
		\node [style=none] (205) at (14.75, -17) {};
		\node [style=none] (206) at (14.75, -16) {};
		\node [style=none] (207) at (10.75, -16) {$=$};
		\node [style=none] (208) at (-15, 5) {$=$};
		\node [style=none] (209) at (-10.5, 10) {$\frac{1}{\beta \gamma}$};
		\node [style=none] (210) at (7, 9.75) {$\frac{1}{\beta \gamma}$};
		\node [style=none] (211) at (12.5, -13.75) {$\frac{1}{\beta \gamma}$};
		\node [style=none] (212) at (1, -15) {$\frac{1}{\beta \gamma}$};
		\node [style=none] (213) at (-11.5, -17) {$\frac{1}{\beta \gamma}$};
	\end{pgfonlayer}
	\begin{pgfonlayer}{edgelayer}
		\draw (84.center) to (85.center);
		\draw (88.center) to (89.center);
		\draw (90.center) to (91.center);
		\draw (91.center) to (89.center);
		\draw (88.center) to (90.center);
		\draw (84.center) to (86.center);
		\draw (86.center) to (87.center);
		\draw (87.center) to (85.center);
		\draw [in=90, out=0, looseness=1.25] (94) to (95);
		\draw [in=-90, out=0, looseness=1.50] (97) to (99);
		\draw [in=180, out=90, looseness=1.25] (96) to (94);
		\draw (97) to (104.center);
		\draw [in=-90, out=180, looseness=1.25] (106) to (98);
		\draw [in=-90, out=0, looseness=1.25] (106) to (102);
		\draw [in=90, out=0, looseness=1.25] (105) to (100);
		\draw [in=-180, out=90] (108.center) to (105);
		\draw (108.center) to (107.center);
		\draw [in=180, out=-90, looseness=1.25] (107.center) to (97);
		\draw (106) to (109.center);
		\draw (110.center) to (105);
		\draw (111.center) to (112.center);
		\draw (115.center) to (116.center);
		\draw (117.center) to (118.center);
		\draw (118.center) to (116.center);
		\draw (115.center) to (117.center);
		\draw (111.center) to (113.center);
		\draw (113.center) to (114.center);
		\draw (114.center) to (112.center);
		\draw [in=-90, out=180, looseness=1.25] (132) to (125);
		\draw [in=-90, out=0, looseness=1.25] (132) to (129);
		\draw [in=90, out=0, looseness=1.25] (131) to (127);
		\draw [in=-180, out=90] (134.center) to (131);
		\draw (134.center) to (133.center);
		\draw (132) to (135.center);
		\draw (136.center) to (131);
		\draw [in=0, out=-180, looseness=0.75] (139) to (138);
		\draw [in=90, out=0, looseness=0.75] (139) to (122);
		\draw [in=-90, out=180] (138) to (133.center);
		\draw (138) to (123);
		\draw (126) to (139);
		\draw (140.center) to (141.center);
		\draw (144.center) to (145.center);
		\draw (146.center) to (147.center);
		\draw (147.center) to (145.center);
		\draw (144.center) to (146.center);
		\draw (140.center) to (142.center);
		\draw (142.center) to (143.center);
		\draw (143.center) to (141.center);
		\draw [in=-90, out=180, looseness=1.25] (158) to (152);
		\draw [in=-90, out=0, looseness=1.25] (158) to (156);
		\draw [in=90, out=0, looseness=1.25] (157) to (154);
		\draw (158) to (161.center);
		\draw (162.center) to (157);
		\draw [in=90, out=-180, looseness=1.25] (157) to (150);
		\draw (164.center) to (165.center);
		\draw (168.center) to (169.center);
		\draw (170.center) to (171.center);
		\draw (171.center) to (169.center);
		\draw (168.center) to (170.center);
		\draw (164.center) to (166.center);
		\draw (166.center) to (167.center);
		\draw (167.center) to (165.center);
		\draw [in=-90, out=180, looseness=1.25] (180) to (175);
		\draw [in=-90, out=0, looseness=1.25] (180) to (178);
		\draw [in=90, out=0, looseness=1.25] (179) to (176);
		\draw (183.center) to (174);
		\draw [in=-180, out=90] (184.center) to (179);
		\draw (186.center) to (187.center);
		\draw (190.center) to (191.center);
		\draw (192.center) to (193.center);
		\draw (193.center) to (191.center);
		\draw (190.center) to (192.center);
		\draw (186.center) to (188.center);
		\draw (188.center) to (189.center);
		\draw (189.center) to (187.center);
		\draw (205.center) to (196);
		\draw (206.center) to (197);
		\draw (180) to (179);
	\end{pgfonlayer}
\end{tikzpicture}}
\ee
It follows that for all normalised states $\ket{\psi_1},\ket{\psi_2}, \ket{\eta}$ of $\ch$ and for all $\alpha_1, \alpha_2 \in \mathbb{C}$ such that $\alpha_1 \ket{\psi_1} + \alpha_2 \ket{\psi_2}$ is a normalised state (i.e. such that $\alpha_1^{2} + \alpha_2^{2} = 1$)
\begin{equation}
\begin{split}
\bra{\eta} \ml_{\chi} \left(\alpha_1 \ket{\psi_1} + \alpha_2 \ket{\psi_2}\right)  = & \left(\bra{\eta} \tc_{12/3456} \bra{\phi}\right)\left( \ml_{\chi_{1/234}}\left(\ket{\phi_+}\right) \tc_{1234/56} \left( \alpha_1 \ket{\psi_1} + \alpha_2 \ket{\psi_2} \right) \right) \\
 = & \alpha_1 \left(\bra{\eta} \tc_{12/3456} \bra{\phi}\right)\left( \ml_{\chi_{1/234}}\left(\ket{\phi_+}\right) \chi_{1234/56} \ket{\psi_1}\right) \\
&  + \alpha_2 \left(\bra{\eta} \tc_{12/3456} \bra{\phi}\right)\left( \ml_{\chi_{1/234}}\left(\ket{\phi_+}\right) \tc_{1234/56} \ket{\psi_2}\right) \\
 = & \alpha_1 \bra{\eta} \ml_{\chi} \left(\ket{\psi_1} \right) +  \alpha_2 \bra{\eta} \ml_{\chi} \left(\ket{\psi_2} \right) \, .
\end{split}
\end{equation} 
and thus that $\ml_{\chi}$ is the function obtained by restricting a linear map of $\cl(\ch)$ to acting only on norm-one elements. We will consider that $\ml_{\chi}$ is itself the linear map and not only its restriction to states.
Let us briefly focus on $\ml_{\nu} = \ml_{\chi_{1/}}$. The equality
\begin{equation}
(\cdot \tn (\cdot \tn \cdot)) = ((\cdot \tn \cdot) \tc_{/}\cdot) 
\end{equation}
is witnessing the comprehension $\nu \sqsubseteq \nu$ and it follows, by Equation \ref{extension}, that $\ml_{\nu}(\ket{\phi} \tn \ket{\psi}) = \ml_{\nu}(\ket{\phi}) \tn \ket{\psi}$, which, now that we know $\ml_{\nu}$ is linear, is equivalent to $\ml_{\nu}\nu^{\dagger} = \nu^{\dagger}(\ml_{\nu} \otimes \id)$. And we can compute that $\ml_{\nu} = \ml_{\nu}\nu^{\dagger}\nu = \nu^{\dagger} (\ml_{\nu} \otimes \id) \nu \in \loc_{\LL}(\nu) = \stloc_{\LL}(\nu)$ (because $\nu$ is canonical). We have thus proved that $\ml_{\nu} : \ch_{\LL} \rightarrow \ch_{\LL}$ is a linear map and moreover an element of $\stloc_{\LL}(\nu)$. We will call it $L$. Finally, because it sends (by definition) normalised states to normalised states, it is an isometry and therefore, by dimensionality, a unitary. It remains to show that $\ml_{\chi} = \ml_{\nu} \tc \id = L \tc \id$. Let us consider the following equality
\begin{equation}
(\cdot \tn (\cdot \tc \cdot)) = ((\cdot \tc \cdot) \tc_{/2}\cdot) 
\end{equation}
which is witnessing the comprehension $\chi \sqsubseteq \chi$. By Equation \ref{extension}, we have that for all states $\ket{\phi}$ and $\ket{\psi}$, $\ml_{\chi}(\ket{\phi} \tc \ket{\psi}) = \ml_{\nu}(\ket{\phi}) \tc \ket{\psi}$, which by linearity of $\ml_{\chi}$ is equivalent to $\ml_{\chi}\chi^{\dagger} = \chi^{\dagger} (\ml_{\nu} \otimes \id)$. And it follows that $\ml_{\chi}  =  \chi^{\dagger} (\ml_{\nu} \otimes \id) \chi = \ml_{\nu} \tc \id$.


Conversely, consider a family $\ml_{\chi} : S_{\ch^{\chi}} \rightarrow S_{\ch^{\chi}}$ indexed by all extensions $\chi$ of a minimal splitting map $\nu$ and such that $\ml_{\chi} = L \tc \id$ where $L$ is a unitary element of $\stloc_{\LL}(\nu)$. First let us remark that that, because $L \in \stloc_{\LL}(\nu)$, there exists $\tilde{L} \in \cons_{\LL}(\nu)$ such that $L = \tilde{L} \tn \id$ which is equal, by Equation \ref{eqminimal1} to $(\tilde{L} \tn \id) \tn \id = L \tn \id = \ml_{\nu}$ and thus that $\ml_{\nu} = L$. Then one has that for every extension $\chi$ of $\nu$, and every $\zeta$ such that $\chi \sqsubseteq \chi$, i.e. such that $\left(\left(\cdot \tm \cdot\right) \tz \cdot\right) = \left( \cdot \tc \left(\cdot \tx \cdot \right)\right)$, 
\be
\begin{split}
\ml_{\chi}(\ket{\phi} \tz \ket{\psi}) & = (L \tc \id)(\ket{\phi} \tz \ket{\psi}) \\
& = (L \tc (\id \tx \id))(\ket{\phi} \tz \ket{\psi})\\
& = ((L \tm \id) \tz \id)(\ket{\phi} \tz \ket{\psi})\\
& = (L \tm \id)\ket{\phi} \tz \ket{\psi} \\
& = \ml_{\mu}\ket{\phi} \tz \ket{\psi}
\end{split}
\ee
which proves Equation \ref{extension}. One can also compute that for every measurement result $m$,
\be
\begin{split}
\ml_{\chi}(\ket{\psi})^{\dagger} (\id \tc m ) \ml_{\chi}(\ket{\psi}) & = \bra{\psi}(L^{\dagger} \tc \id)(\id \tc m)(L \tc \id) \ket{\psi} \\
& = \bra{\psi}(L^{\dagger} \tc \id)(L \tc \id)(\id \tc m)\ket{\psi} \\
& =  \bra{\psi}(\id \tc m )\ket{\psi}
\end{split}
\ee
where the second equality come from the fact that $L \in \stloc_{\LL}(\nu) = \cons_{\LL}(\nu) = \cons_{\LL}(\chi)$ and the third comes from its unitarity. This proves Equation \ref{probas} and it remains to show that Equation \ref{update} is also true which can be done by computing that  
\be
\begin{split}
\ml_{\chi} \left( \frac{(\id \tc m)\ket{\psi}}{\sqrt{\bra{\psi}(\id \tc m )\ket{\psi}}} \right) & = \frac{(L \tc \id)(\id \tc m)\ket{\psi}}{\sqrt{\bra{\psi}(\id \tc m )\ket{\psi}}} \\ 
& =  \frac{(\id \tc m)(L \tc \id)\ket{\psi}}{\sqrt{\bra{\psi}(\id \tc m )\ket{\psi}}} \\
& =  \frac{(\id \tc m)\ml_{\chi}\ket{\psi}}{\sqrt{\bra{\psi}(\id \tc m )\ket{\psi}}} \\
& = \frac{(\id \tc m)\ml_{\chi}(\ket{\psi})}{\sqrt{\ml_{\chi}(\ket{\psi})^{\dagger} (\id \tc m ) \ml_{\chi}(\ket{\psi}) }}
\end{split}
\ee
where the second equality is true because $L \in \cons_{\LL}(\chi)$ and the last because of the previous computation.
\end{proof}

\begin{remark}
As it is explained in \cite{mestoudjian2026picturinggeneralquantumsubsystems} and as we've considered it throughout this article, $\nu$ is a splitting map representative of the algebraic non-factor subsystem $\stloc_{\LL}(\nu) \subseteq \cl(\ch)$. The admissible evolutions of the system described by $\nu$, as they're defined in Definition \ref{deflocalevolutions} and characterised in \ref{linearity} naturally translates into a family of admissible evolutions of $\stloc_{\LL}(\nu)$ of the form $A \rightarrow U^{\dagger} A U$ where $U$ is a unitary element of $\stloc_{\LL}(\nu)$. These evolutions are precisely the inner *-automorphisms of $\stloc_{\LL}(\nu)$ and these evolution can be extended to any algebra containing $\stloc_{\LL}(\nu)$ and will as expected act trivially on the complementary system of $\stloc_{\LL}(\nu)$, $\stloc_{\LL}(\nu)'$.
\end{remark}

\section{On the respective consequences of local applicability and no superluminal signalling}
In this section we explore the precise sense in which specifically no-superluminal signalling alone implies linearity. 
In order to do so, let us explore the consequences of only state locality and update commutativity, which encode the possibility to apply a transformation locally, without a prohibition on superluminal signals. 
\begin{proposition}
\label{localapplicabilityonly}
Let $\nu : \ch^{\nu} \rightarrow \ch^{\nu} \otimes \ch_{\R}^{\nu}$ be a minimal splitting map and let us consider a family $(\ml_{\chi})_{\chi}$ of functions $\ml_{\chi} : S_{\ch^{\chi}} \rightarrow S_{\ch^{\chi}}$ indexed by the extensions $\chi : \ch^{\chi} \rightarrow \ch^{\nu} \otimes \chrc$ of $\nu$ and satisfying, for every $\zeta$ such that $\chi \sqsubseteq \zeta$, i.e. such that $\left(\left(\cdot \tm \cdot\right) \tz \cdot\right) = \left( \cdot \tc \left(\cdot \tx \cdot \right)\right)$ and every $\ket{\phi}, \ket{\psi}$ such that $\ket{\phi} \tz \ket{\psi} \neq 0$,
\be
\ml_{\chi} \left( \frac{\ket{\phi} \tz \ket{\psi}}{\left\Vert \ket{\phi} \tz \ket{\psi} \right\Vert} \right) = \frac{\ml_{\mu}(\ket{\phi}) \tz \ket{\psi}}{\left\Vert \ml_{\mu}(\ket{\phi}) \tz \ket{\psi} \right\Vert } \,
\ee
and for every right $\chi$-consistent measurement result $m$,
\be
\frac{(\id \tc m)\ml_{\chi}(\ket{\psi})}{\lv (\id \tc m)\ml_{\chi}(\ket{\psi}) \rv} = \ml_{\chi} \left( \frac{(\id \tc m)\ket{\psi}}{\lv (\id \tc m)\ket{\psi} \rv} \right) \, .
\ee
Then there exists an invertible element $A$ of $\stloc_{\LL}(\nu)$ such that for every extension $\chi$ of $\nu$, 
\be
\ml_{\chi}(\ket{\psi}) = \frac{(A \tc \id)\ket{\psi}}{\left\Vert  (A \tc \id)\ket{\psi}  \right\Vert}  \, .
\ee
\end{proposition}

\begin{proof}
See Appendix \ref{proofoflocalapplicabilityonly}.
\end{proof}

Concretely, the class of maps derived are those given by invertible (and in this case block-diagonal) linear maps with renormalisation. 
This result matches an example of non-linear quantum theory with superluminal signalling given in \cite{aaronson2004quantummechanicsislandtheoryspace}, which has the notable feature as a model of computation that it allows one transformations which are invertible but approximately closed to deterministic post selection. As a precise example, whilst the projection $\pi_0 = \ket{0}\bra{0}$ is not invertible, the linear map $p_0^{\epsilon} =  \ket{0}\bra{0} + \epsilon  \ket{1}\bra{1}$ \textit{is} invertible, and applying it up to renormalisation leads to a dynamics which is state-local and update-commutative, but supports faster-than-light signals.

On the one hand, this demarcation between the axioms, gives some indication of what role exactly the prohibition on superluminal signals plays in such derivations. On the other hand, the possibility to prove that under only locality conditions, the precise class of invertible maps with renormalisation is isolated, suggests it as a particularly well-behaved and potential canonical model; of non-linear quantum mechanics, in particular one which might be arrangeable into a circuit model, and an operational probabilistic theory, albeit an interesting example of an operational theory which is non-causal (the OPT framing of no-superluminal signalling). 


\section{Conclusion}
Traditional arguments for the derivation of linearity of quantum mechanics via relativistic principles \cite{Gisin_1989} have focused on factor-systems. 
In this article, we extended a recent formalisation of these arguments \cite{wilson2023originlinearityunitarityquantum} to non-factor systems, such as those which arise naturally within causal decompositions and superpositions of geometries. We found that compatibility with relativity for evolutions of non-factor systems enforces linearity and furthermore sectorisation of evolutions.
By inspecting closer the role of each of the axioms of locally-applicable transformations, we teased apart the consequences of locality and causality, proving that locality without causality recovers a broader well-motivated class of non-linear dynamics \cite{aaronson2004quantummechanicsislandtheoryspace}.

The results of this paper, and in particular the diagrammatic nature of their proof suggest a few natural questions. 
Regarding compatibility with relativity in the sense of no-superluminal signalling, it is natural to consider whether the diagrammatic nature of the argumentation renders it applicable to the broader class of generalised operational \cite{chiribella_purification}, and categorical \cite{gogioso_cpt} probabilistic theories. The step from factor to non-factor also represents a step towards quantum gravity, where natural next steps involve adaptation of the reasoning to understand the consequences of relativistic compatibility within spin foams as opposed to traditional unitary evolutions. 

The identification of invertible linear maps with renormalisation as the most general conceivable quantum dynamics (without relativistic compatibility) also suggests a path to extension of categorical quantum mechanics \cite{coecke_kissinger_2017} to the non-linear and non-causal \cite{coecke_causalcats} setting. Future work would naturally involve formalising this identification into the development of a circuit model for a fragment of non-linear quantum computation, connection to other non-linear-with-renormalisation principles in quantum theory such as Deustch'es closed timelike curves \cite{PhysRevD.44.3197, pinzani2019categoricalsemanticstimetravel}, and analysis of the structural categorical features of such a the non-linear circuit model - such as characterisation of bases, measurements, Frobenius algebras, and higher-order processes \cite{taranto2025higherorderquantumoperations}. 

Finally, the methods of this paper rely heavily on the recently introduced conceptual and formal notion of splitting map  \cite{mestoudjian2026picturinggeneralquantumsubsystems} , it is natural to wonder for what other foundational or practical questions regarding non-factor systems in quantum theory the notion of splitting map and associated chi-tensors \cite{Arrighi2024quantumnetworks} might find application 

\section*{Acknowledgements}

OM and MW are partially funded by the European Union through the MSCA SE project QCOMICAL, by the French National Research Agency (ANR): projects TaQC ANR-22-CE47-0012 and within the framework of `Plan France 2030', under the research projects EPIQ ANR-22-PETQ-0007, OQULUS ANR-23-PETQ-0013, HQI-Acquisition ANR-22-PNCQ-0001 and HQI-R\&D
ANR-22-PNCQ-0002, and by the ID \#62312 grant from the John Templeton Foundation, as part of the \href{https://www.templeton.org/grant/the-quantum-information-structure-of-spacetime-qiss-second-phase}{‘The Quantum Information Structure of Spacetime’ Project (QISS)}. MW was funded by the Engineering and Physical Sciences Research Council [grant number EP/W524335/1]. The opinions expressed in this publication are those of the authors and do not necessarily reflect the views of the John Templeton Foundation.

\bibliographystyle{quantum}
\bibliography{biblio}

\newpage

\appendix

\section{Frobenius structure}
\label{frobenius}

\begin{proposition}
The isometry $\tikzcircle{4pt} : \ch_{\R} \rightarrow \ch_{\R} \otimes \ch_{\R}$ is a special and commutative dagger Frobenius structure (or classical structure), meaning that there exists an orthonormal basis of $\ch_{\R}$ such that for any element $\ket{j}$ of the basis, $\tikzcircle{4pt}\ket{j} = \ket{j} \otimes \ket{j}$.
\end{proposition}

\begin{proof}
We remind that because it is a canonical slitting map, $\nu$ is of the form 
\begin{equation}
\nu : \ch \cong \bigoplus_i (\ch_{\LL}^{i} \otimes \ch_{\R}^{i}) \hookrightarrow (\bigoplus_i \ch_{\LL}^{i}) \otimes (\bigoplus_i  \ch_{\R}^{i}) = \ch \otimes \ch_{\R}\\ 
\end{equation}
Let $\{ \ket{j} \}$ be an orthonormal basis of $\ch_{\R} = \bigoplus_i  \ch_{\R}^{i}$ such that each $\ket{j}$ is precisely in one of the $\ch_{\R}^{i}$. Let us fix $\ket{j} \in \ch_{\R}^{i}$ and let $\ket{\psi}$ be a (normalised) state living in the subspace $\ch_{\LL}^{i}$ of $\ch$. Let us remark that because they live in corresponding subspaces $\nu \nu^{\dagger}(\ket{\psi} \otimes \ket{j}) = \ket{\psi} \otimes \ket{j}$ and therefore that 
applying Equation \ref{eqminimal1} on any $\nu^{\dagger} (\ket{\psi} \otimes \ket{j})$ gives us 
\begin{equation}
 \begin{tikzpicture}
	\begin{pgfonlayer}{nodelayer}
		\node [style=white] (1) at (-7, 2) {$\nu$};
		\node [style=none] (2) at (-3, 4) {};
		\node [style=none] (3) at (-7, 0) {};
		\node [style=none] (4) at (-9, 4) {};
		\node [style=none] (5) at (-5, 4) {};
		\node [style=none] (6) at (-0.5, 2) {$=$};
		\node [style=black] (8) at (6, 2) {};
		\node [style=none] (9) at (2, 4) {};
		\node [style=none] (10) at (4, 4) {};
		\node [style=none] (11) at (8, 4) {};
		\node [style=none] (12) at (2, 0) {};
		\node [style=state] (13) at (-3, 0) {$j$};
		\node [style=state] (14) at (6, 0) {$j$};
		\node [style=state] (15) at (2, 0) {$\psi$};
		\node [style=state] (16) at (-7, 0) {$\psi$};
	\end{pgfonlayer}
	\begin{pgfonlayer}{edgelayer}
		\draw [in=-180, out=-90, looseness=1.25] (4.center) to (1);
		\draw [in=-90, out=0, looseness=1.25] (1) to (5.center);
		\draw [in=-90, out=-180, looseness=1.25] (8) to (10.center);
		\draw [in=-90, out=0, looseness=1.25] (8) to (11.center);
		\draw (1) to (3.center);
		\draw (2.center) to (13);
		\draw (9.center) to (12.center);
		\draw (8) to (14);
	\end{pgfonlayer}
\end{tikzpicture} \, .
\end{equation}
As this is true for any pair of matching $\ket{\psi}$ and $\ket{j}$, we deduce that there exists a family of elements $\ket{m_i}$ indexed by the subspace $\ch_{\R}^{i}$ in which $\ket{j}$ is, such that 
\be
\nu \ket{\psi} = \ket{\psi} \otimes \ket{m_i} \textrm{   and   } \tikzcircle{4pt}\ket{j} = \ket{m_i} \otimes \ket{j} \, .
\ee
Equation \ref{eqminimal2} then tells us that two different elements $\ket{j}$ have different $\ket{m_i}$ and thus belong to different subspaces $\ch_{\R}^{i}$. It follows that these subspaces are all of dimension $1$ and that $\ket{m_i} = \ket{j}$ which makes $\tikzcircle{4pt}$ the classical structure (special and commutative dagger Frobenius structure) associated to the orthonormal basis $\{ \ket{j} \}$. 

\end{proof}

\section{Proof of Proposition \ref{shapeofminimal}}

\begin{proof}
Because it is a canonical splitting map, $\chi$ is of the form
\be
\nu : \ch_{\LL} \cong \bigoplus_{i=1}^N (\ch_{\LL}^{i} \otimes \ch_{\R}^{i}) \hookrightarrow (\bigoplus_{i=1}^N  \ch_{\LL}^{i}) \otimes (\bigoplus_{i=1}^N   \ch_{\R}^{i}) = \ch_{\LL} \otimes \ch_{\R} \, .
\ee
We have seen in Appendix \ref{frobenius} that all $\ch_{\R}^{i}$ are of dimension $1$ and it immediately follows that
\begin{equation}
\nu : \ch = \bigoplus_{i=1}^N (\ch^{i} \otimes \mathbb{C}) \hookrightarrow (\bigoplus_{i=1}^N  \ch^{i}) \otimes (\bigoplus_{i=1}^N   \mathbb{C}) = \ch \otimes \mathbb{C}^{N} \, .
\end{equation}
 and that $\stloc_{\LL}(\nu) = \cons_{\LL}(\nu) = \bigoplus_i \cl(\ch_i)$. We could also have found this equality diagrammatically in the following way. Let $A \in \cons_{\LL}(\nu) \subseteq \cl(\ch_{\LL})$, we can compute that 
\begin{equation}
 \scalebox{0.90}{\begin{tikzpicture}
	\begin{pgfonlayer}{nodelayer}
		\node [style=white] (0) at (0.25, -1) {$\nu$};
		\node [style=white] (1) at (0.25, 1) {$\nu$};
		\node [style=none] (2) at (-1.75, 3.25) {};
		\node [style=none] (3) at (2.25, 3.25) {};
		\node [style=none] (4) at (-1.75, -3.25) {};
		\node [style=none] (5) at (2.25, -3.25) {};
		\node [style=none] (6) at (3.5, 0.5) {$=$};
		\node [style=none] (7) at (-3.5, 0.5) {$=$};
		\node [style=black] (15) at (-7, -2) {};
		\node [style=white] (16) at (-10.5, 2) {$\nu$};
		\node [style=none] (17) at (-10.5, 3.5) {};
		\node [style=none] (18) at (-5, 3.25) {};
		\node [style=none] (19) at (-7, -3.25) {};
		\node [style=none] (20) at (-12.5, -3.25) {};
		\node [style=white] (21) at (7, 0.25) {$\nu$};
		\node [style=white] (22) at (7, 2.25) {$\nu$};
		\node [style=none] (23) at (5, 4.5) {};
		\node [style=none] (24) at (9, 4.5) {};
		\node [style=none] (25) at (5, -2) {};
		\node [style=none] (26) at (9, -2) {};
		\node [style=none] (27) at (10, 0.5) {$=$};
		\node [style=black] (28) at (17, -0.75) {};
		\node [style=white] (29) at (13.5, 3.25) {$\nu$};
		\node [style=none] (30) at (13.5, 4.5) {};
		\node [style=none] (31) at (19, 4.5) {};
		\node [style=none] (32) at (17, -2) {};
		\node [style=none] (33) at (11.5, -2) {};
		\node [style=sqr] (34) at (-10.5, 3.5) {$A$};
		\node [style=sqr] (35) at (-1.75, 3.25) {$A$};
		\node [style=sqr] (36) at (5, -2) {$A$};
		\node [style=sqr] (37) at (11.5, -2) {$A$};
		\node [style=none] (38) at (-10.5, 4.5) {};
		\node [style=none] (39) at (-1.75, 4.5) {};
		\node [style=none] (40) at (-5, 4.5) {};
		\node [style=none] (41) at (2.25, 4.5) {};
		\node [style=none] (42) at (5, -3.25) {};
		\node [style=none] (43) at (9, -3.25) {};
		\node [style=none] (44) at (11.5, -3.25) {};
		\node [style=none] (45) at (17, -3.25) {};
	\end{pgfonlayer}
	\begin{pgfonlayer}{edgelayer}
		\draw [in=180, out=-90, looseness=1.25] (2.center) to (1);
		\draw [in=-90, out=0, looseness=1.25] (1) to (3.center);
		\draw (1) to (0);
		\draw [in=90, out=-180] (0) to (4.center);
		\draw [in=90, out=0] (0) to (5.center);
		\draw [in=-180, out=0, looseness=1.25] (16) to (15);
		\draw [in=90, out=-180, looseness=0.75] (16) to (20.center);
		\draw (16) to (17.center);
		\draw (15) to (19.center);
		\draw [in=0, out=-90, looseness=0.75] (18.center) to (15);
		\draw [in=180, out=-90, looseness=1.25] (23.center) to (22);
		\draw [in=-90, out=0, looseness=1.25] (22) to (24.center);
		\draw (22) to (21);
		\draw [in=90, out=-180] (21) to (25.center);
		\draw [in=90, out=0] (21) to (26.center);
		\draw [in=-180, out=0, looseness=1.25] (29) to (28);
		\draw [in=90, out=-180, looseness=0.75] (29) to (33.center);
		\draw (29) to (30.center);
		\draw (28) to (32.center);
		\draw [in=0, out=-90, looseness=0.75] (31.center) to (28);
		\draw (34) to (38.center);
		\draw (18.center) to (40.center);
		\draw (35) to (39.center);
		\draw (3.center) to (41.center);
		\draw (36) to (42.center);
		\draw (26.center) to (43.center);
		\draw (37) to (44.center);
		\draw (32.center) to (45.center);
	\end{pgfonlayer}
\end{tikzpicture}} 
\end{equation}
Because $\tikzcircle{4pt}$ is a classical structure associated to the basis $\ket{j}$, we have that for all $j$,
\be
 \scalebox{0.87}{\begin{tikzpicture}
	\begin{pgfonlayer}{nodelayer}
		\node [style=none] (7) at (-4, 0.75) {$=$};
		\node [style=black] (15) at (-7, -2) {};
		\node [style=white] (16) at (-10.5, 2) {$\nu$};
		\node [style=none] (17) at (-10.5, 3.5) {};
		\node [style=none] (18) at (-5, 3.25) {};
		\node [style=none] (19) at (-7, -3.25) {};
		\node [style=none] (20) at (-12.5, -3.25) {};
		\node [style=none] (27) at (5.5, 0) {$=$};
		\node [style=black] (28) at (2.5, -0.75) {};
		\node [style=white] (29) at (-1, 3.25) {$\nu$};
		\node [style=none] (30) at (-1, 4.5) {};
		\node [style=none] (31) at (4.5, 4.5) {};
		\node [style=none] (32) at (2.5, -2) {};
		\node [style=none] (33) at (-3, -2) {};
		\node [style=sqr] (34) at (-10.5, 3.5) {$A$};
		\node [style=sqr] (37) at (-3, -2) {$A$};
		\node [style=none] (38) at (-10.5, 4.5) {};
		\node [style=none] (40) at (-5, 4.5) {};
		\node [style=none] (44) at (-3, -3.25) {};
		\node [style=none] (45) at (2.5, -3.25) {};
		\node [style=white] (46) at (-17.25, 2) {$\nu$};
		\node [style=none] (47) at (-17.25, 3.5) {};
		\node [style=none] (48) at (-20.25, -3.25) {};
		\node [style=sqr] (49) at (-17.25, 3.5) {$A$};
		\node [style=none] (50) at (-17.25, 4.5) {};
		\node [style=white] (51) at (9.75, 3.25) {$\nu$};
		\node [style=none] (52) at (9.75, 4.5) {};
		\node [style=none] (53) at (6.75, -1.75) {};
		\node [style=sqr] (54) at (6.75, -1.75) {$A$};
		\node [style=none] (55) at (6.75, -3.25) {};
		\node [style=state] (56) at (-14.25, -3.25) {$j$};
		\node [style=state] (57) at (-7, -3.25) {$j$};
		\node [style=state] (58) at (2.5, -3.25) {$j$};
		\node [style=state] (59) at (12.75, -3.25) {$j$};
		\node [style=effect] (60) at (-5, 4.5) {$j$};
		\node [style=effect] (61) at (4.5, 4.5) {$j$};
		\node [style=none] (62) at (-13.5, 0.5) {$=$};
	\end{pgfonlayer}
	\begin{pgfonlayer}{edgelayer}
		\draw [in=-180, out=0, looseness=1.25] (16) to (15);
		\draw [in=90, out=-180, looseness=0.75] (16) to (20.center);
		\draw (16) to (17.center);
		\draw (15) to (19.center);
		\draw [in=0, out=-90, looseness=0.75] (18.center) to (15);
		\draw [in=-180, out=0, looseness=1.25] (29) to (28);
		\draw [in=90, out=-180, looseness=0.75] (29) to (33.center);
		\draw (29) to (30.center);
		\draw (28) to (32.center);
		\draw [in=0, out=-90, looseness=0.75] (31.center) to (28);
		\draw (34) to (38.center);
		\draw (18.center) to (40.center);
		\draw (37) to (44.center);
		\draw (32.center) to (45.center);
		\draw [in=90, out=-180] (46) to (48.center);
		\draw (46) to (47.center);
		\draw (49) to (50.center);
		\draw [in=90, out=180] (51) to (53.center);
		\draw (51) to (52.center);
		\draw (54) to (55.center);
		\draw [in=90, out=0] (46) to (56);
		\draw [in=90, out=0] (51) to (59);
	\end{pgfonlayer}
\end{tikzpicture}} 
\ee
which proves that $A \in \stloc_{\LL}(\nu)$ and thus that $ \cons_{\LL}(\nu)  \subseteq \stloc_{\LL}(\nu)$. 
Conversely, suppose that $A \in \stloc_{\LL}(\nu) \subseteq \cl(\ch_{\LL})$, then there exists $\tilde{A} \in \cons_{\LL}(\nu)$ such that $A = \nu^{\dagger} (\tilde{A} \otimes \id) \nu$ and it follows that
\begin{equation}
\scalebox{1}{ \begin{tikzpicture}
	\begin{pgfonlayer}{nodelayer}
		\node [style=white] (2) at (-3.75, 0) {};
		\node [style=white] (3) at (-1.75, -2) {$\nu$};
		\node [style=none] (4) at (-3.75, 0) {};
		\node [style=none] (5) at (0.25, 0) {};
		\node [style=none] (6) at (-1.75, -3) {};
		\node [style=none] (10) at (0.25, 5) {};
		\node [style=none] (16) at (1.5, -2) {$=$};
		\node [style=sqr] (17) at (2.25, 1) {$\tilde{A}$};
		\node [style=white] (18) at (3.5, 3.75) {$\nu$};
		\node [style=white] (20) at (4.25, -1) {$\nu$};
		\node [style=none] (21) at (2.25, 1) {};
		\node [style=none] (22) at (6.25, 1) {};
		\node [style=none] (23) at (4.25, -2) {};
		\node [style=white] (24) at (4.25, -3) {$\nu$};
		\node [style=none] (29) at (3.5, 5) {};
		\node [style=black] (31) at (6.25, 1) {};
		\node [style=none] (33) at (7.25, 2.5) {};
		\node [style=none] (34) at (7.25, 5) {};
		\node [style=none] (38) at (7.25, -9) {};
		\node [style=none] (39) at (3.5, -9) {};
		\node [style=white] (41) at (3.5, -7.75) {$\nu$};
		\node [style=white] (42) at (-1.75, -4) {$\nu$};
		\node [style=none] (43) at (0.25, -6) {};
		\node [style=none] (44) at (0.25, -9) {};
		\node [style=none] (45) at (-3.75, -9) {};
		\node [style=white] (46) at (-3.75, -6) {$\nu$};
		\node [style=white] (47) at (-3.75, -8) {$\nu$};
		\node [style=sqr] (49) at (-12.25, 1.75) {$\tilde{A}$};
		\node [style=white] (50) at (-10.75, 3.75) {$\nu$};
		\node [style=white] (51) at (-10.75, -0.25) {$\nu$};
		\node [style=white] (52) at (-8.75, -2.25) {$\nu$};
		\node [style=none] (53) at (-10.75, -0.25) {};
		\node [style=none] (54) at (-6.75, -0.25) {};
		\node [style=none] (56) at (-6.75, 5) {};
		\node [style=none] (57) at (-10.75, 5) {};
		\node [style=none] (58) at (-5.5, -2) {$=$};
		\node [style=white] (59) at (-8.75, -4.25) {$\nu$};
		\node [style=none] (60) at (-6.75, -6.25) {};
		\node [style=none] (63) at (-10.75, -6.25) {};
		\node [style=black] (64) at (6.25, -5) {};
		\node [style=none] (95) at (-9.25, 1.75) {};
		\node [style=sqr] (96) at (-5.25, 2) {$\tilde{A}$};
		\node [style=white] (97) at (-3.75, 4) {$\nu$};
		\node [style=white] (98) at (-3.75, 0) {$\nu$};
		\node [style=none] (99) at (-3.75, 0) {};
		\node [style=none] (100) at (-3.75, 5) {};
		\node [style=none] (101) at (-2.25, 2) {};
		\node [style=none] (102) at (2.25, -5) {};
		\node [style=none] (103) at (7.25, -6) {};
		\node [style=none] (133) at (-6.75, -9) {};
		\node [style=none] (134) at (-10.75, -9) {};
		\node [style=white] (138) at (-16.75, -2.25) {$\nu$};
		\node [style=none] (139) at (-18.75, -0.25) {};
		\node [style=none] (140) at (-14.75, -0.25) {};
		\node [style=none] (141) at (-14.75, 5) {};
		\node [style=white] (143) at (-16.75, -4.25) {$\nu$};
		\node [style=none] (144) at (-14.75, -6.25) {};
		\node [style=none] (145) at (-18.75, -6.25) {};
		\node [style=none] (147) at (-14.75, -9) {};
		\node [style=none] (148) at (-18.75, -9) {};
		\node [style=none] (149) at (-13, -2) {$=$};
		\node [style=none] (150) at (-18.75, 5) {};
		\node [style=sqr] (151) at (-18.75, 2.25) {$A$};
		\node [style=white] (152) at (-2.5, -18) {$\nu$};
		\node [style=none] (153) at (-0.5, -20) {};
		\node [style=none] (154) at (-2.5, -17) {};
		\node [style=none] (155) at (-0.5, -25.25) {};
		\node [style=none] (156) at (-4.5, -25.25) {};
		\node [style=white] (157) at (-2.5, -15.75) {$\nu$};
		\node [style=none] (158) at (-0.5, -14) {};
		\node [style=none] (159) at (-4.5, -14) {};
		\node [style=none] (160) at (-7, -18.25) {$=$};
		\node [style=white] (161) at (-12.25, -12.5) {$\nu$};
		\node [style=white] (162) at (-11.5, -17.25) {$\nu$};
		\node [style=none] (163) at (-13.5, -15.25) {};
		\node [style=none] (164) at (-9.5, -15.25) {};
		\node [style=none] (165) at (-11.5, -18.25) {};
		\node [style=white] (166) at (-11.5, -19.25) {$\nu$};
		\node [style=none] (167) at (-12.25, -11.25) {};
		\node [style=black] (168) at (-9.5, -15.25) {};
		\node [style=none] (169) at (-8.5, -13.75) {};
		\node [style=none] (170) at (-8.5, -11.25) {};
		\node [style=none] (171) at (-8.5, -25.25) {};
		\node [style=none] (172) at (-12.25, -25.25) {};
		\node [style=white] (173) at (-12.25, -24) {$\nu$};
		\node [style=black] (174) at (-9.5, -21.25) {};
		\node [style=none] (175) at (-13.5, -21.25) {};
		\node [style=none] (176) at (-8.5, -22.25) {};
		\node [style=none] (177) at (-15.25, -18.25) {$=$};
		\node [style=sqr] (178) at (-13.5, -21.25) {$\tilde{A}$};
		\node [style=white] (179) at (-4.5, -24) {};
		\node [style=none] (180) at (-4.5, -24) {};
		\node [style=sqr] (181) at (-6, -22) {$\tilde{A}$};
		\node [style=white] (182) at (-4.5, -20) {$\nu$};
		\node [style=white] (183) at (-4.5, -24) {$\nu$};
		\node [style=none] (184) at (-4.5, -24) {};
		\node [style=none] (185) at (-3, -22) {};
		\node [style=none] (186) at (-4.5, -11.25) {};
		\node [style=none] (187) at (-0.5, -11.25) {};
		\node [style=white] (188) at (5.25, -17.75) {$\nu$};
		\node [style=none] (189) at (3.25, -20) {};
		\node [style=none] (190) at (7.25, -20) {};
		\node [style=none] (191) at (7.25, -25.25) {};
		\node [style=white] (192) at (5.25, -15.75) {$\nu$};
		\node [style=none] (193) at (7.25, -14) {};
		\node [style=none] (194) at (3.25, -14) {};
		\node [style=none] (195) at (7.25, -11.25) {};
		\node [style=none] (196) at (3.25, -11.25) {};
		\node [style=none] (197) at (3.25, -25.25) {};
		\node [style=sqr] (198) at (3.25, -22.5) {$A$};
		\node [style=none] (199) at (1, -18.25) {$=$};
	\end{pgfonlayer}
	\begin{pgfonlayer}{edgelayer}
		\draw [in=180, out=-90, looseness=1.25] (4.center) to (3);
		\draw [in=0, out=-90, looseness=1.25] (5.center) to (3);
		\draw (3) to (6.center);
		\draw (10.center) to (5.center);
		\draw [in=-180, out=90] (17) to (18);
		\draw [in=180, out=-90, looseness=1.25] (21.center) to (20);
		\draw [in=0, out=-90, looseness=1.25] (22.center) to (20);
		\draw (20) to (23.center);
		\draw (23.center) to (24);
		\draw (29.center) to (18);
		\draw [in=-90, out=0] (31) to (33.center);
		\draw (33.center) to (34.center);
		\draw [in=180, out=0] (18) to (31);
		\draw (41) to (39.center);
		\draw [in=90, out=0, looseness=1.25] (42) to (43.center);
		\draw (43.center) to (44.center);
		\draw [in=90, out=180, looseness=1.25] (42) to (46);
		\draw [in=-180, out=180, looseness=1.25] (46) to (47);
		\draw (47) to (45.center);
		\draw [in=0, out=0, looseness=1.25] (46) to (47);
		\draw (6.center) to (42);
		\draw [in=-90, out=180, looseness=1.25] (51) to (49);
		\draw [in=180, out=90, looseness=1.25] (49) to (50);
		\draw [in=180, out=-90, looseness=1.25] (53.center) to (52);
		\draw [in=0, out=-90, looseness=1.25] (54.center) to (52);
		\draw (56.center) to (54.center);
		\draw (57.center) to (50);
		\draw [in=90, out=0, looseness=1.25] (59) to (60.center);
		\draw [in=90, out=180, looseness=1.25] (59) to (63.center);
		\draw [in=90, out=0] (24) to (64);
		\draw [in=0, out=180] (64) to (41);
		\draw (52) to (59);
		\draw [in=90, out=0, looseness=1.25] (50) to (95.center);
		\draw [in=0, out=-90, looseness=1.25] (95.center) to (53.center);
		\draw [in=-90, out=180, looseness=1.25] (98) to (96);
		\draw [in=180, out=90, looseness=1.25] (96) to (97);
		\draw (100.center) to (97);
		\draw [in=90, out=0, looseness=1.25] (97) to (101.center);
		\draw [in=0, out=-90, looseness=1.25] (101.center) to (99.center);
		\draw [in=90, out=180, looseness=1.25] (24) to (102.center);
		\draw [in=180, out=-90] (102.center) to (41);
		\draw [in=90, out=0, looseness=1.25] (64) to (103.center);
		\draw (103.center) to (38.center);
		\draw (63.center) to (134.center);
		\draw (60.center) to (133.center);
		\draw [in=180, out=-90, looseness=1.25] (139.center) to (138);
		\draw [in=0, out=-90, looseness=1.25] (140.center) to (138);
		\draw (141.center) to (140.center);
		\draw [in=90, out=0, looseness=1.25] (143) to (144.center);
		\draw [in=90, out=180, looseness=1.25] (143) to (145.center);
		\draw (138) to (143);
		\draw (145.center) to (148.center);
		\draw (144.center) to (147.center);
		\draw (150.center) to (151);
		\draw (151) to (139.center);
		\draw [in=0, out=90, looseness=1.25] (153.center) to (152);
		\draw (152) to (154.center);
		\draw (155.center) to (153.center);
		\draw [in=-90, out=0, looseness=1.25] (157) to (158.center);
		\draw (154.center) to (157);
		\draw [in=-90, out=-180, looseness=1.25] (157) to (159.center);
		\draw [in=180, out=-90, looseness=1.25] (163.center) to (162);
		\draw [in=0, out=-90, looseness=1.25] (164.center) to (162);
		\draw (162) to (165.center);
		\draw (165.center) to (166);
		\draw (167.center) to (161);
		\draw [in=-90, out=0] (168) to (169.center);
		\draw (169.center) to (170.center);
		\draw [in=180, out=0] (161) to (168);
		\draw (173) to (172.center);
		\draw [in=90, out=0] (166) to (174);
		\draw [in=0, out=180] (174) to (173);
		\draw [in=90, out=-180, looseness=1.25] (166) to (175.center);
		\draw [in=180, out=-90] (175.center) to (173);
		\draw [in=90, out=0, looseness=1.25] (174) to (176.center);
		\draw (176.center) to (171.center);
		\draw [in=90, out=-180] (161) to (163.center);
		\draw [in=-90, out=180, looseness=1.25] (183) to (181);
		\draw [in=180, out=90, looseness=1.25] (181) to (182);
		\draw [in=90, out=0, looseness=1.25] (182) to (185.center);
		\draw [in=0, out=-90, looseness=1.25] (185.center) to (184.center);
		\draw [in=90, out=-180, looseness=1.25] (152) to (182);
		\draw (184.center) to (156.center);
		\draw (186.center) to (159.center);
		\draw (187.center) to (158.center);
		\draw [in=-180, out=90, looseness=1.25] (189.center) to (188);
		\draw [in=0, out=90, looseness=1.25] (190.center) to (188);
		\draw (191.center) to (190.center);
		\draw [in=-90, out=0, looseness=1.25] (192) to (193.center);
		\draw [in=-90, out=-180, looseness=1.25] (192) to (194.center);
		\draw (188) to (192);
		\draw (194.center) to (196.center);
		\draw (193.center) to (195.center);
		\draw (197.center) to (198);
		\draw (198) to (189.center);
	\end{pgfonlayer}
\end{tikzpicture}}
\end{equation}
which proves that $\stloc_{\LL}(\nu) \subseteq \cons_{\LL}(\nu)$. 
\end{proof}

\section{Proof of Proposition \ref{shapeofextensions}}
\label{proofextensions}

\begin{proof}
Because it is a canonical splitting map, $\chi$ is of the form
\be
\chi : \ch \cong \bigoplus_{i=1}^N (\ch_{\LL}^{i} \otimes \ch_{\R}^{i}) \hookrightarrow (\bigoplus_{i=1}^N  \ch_{\LL}^{i}) \otimes (\bigoplus_{i=1}^N   \ch_{\R}^{i}) = \ch_{\LL} \otimes \ch_{\R} \, .
\ee
We can define, for every $i$, a basis $\ket{e}$ of $\ch_{\LL}^{i}$ and a basis $\ket{f}$ of $\ch_{\R}^{i}$ respectively indexed by the sets $E^{i}$ and $F^{i}$. When $i$ ranges from $1$ to $N$, the $\ket{e} \otimes \ket{f}$ with $e \in E^{i}$ and $f \in F^{i}$  form a basis of $\bigoplus_{i=1}^N (\ch_{\LL}^{i} \otimes \ch_{\R}^{i})$ which, because theyre unitarily isomorphic through the mapping made by $\chi$, gives us a basis $\{ \ket{ef} |  e \in E^{i}, \, f \in F^{i}, \, 1 \leq i \leq N\}$ of $\ch$ such that $\chi\ket{ef} = \ket{e} \otimes \ket{f} \in \chlc \otimes \chrc$. Applying Equation \ref{extension1} to any element of the basis gives us 
\be
\nu\ket{e} \otimes \ket{f} = \ket{e} \otimes \bigcirc \ket{f}
\ee
and therefore that there exists an element $\ket{m_{ef}} \in \ch_{\MM}$ for each $\ket{ef}$ such that 
\be
\nu\ket{e} = \ket{e} \otimes \ket{m_{ef}} \textrm{   and   } \bigcirc\ket{f} =  \ket{m_{ef}} \otimes \ket{f} \, .
\ee
We can immediately remark that this $\ket{m_{ef}}$, because it can be computed separately from $\ket{e}$ and from $\ket{f}$ depends neither on $e$ nor on $f$ but only on the index $i$ such that $e \in E^{i}$ and $f \in F^{i}$. We can therefore write it $\ket{m_i}$. Let us now consider $A \in \cons_{\LL}(\chi) = \bigoplus_{i=1}^{N} \cl(\ch_{\LL}^{i})$. For any $e \in E^{i}$ and $\ket{m} = \ket{m_j}$, we have that $A \ket{e} \in \ch_{\LL}^i$ and thus that 
\be
(A \otimes \id) \nu \nu^{\dagger} (\ket{e} \otimes \ket{m_j}) = \delta_{m_i m_j} A\ket{e} \otimes \ket{m_j} = \nu \nu^{\dagger} (A\ket{e} \otimes \ket{m_j}) \, .
\ee
As the $\ket{e} \otimes \ket{m_j}$ form a basis of $\ch_{\LL}^{\nu} \otimes \ch_{\R}^{\nu}$, we obtain that $(A \otimes \id) \nu \nu^{\dagger} = \nu \nu^{\dagger} (A \otimes \id)$ and thus that $A \in \cons_{\LL}(\nu)$.
Conversely, let $A \in \cons_{\LL}(\nu) $, thanks to Equation \ref{extension2}, we can compute that 
\begin{equation}
 \begin{tikzpicture}
	\begin{pgfonlayer}{nodelayer}
		\node [style=white] (0) at (1, -1) {$\chi$};
		\node [style=white] (1) at (1, 1) {$\chi$};
		\node [style=none] (2) at (-1, 3.25) {};
		\node [style=none] (3) at (3, 3.25) {};
		\node [style=none] (4) at (-1, -3.25) {};
		\node [style=none] (5) at (3, -3.25) {};
		\node [style=none] (6) at (5, 0) {$=$};
		\node [style=white] (8) at (16.5, 1.25) {$\nu$};
		\node [style=none] (11) at (14.5, 6.5) {};
		\node [style=none] (12) at (16.5, 0) {};
		\node [style=none] (13) at (22, 0) {};
		\node [style=none] (14) at (20, 7.5) {};
		\node [style=white] (16) at (16.5, -1.25) {$\nu$};
		\node [style=none] (17) at (16.5, 0) {};
		\node [style=none] (18) at (22, 0) {};
		\node [style=none] (19) at (20, -6.5) {};
		\node [style=none] (20) at (14.5, -6.5) {};
		\node [style=white] (21) at (20, 5.25) {};
		\node [style=white] (22) at (20, -5.25) {};
		\node [style=sqr] (23) at (-1, 3.25) {$A$};
		\node [style=white] (24) at (9.5, -4.25) {$\chi$};
		\node [style=white] (25) at (9.5, -2.25) {$\chi$};
		\node [style=none] (26) at (7.5, 0) {};
		\node [style=none] (27) at (11.5, 0) {};
		\node [style=none] (28) at (7.5, -6.5) {};
		\node [style=none] (29) at (11.5, -6.5) {};
		\node [style=white] (30) at (9.5, 2.25) {$\chi$};
		\node [style=white] (31) at (9.5, 4.25) {$\chi$};
		\node [style=none] (32) at (7.5, 6.5) {};
		\node [style=none] (33) at (11.5, 6.5) {};
		\node [style=none] (34) at (7.5, 0) {};
		\node [style=none] (35) at (11.5, 0) {};
		\node [style=sqr] (36) at (7.5, 6.5) {$A$};
		\node [style=none] (37) at (13.5, 0) {$=$};
		\node [style=sqr] (38) at (14.5, 6.5) {$A$};
		\node [style=white] (39) at (3.5, -13.25) {$\nu$};
		\node [style=none] (40) at (1.5, -8) {};
		\node [style=none] (41) at (3.5, -14.5) {};
		\node [style=none] (42) at (9, -14.5) {};
		\node [style=none] (43) at (7, -8) {};
		\node [style=white] (44) at (3.5, -15.75) {$\nu$};
		\node [style=none] (45) at (3.5, -14.5) {};
		\node [style=none] (46) at (9, -14.5) {};
		\node [style=none] (47) at (7, -22) {};
		\node [style=none] (48) at (1.5, -21) {};
		\node [style=white] (49) at (7, -9.25) {};
		\node [style=white] (50) at (7, -19.75) {};
		\node [style=none] (51) at (0.5, -14.5) {$=$};
		\node [style=sqr] (52) at (1.5, -21) {$A$};
		\node [style=white] (53) at (14.5, -15.5) {$\chi$};
		\node [style=white] (54) at (14.5, -13.5) {$\chi$};
		\node [style=none] (55) at (12.5, -11.25) {};
		\node [style=none] (56) at (16.5, -11.25) {};
		\node [style=none] (57) at (12.5, -17.75) {};
		\node [style=none] (58) at (16.5, -17.75) {};
		\node [style=none] (59) at (10.5, -14.5) {$=$};
		\node [style=sqr] (60) at (12.5, -17.75) {$A$};
		\node [style=none] (61) at (-1, 4.5) {};
		\node [style=none] (62) at (3, 4.5) {};
		\node [style=none] (63) at (7.5, 7.5) {};
		\node [style=none] (64) at (11.5, 7.5) {};
		\node [style=none] (65) at (14.5, 7.5) {};
		\node [style=none] (67) at (1.5, -22) {};
		\node [style=none] (68) at (12.5, -19) {};
		\node [style=none] (69) at (16.5, -19) {};
	\end{pgfonlayer}
	\begin{pgfonlayer}{edgelayer}
		\draw [in=180, out=-90, looseness=1.25] (2.center) to (1);
		\draw [in=-90, out=0, looseness=1.25] (1) to (3.center);
		\draw (1) to (0);
		\draw [in=90, out=-180] (0) to (4.center);
		\draw [in=90, out=0] (0) to (5.center);
		\draw [in=-180, out=-90, looseness=0.75] (11.center) to (8);
		\draw (8) to (12.center);
		\draw [in=90, out=-180, looseness=0.75] (16) to (20.center);
		\draw (16) to (17.center);
		\draw [in=180, out=0] (16) to (22);
		\draw [in=-90, out=0, looseness=0.75] (22) to (18.center);
		\draw (22) to (19.center);
		\draw [in=-180, out=0] (8) to (21);
		\draw (21) to (14.center);
		\draw [in=90, out=0, looseness=0.75] (21) to (13.center);
		\draw [in=180, out=-90, looseness=1.25] (26.center) to (25);
		\draw [in=-90, out=0, looseness=1.25] (25) to (27.center);
		\draw (25) to (24);
		\draw [in=90, out=-180] (24) to (28.center);
		\draw [in=90, out=0] (24) to (29.center);
		\draw [in=180, out=-90, looseness=1.25] (32.center) to (31);
		\draw [in=-90, out=0, looseness=1.25] (31) to (33.center);
		\draw (31) to (30);
		\draw [in=90, out=180, looseness=1.25] (30) to (34.center);
		\draw [in=90, out=0, looseness=1.25] (30) to (35.center);
		\draw [in=-180, out=-90, looseness=0.75] (40.center) to (39);
		\draw (39) to (41.center);
		\draw [in=90, out=-180, looseness=0.75] (44) to (48.center);
		\draw (44) to (45.center);
		\draw [in=180, out=0] (44) to (50);
		\draw [in=-90, out=0, looseness=0.75] (50) to (46.center);
		\draw (50) to (47.center);
		\draw [in=-180, out=0] (39) to (49);
		\draw (49) to (43.center);
		\draw [in=90, out=0, looseness=0.75] (49) to (42.center);
		\draw [in=180, out=-90, looseness=1.25] (55.center) to (54);
		\draw [in=-90, out=0, looseness=1.25] (54) to (56.center);
		\draw (54) to (53);
		\draw [in=90, out=-180] (53) to (57.center);
		\draw [in=90, out=0] (53) to (58.center);
		\draw (52) to (67.center);
		\draw (60) to (68.center);
		\draw (58.center) to (69.center);
		\draw (23) to (61.center);
		\draw (3.center) to (62.center);
		\draw (36) to (63.center);
		\draw (33.center) to (64.center);
		\draw (38) to (65.center);
	\end{pgfonlayer}
\end{tikzpicture} 
\end{equation}
proving that $A \in \cons_{\LL}(\chi)$ and therefore that $\cons_{\LL}(\nu) \subseteq  \cons_{\LL}(\nu)$. This proves that if $\chi$ is an extension of $\nu$, then $\cons_{\LL}(\chi) = \cons_{\LL}(\nu)$. 

Conversely, suppose that $\nu : \ch_{\LL} \rightarrow \ch_{\LL} \otimes \ch_{\MM}$ is a minimal splitting map and $\chi : \ch \rightarrow \ch_{\LL} \otimes \ch_{\R}$ a canonical splitting map such that $\cons_{\LL}(\chi) = \cons_{\LL}(\nu)$. By definition, we know that $\nu$ and $\chi$ are respectively of the form 
\begin{equation}
\nu : \ch_{\LL} = \bigoplus_{i=1}^N (\ch_{\LL}^{i} \otimes \mathbb{C}) \hookrightarrow (\bigoplus_{i=1}^N  \ch_{\LL}^{i}) \otimes (\bigoplus_{i=1}^N   \mathbb{C}) = \ch_{\LL} \otimes \mathbb{C}^{N} \, , 
\end{equation}
\begin{equation}
\chi : \ch \cong \bigoplus_{j=1}^M (\tilde{\ch}_{\LL}^{j} \otimes \tilde{\ch}_{\R}^{j}) \hookrightarrow (\bigoplus_{j=1}^M  \tilde{\ch}_{\LL}^{j}) \otimes (\bigoplus_{j=1}^M   \tilde{\ch}_{\R}^{j}) = \ch_{\LL} \otimes \ch_{\R} \, .
\end{equation}
Their respective algebras of left-consistent operators are $\cons_{\LL}(\nu) = \bigoplus_{i = 1}^{N} \cl(\ch_{\LL}^{i})$ and $\cons_{\LL}(\chi) = \bigoplus_{j = 1}^{M} \cl(\tilde{\ch}_{\LL}^{i})$. The two algebras being equal implies that $N = M$ and that the $\ch_{\LL}^{i}$ and $\tilde{\ch_{\LL}^{j}}$ can be identified, yielding that
\begin{equation}
\chi : \ch \cong \bigoplus_{i=1}^N (\ch_{\LL}^{i} \otimes \ch_{\R}^{i}) \hookrightarrow (\bigoplus_{i=1}^N  \ch_{\LL}^{i}) \otimes (\bigoplus_{i=1}^N   \ch_{\R}^{i}) = \ch_{\LL} \otimes \ch_{\R} \, .
\end{equation}
Choosing
\be
\begin{split}
\bigcirc : & \chrc \rightarrow \mathbb{C}^{N} \otimes \chrc \\
& \ket{\psi} \in \ch_{\R}^{i} \rightarrow \ket{i} \otimes \ket{\psi} \, ,
\end{split}
\ee
one can check that Equations \ref{extension1} and \ref{extension2} are satisfied, which proves that $\chi$ is an extension of $\nu$.
\end{proof}

\section{Proof of Proposition \ref{localapplicabilityonly}}
\label{proofoflocalapplicabilityonly}

The proof of Proposition \ref{localapplicabilityonly} follows the same reasoning as the one of Theorem \ref{linearity} but the absence of Equation \ref{probas} makes it impossible to prove and thus to use Proposition \ref{propositionnorms}.

\begin{proof}
Let $\nu$ be a canonical splitting map and let us prove that a local evolution associated to $\nu$ is a family of linear maps of the form $\cl_{\chi} = L \tc \id$. Let $\chi$ be a splitting map extending $\nu$ and let us prove that $\cl_{\chi}$ is linear. We remind that $\nu$ and $\chi$ are of the form 
\begin{equation}
\nu : \ch_{\LL} = \bigoplus_{i=1}^N (\ch_{\LL}^{i} \otimes \mathbb{C}) \hookrightarrow (\bigoplus_{i=1}^N  \ch_{\LL}^{i}) \otimes (\bigoplus_{i=1}^N   \mathbb{C}) = \ch_{\LL} \otimes \mathbb{C}^{N} \, , 
\end{equation}
\begin{equation}
\chi : \ch = \bigoplus_{i=1}^N (\ch_{\LL}^{i} \otimes \ch_{\R}^{i}) \hookrightarrow (\bigoplus_{i=1}^N  \ch_{\LL}^{i}) \otimes (\bigoplus_{i=1}^N   \ch_{\R}^{i}) = \ch_{\LL} \otimes \ch_{\R} \, .
\end{equation}
Let us now define the family of splitting maps $\chi_{k \cdots l / l+1 \cdots m}$, for $k \leq l \leq m \in \mathbb{Z}^{*}$, by 
\begin{equation}
\chi_{k \cdots l / l+1 \cdots m} : \bigoplus_{i=1}^N \left( \ch_{k}^{i} \otimes \cdots \otimes \ch_{m}^{i} \right) \hookrightarrow \left( \bigoplus_{i=1}^N \ch_{k}^{i} \otimes \cdots \otimes \ch_{l}^{i} \right) \otimes \left( \bigoplus_{i=1}^N \ch_{l+1}^{i} \otimes \cdots \otimes \ch_{m}^{i} \right)
\end{equation}
where $\ch_{j}^{i} = \ch_{\LL}^{i}$ when $j$ is odd and $\ch_{j}^{i} = \ch_{\R}^{i}$ when $j$ is even. We will write $\chi_{/l \cdots m}$ when $k=l$ and $\chi_{k \cdots l/ }$ when $l=m$. The $\chi_{k \cdots l / l+1 \cdots m}$ form a family of  canonical splitting maps such that :
\begin{itemize}
\item $\chi_{1/} = \nu$,
\item $\chi_{1/2} = \chi$,
\item $\chi_{1/2 \cdots l}$ is an extension of $\nu$ for all $l \geq 1$,
\item $\chi_{k \cdots l/l+1 \cdots m} \sqsubseteq \chi_{k \cdots l'/l' +1 \cdots m}$ if and only if $k \leq l \leq l' \leq m$ and this comprehension can be witnessed through the isometries $\mu = \chi_{k \cdots l /l+1 \cdots l'}$ and $\xi = \chi_{l \cdots l' / l'+1 \cdots m}$.
\end{itemize}
Let us call for every $1 \leq i \leq N$, $\ch^{i} = \ch_{\LL}^{i} \otimes \ch_{\R}^{i}$ and $d_i = \dim\left(\ch^{i}\right)$. Let $\{ \ket{j} \}_{1 \leq j \leq d_i}$ be an orthonormal basis of $\ch^{i}$ and define $\ket{\phi_{i}} = \sum_{j=1}^{d_i} \frac{1}{\sqrt{d_i}} \ket{jj}$, $\ket{\phi_{+}} = \sum_{i = 1}^{N} \frac{1}{\sqrt{N}} \ket{\phi_{i}}$ as well as $\ket{\phi} = \sum_{i = 1}^{N} d_i\sqrt{N}\ket{\phi_{i}}$ which is an unormalised state. Let 
\be
\ket{\psi} = \sum_{k=1}^N \sum_{j=1}^{d_k} \alpha_j \ket{j} \in \ch = \bigoplus_i \ch_{\LL}^{i} \otimes \ch_{\R}^{i}\, , 
\ee
be a state of norm one and let us denote 
\be
\beta_{\psi} = \left\Vert \ket{\phi_+} \tc_{1234/56} \ket{\psi} \right\Vert \, ,
\ee
and
\be
\gamma_{\psi} = \left\Vert \left (\id \tc_{12/3456 }\sum_{k = 1}^N \ketbra{\phi_k}{\phi_k} \right) \frac{\ket{\phi_+} \tc_{1234/56} \ket{\psi}}{\left\Vert \ket{\phi_+} \tc_{1234/56} \ket{\psi} \right\Vert}  \right\Vert   \, .
\ee
Note that whichever $\ket{\psi}$ has been chosen, $\beta_{\psi}$ and $\gamma_{\psi}$ are non-zero and that by Equations \ref{extension} and \ref{update} the real positive numbers 
\be
\upsilon_{\psi} = \left\Vert  \ml_{\chi_{1/234}}\left(\ket{\phi_+}\right) \tc_{1234/56} \ket{\psi} \right\Vert \, 
\ee
\be
\eta_{\psi} = \left\Vert \left( \id \tc_{12/3456 }\sum_{k = 1}^N \ketbra{\phi_k}{\phi_k} \right) \ml_{\chi_{1/23456}} \left( \frac{\ket{\phi_+} \tc_{1234/56} \ket{\psi}}{\beta_{\psi}} \right) \right\Vert
\ee
and 
\be
\tau_{\psi} = \left\Vert \ml_{\chi_{1/2}}(\ket{\psi}) \tc_{12/3456} \left(\frac{1}{\lambda} \sum_{k = 1}^N \frac{1}{d_k\sqrt{N}} \ket{\phi_k}\right) \right\Vert \, ,
\ee
where
\be
\lambda = \left\Vert \sum_{k = 1}^N \frac{1}{d_k\sqrt{N}} \ket{\phi_k} \right\Vert = \sqrt{\sum_{k=1}^{N}\frac{1}{Nd_k^{2}}} \, ,
\ee
are also non-zero. We can then compute that 
\be
\begin{split}
& \frac{1}{\eta_{\psi}\upsilon_{\psi}} \left( \id \tc_{12/3456 }\sum_{k = 1}^N \ketbra{\phi_k}{\phi_k} \right) \ml_{\chi_{1/234}}\left(\ket{\phi_+}\right) \tc_{1234/56} \ket{\psi}\\
= & \frac{1}{\eta_{\psi}} \left( \id \tc_{12/3456 }\sum_{k = 1}^N \ketbra{\phi_k}{\phi_k} \right) \ml_{\chi_{1/23456}} \left( \frac{\ket{\phi_+} \tc_{1234/56} \ket{\psi}}{\beta_{\psi}} \right) \\
= & \ml_{\chi_{1/23456}}  \left( \frac{1}{\beta_{\psi} \gamma_{\psi}} \left( \id \tc_{12/3456 }\sum_{k = 1}^N \ketbra{\phi_k}{\phi_k} \right) \ket{\phi_+} \tc_{1234/56} \ket{\psi}\right) \\
= & \ml_{\chi_{1/23456}}\left( \frac{1}{\beta_{\psi} \gamma_{\psi}} \left( \id \tc_{12/3456 }\sum_{k = 1}^N \ketbra{\phi_k}{\phi_k} \right) \sum_{i = 1}^N \sum_{j = 1}^{d_i} \frac{1}{\sqrt{Nd_i}} \left(\ket{j} \tc_{12/34} \ket{j}\right) \tc_{1234/56} \ket{\psi}\right) \\
= & \ml_{\chi_{1/23456}}\left(\frac{1}{\beta_{\psi} \gamma_{\psi}}\left(\id \tc_{12/3456} \sum_{k = 1}^N \ketbra{\phi_k}{\phi_k}\right) \sum_{i = 1}^N \sum_{j = 1}^{d_i} \frac{1}{\sqrt{Nd_i}} \ket{j} \tc_{12/3456} \left(\ket{j} \tc_{34/56} \ket{\psi}\right)\right) \\
= & \ml_{\chi_{1/23456}}\left(\frac{1}{\beta_{\psi} \gamma_{\psi}} \sum_{i = 1}^N \sum_{j = 1}^{d_i} \frac{1}{\sqrt{Nd_i}} \ket{j} \tc_{12/3456} \left(\sum_{k = 1}^N \ketbra{\phi_k}{\phi_k} \left(\ket{j} \tc_{34/56} \ket{\psi}\right)\right)\right)\\
= & \ml_{\chi_{1/23456}}\left(\frac{1}{\beta_{\psi} \gamma_{\psi}} \sum_{i = 1}^N \sum_{j = 1}^{d_i} \frac{1}{\sqrt{Nd_i}} \ket{j} \tc_{12/3456} \left(\sum_{k = 1}^N  \ket{\phi_k} \sum_{l = 1}^{d_k} \frac{1}{d_k}\left(\bra{l} \tc_{34/56} \bra{l}\right) \left(\ket{j} \tc_{34/56} \ket{\psi}\right)\right)\right)\\
= & \ml_{\chi_{1/23456}}\left(\frac{1}{\beta_{\psi} \gamma_{\psi}} \sum_{i = 1}^N \sum_{j = 1}^{d_i} \frac{1}{\sqrt{Nd_i}} \ket{j} \tc_{12/3456} \left(\sum_{k = 1}^N \frac{\alpha_j}{d_k} \ket{\phi_k}\right)\right)\\
= & \ml_{\chi_{1/23456}}\left( \frac{1}{\beta_{\psi} \gamma_{\psi}} \sum_{i = 1}^N \sum_{j = 1}^{d_i} \alpha_j \ket{j} \tc_{12/3456} \left(\sum_{k = 1}^N \frac{1}{d_k\sqrt{N}} \ket{\phi_k}\right)\right)\\
= & \ml_{\chi_{1/23456}}\left(\frac{1}{\beta_{\psi} \gamma_{\psi}} \ket{\psi} \tc_{12/3456} \left(\sum_{k = 1}^N \frac{1}{d_k\sqrt{N}} \ket{\phi_k}\right)\right)\\ 
= &  \ml_{\chi_{1/23456}}\left( \frac{\lambda}{\beta_{\psi} \gamma_{\psi}} \ket{\psi} \tc_{12/3456} \left(\frac{1}{\lambda} \sum_{k = 1}^N \frac{1}{d_k\sqrt{N}} \ket{\phi_k}\right) \right) \\
= & \frac{1}{\tau_{\psi}} \ml_{\chi_{1/2}}(\ket{\psi}) \tc_{12/3456} \left(\frac{1}{\lambda} \sum_{k = 1}^N \frac{1}{d_k\sqrt{N}} \ket{\phi_k}\right)
\end{split}
\ee
The first equality comes from applying Equation \ref{extension} and the second from applying Equation \ref{update} to the right $\chi_{1/23456}$-local measurement 
\be
\id \tc_{12/3456} m = \id \tc_{12/3456} \sum_{k=1}^N \ketbra{\phi_k}{\phi_k} = \id \tc_{1/23456} \sum_{k=1}^N (\id \tc_{2/3456} \ketbra{\phi_k}{\phi_k}) \, .
\ee
The next equalities are simple computations allowing us to rewrite the state on which $\ml_{\chi_{1/23456}}$ is applied as a different $\chi$-product. And the last equality comes from applying Equation \ref{extension} to the normalised state 
\be
\begin{split}
& \frac{1}{\beta_{\psi} \gamma_{\psi}} \ket{\psi} \tc_{12/3456} \left(\sum_{k = 1}^N \frac{1}{d_k\sqrt{N}} \ket{\phi_k}\right)\\
= & \frac{\lambda}{\beta_{\psi} \gamma_{\psi}} \ket{\psi} \tc_{12/3456} \left(\frac{1}{\lambda} \sum_{k = 1}^N \frac{1}{d_k\sqrt{N}} \ket{\phi_k}\right) \, .
\end{split}
\ee
It follows that for every (normalised) state $\ket{\psi} = \sum_{k=1}^{N} \ket{\psi_k} \in \ch$, there exists a non-zero real positive number $\kappa_{\psi} = \frac{\tau_{\psi}}{\eta_{\psi} \upsilon_{\psi}}$ such that 
\be
\begin{split}
& \kappa_{\psi} \left( \id \tc_{12/3456 }\sum_{k = 1}^N \ketbra{\phi_k}{\phi_k} \right) \ml_{\chi_{1/234}}\left(\ket{\phi_+}\right) \tc_{1234/56} \ket{\psi} \\
= & \ml_{\chi_{1/2}}\left(\ket{\psi}\right) \tc_{12/3456} \left(\sum_{k = 1}^N \frac{1}{d_k\sqrt{N}} \ket{\phi_k} \right)
\end{split}
\ee
which allows us to compute that for any state $\ket{\eta} = \sum_{k=1}^{N} \ket{\eta_k} \in \ch$,
\be
\begin{split}
&  \kappa_{\psi} \left(\bra{\eta} \tc_{12/3456} \bra{\phi}\right) \ml_{\chi_{1/234}}\left(\ket{\phi_+}\right) \tc_{1234/56} \ket{\psi} \\
= & \kappa_{\psi} \left(\bra{\eta} \tc_{12/3456} \left( \bra{\phi} \sum_{k=1}^N \ketbra{\phi_k}{\phi_k} \right) \right) \ml_{\chi_{1/234}}\left(\ket{\phi_+}\right) \tc_{1234/56} \ket{\psi} \\
= & \kappa_{\psi} \left(\bra{\eta} \tc_{12/3456} \bra{\phi}\right)\left( \id \tc_{12/3456 }\sum_{k = 1}^N \ketbra{\phi_k}{\phi_k} \right) \ml_{\chi_{1/234}}\left(\ket{\phi_+}\right) \tc_{1234/56} \ket{\psi} \\
= &  \left(\bra{\eta} \tc_{12/3456} \bra{\phi}\right) \ml_{\chi_{1/2}}\left(\ket{\psi}\right) \tc_{12/3456} \left(\sum_{k = 1}^N \frac{1}{d_k\sqrt{N}} \ket{\phi_k} \right) \\
= & \left( \sum_{k=1}^{N} \bra{\eta_k} \otimes d_k\sqrt{N}\bra{\phi_k} \right) \left( \sum_{k=1}^{N}\ml_{\chi_{1/2}}\left(\ket{\psi}\right)_k \otimes \frac{1}{d_k\sqrt{N}} \ket{\phi_k} \right) \\
= & \sum_{k=1}^{N} \bra{\eta_k}\ml_{\chi_{1/2}}\left(\ket{\psi}\right)_k \braket{\phi_k}{\phi_k} \\
= & \bra{\eta} \ml_{\chi_{1/2}}\left(\ket{\psi}\right) \, .
\end{split}
\ee
Denoting 
\be
\tilde{\ml_{\chi}} : \ket{\psi} \rightarrow \frac{1}{\kappa_{\psi}}\ml_{\chi}(\ket{\psi}) 
\ee
 we obtain  for all normalised states $\ket{\psi_1},\ket{\psi_2}, \ket{\eta}$ of $\ch$ and for all $\alpha_1, \alpha_2 \in \mathbb{C}$ such that $\alpha_1 \ket{\psi_1} + \alpha_2 \ket{\psi_2}$ is a normalised state (i.e. such that $\alpha_1^{2} + \alpha_2^{2} = 1$)
\begin{equation}
\begin{split}
\bra{\eta} \tilde{\ml_{\chi}} \left(\alpha_1 \ket{\psi_1} + \alpha_2 \ket{\psi_2}\right)  = & \left(\bra{\eta} \tc_{12/3456} \bra{\phi}\right)\left( \ml_{\chi_{1/234}}\left(\ket{\phi_+}\right) \tc_{1234/56} \left( \alpha_1 \ket{\psi_1} + \alpha_2 \ket{\psi_2} \right) \right) \\
 = & \alpha_1 \left(\bra{\eta} \tc_{12/3456} \bra{\phi}\right)\left( \ml_{\chi_{1/234}}\left(\ket{\phi_+}\right) \chi_{1234/56} \ket{\psi_1}\right) \\
&  + \alpha_2 \left(\bra{\eta} \tc_{12/3456} \bra{\phi}\right)\left( \ml_{\chi_{1/234}}\left(\ket{\phi_+}\right) \tc_{1234/56} \ket{\psi_2}\right) \\
 = & \alpha_1 \bra{\eta} \tilde{\ml_{\chi}} \left(\ket{\psi_1} \right) +  \alpha_2 \bra{\eta} \tilde{\ml_{\chi}} \left(\ket{\psi_2} \right) \, .
\end{split}
\end{equation} 
and thus that $\tilde{\ml_{\chi}}$ is the function obtained by restricting an element $A_{\chi}$ of $\cl(\ch)$ to acting only on norm-one elements. We will consider that $\tilde{\ml_{\chi}}$ is itself the linear map and not only its restriction to states and thus that $\tilde{\ml_{\chi}} = A_{\chi}$. Also identifying $\ml_{\chi}$ with its action on the whole Hilbert space, we obtain that 
\be
\ml_{\chi} : \ket{\psi} \rightarrow \frac{1}{\kappa_{\psi}} A_{\chi} \ket{\psi} \, .
\ee
Moreover, the fact that $\ml_{\chi}$ sends normalised states to normalised states forces $A_{\chi}$ to be invertible and the real and positive number $\kappa_{\psi}$ to be equal to $\kappa_{\psi} = \left\Vert A_{\chi} \ket{\psi} \right\Vert$. 

Note that for any extension $\chi$ of $\nu$, the equality
\begin{equation}
(\cdot \tn (\cdot \tc \cdot)) = ((\cdot \tc \cdot) \tc_{/2}\cdot) 
\end{equation}
is witnessing the comprehension $\chi \sqsubseteq \chi$ and it follows, by Equation \ref{extension}, that for any $\ket{\psi}$ and $\ket{\phi}$ such that $\ket{\psi} \tc \ket{\phi} \neq 0$,
\be
\frac{A_{\chi}(\ket{\psi} \tc \ket{\phi})}{\left\Vert A_{\chi}(\ket{\psi} \tc \ket{\phi}) \right\Vert} = \frac{A_{\nu}\ket{\psi} \tc \ket{\phi}}{\lv A_{\nu}\ket{\psi} \tc \ket{\phi} \rv} \, .
\ee
It follows that for all such pair of states $\ket{\psi}$ and $\ket{\phi}$, there exists a positive real number $\alpha_{\psi \phi}$ such that 
\be
A_{\chi}(\ket{\psi} \tc \ket{\phi}) = \alpha_{\psi \phi}A_{\nu}\ket{\psi} \tc \ket{\phi} \, .
\ee
The linearity of $A_{\chi}$, $A_{\nu}$ and the bilinearity of $\tc$ forces all the $\alpha_{\psi \phi}$ to be equal and we can therefore write that 
\be 
A_{\chi}(\ket{\psi} \tc \ket{\phi}) = \alpha A_{\nu}\ket{\psi} \tc \ket{\phi} 
\ee
with $\alpha$ a positive real constant, which is equivalent to 
\be
A_{\chi} \chi^{\dagger} = \chi^{\dagger}(\alpha A_{\nu} \otimes \id) \, ,
\ee
from which we can deduce that 
\be
A_{\chi} = A_{\chi}\chi^{\dagger}\chi = \chi^{\dagger} (\alpha A_{\nu} \otimes \id) \chi = \alpha A_{\nu}\tc \id \, . 
\ee
In the case where $\chi = \nu$, this tells us that $A_{\nu} = \alpha A_{\nu} \tn \id$ and is thus left $\nu$-local and therefore strictly left $\nu$ local, because $\nu$ is canonical. Finally we can check that for all $\ket{\psi}$,
\be
\ml_{\chi}(\ket{\psi}) = \frac{A_{\chi}\ket{\psi}}{\lv A_{\chi}\ket{\psi} \rv} = \frac{(\alpha A_{\nu}\tc \id) \ket{\psi}}{\lv (\alpha A_{\nu}\tc \id)\ket{\psi} \rv} = \frac{\alpha ( A_{\nu}\tc \id) \ket{\psi}}{\alpha \lv (A_{\nu}\tc \id)\ket{\psi} \rv} = \frac{ ( A_{\nu}\tc \id) \ket{\psi}}{ \lv (A_{\nu}\tc \id)\ket{\psi} \rv} 
\ee
which concludes the proof.

\end{proof}

\end{document}